\documentclass{daj}

\usepackage[utf8]{inputenc}
\usepackage{amsmath, amsthm, amsfonts, amssymb}
\usepackage{tikz}
\usepackage{thm-restate}
\usepackage[retainorgcmds]{IEEEtrantools}
\usepackage{enumerate}

\usepackage{bbm}
\newcommand{\1}{\mathbbm{1}}

\theoremstyle{plain}
\newtheorem{theorem}{Theorem}[section]
\newtheorem{lemma}[theorem]{Lemma}
\newtheorem{corollary}[theorem]{Corollary}
\newtheorem{claim}[theorem]{Claim}

\theoremstyle{definition}
\newtheorem{definition}[theorem]{Definition}
\newtheorem{example}[theorem]{Example}
\newtheorem{remark}[theorem]{Remark}

\DeclareMathOperator*{\EE}{E}
\DeclareMathOperator*{\supp}{supp}
\DeclareMathOperator*{\Var}{Var}
\DeclareMathOperator*{\Cov}{Cov}
\DeclareMathOperator{\Id}{Id}
\DeclareMathOperator{\Inf}{Inf}
\newcommand{\cE}{\mathcal{E}}
\newcommand{\cM}{\mathcal{M}}
\newcommand{\cP}{\mathcal{P}}
\newcommand{\cR}{\mathcal{R}}
\newcommand{\bbN}{\mathbb{N}}
\newcommand{\bbR}{\mathbb{R}}

\definecolor{DSgray}{cmyk}{0,0,0,0.7}
\definecolor{DSred}{cmyk}{0,0.7,0,0.7}

\dajAUTHORdetails{%
  title = {Product Space Models of Correlation:
    Between Noise Stability and Additive Combinatorics},
  author = {Jan H\k{a}z\l{}a, Thomas Holenstein, and Elchanan Mossel},
  plaintextauthor = {Jan Hazla, Thomas Holenstein, and Elchanan Mossel},
  keywords = {correlated product spaces, invariance principle, noise stability},
}   

\dajEDITORdetails{%
   year={2018},
   number={20},
   received={20 June 2017},   
   published={27 December 2018},  
   doi={10.19086/da.6513},       
}   

\begin{document}

\begin{frontmatter}[classification=text]


  \author[jh]{Jan H\k{a}z\l{}a\thanks{
      During his studies at ETH Zurich
      J.\,H.~was supported by the Swiss National Science Foundation (SNF),
      project no.~200021-132508.
    }}
\author[th]{Thomas Holenstein}
\author[em]{Elchanan Mossel\thanks{
    E.\,M. was supported by NSF grant DMS-1106999, NSF grant CCF 1320105 and DOD ONR grant N000141110140 and grant 328025 from the Simons foundation.
    Part of this work was done while T.\,H.~and E.\,M.~were at the 
    Simons Institute.
  }}

\begin{abstract}
  There is a common theme to some research questions in
  additive combinatorics and noise stability.
  Both study the following basic question: 
  Let $\mathcal{P}$ be a probability distribution
  over a space $\Omega^\ell$ with all $\ell$ marginals equal. Let
  $\underline{X}^{(1)}, \ldots, \underline{X}^{(\ell)}, \allowbreak
  \underline{X}^{(j)} = (X_1^{(j)}, \ldots, X_n^{(j)})$ be random vectors such
  that for every coordinate $i \in [n]$ the tuples $(X_i^{(1)}, \ldots, X_i^{(\ell)})$ 
  are i.i.d.~according to $\cP$.
  
  A central question that is addressed in both areas is: 
  \begin{itemize}
  \item 
    Does there exist a function $c_\cP()$ independent
    of $n$ such that
    for every $f: \Omega^n \to [0, 1]$ with 
$\EE[f(\underline{X}^{(1)})] = \mu > 0$:
\begin{align*}
  \EE \left[ \prod_{j=1}^\ell f(\underline{X}^{(j)}) \right] 
  \ge c_\cP(\mu) > 0 \; ?
\end{align*}
Instances of this question include the finite field model versions of Roth's and 
 Szemerédi's theorems as well as Borell's result about the optimality of noise stability 
of half-spaces.
\end{itemize}

Our goal in this paper is to interpolate between the noise stability theory and
the finite field additive combinatorics theory and address the question above
in greater generality than considered before. 
In particular, we settle the question for $\ell = 2$ and when $\ell > 2$ and $\cP$
has bounded correlation $\rho(\cP) < 1$. Under the same conditions we also
characterize the {\em obstructions} for similar lower bounds in the case of
$\ell$ different functions. 
Part of the novelty in our proof is the combination of analytic arguments from
the theories of influences and hyper-contraction with arguments from
additive combinatorics.
\end{abstract}
\end{frontmatter}

\section{Introduction}

\subsection{Setup and same-set hitting} 
In this paper we analyze a general framework which includes many fundamental
questions in both
the theory of noise stability and in finite field models of additive combinatorics. 
We begin with formally defining this general setting. 
Let $\Omega$ be a finite set  
and assume we are given a probability distribution $\mathcal{P}$
over $\Omega^\ell$
for some $\ell \ge 2$ -- we will call it an
\emph{$\ell$-step probability distribution over $\Omega$}.

Furthermore, assume we are given $n \in \bbN$.
We consider $\ell$ vectors $\underline{X}^{(1)}, \allowbreak \ldots, 
\allowbreak \underline{X}^{(\ell)}$,
$\underline{X}^{(j)} = (X_1^{(j)},\allowbreak \ldots, X_n^{(j)})$ such that 
for every $i \in [n]$, the $\ell$-tuple $(X_i^{(1)}, \ldots, \allowbreak  X_i^{(\ell)})$ 
is sampled according to $\mathcal{P}$, independently of the other
coordinates $i' \neq i$
(see Figure~\ref{fig:naming} for an overview of the notation). 

\begin{definition}
\label{def:same-hitting}
Let $\mu, \delta \in (0, 1]$. We say that a distribution
\emph{$\mathcal{P}$ is $(\mu, \delta)$-same-set hitting},
if, for all $n \geq 1$, whenever a function $f: \Omega^n \to [0, 1]$ satisfies
$\EE[f(\underline{X}^{(j)})] \ge \mu$
for every $j \in [\ell] := \{1, \ldots, \ell\}$,
we have
\begin{align*}
\EE \left[ \prod_{j=1}^\ell f(\underline{X}^{(j)})\right]
  \ge \delta \; .
\end{align*}

We call $\mathcal{P}$ \emph{same-set hitting} if for every 
$\mu \in (0, 1]$ there exists $\delta \in (0, 1]$ such that
$\mathcal{P}$ is $(\mu, \delta)$-same-set hitting.
\end{definition}

It is not difficult to see that the definition
of same-set hitting is equivalent to the one where functions $f$
are restricted to be set indicators $f: \Omega^n \to \{0,1\}$.
The value $\EE\left[\prod_{j=1}^\ell f\left(\underline{X}^{(j)}\right)\right]$
then can be interpreted as 
$\Pr\left[\bigwedge_{j=1}^\ell \underline{X}^{(j)} \in S\right]$
for the respective set $S := \{\underline{x}: f(\underline{x})=1\}$ of density
at least $\mu$. This special case motivated the name ``same-set hitting'',
and all our theorems and proofs can be read with that case
in mind.

In this paper we address the question: which distributions $\mathcal{P}$
are same-set hitting? We achieve full characterization for $\ell = 2$
and answer the question affirmatively for a large class of distributions
with $\ell > 2$.

The question of set hitting was studied extensively in additive combinatorics and in the theory of influences and noise stability. Perhaps the most well-studied case is that of random arithmetic progressions. Let $Z$ be a finite additive group
and $\ell \in \mathbb{N}$.
Then, we can define a distribution $\mathcal{P}_{Z,\ell}$ 
of random $\ell$-step arithmetic progressions in $Z$. 
Specifically, for every $x, r \in Z$ we set:
\begin{align*}
  \cP_{Z,\ell}(x, x+r, x+2r, \ldots, x+(\ell-1)r) := 1/|Z|^2 \; .
\end{align*}

Some of the distributions $\cP_{Z,\ell}$ can be shown to be
same-set hitting using, e.g., the hypergraph regularity lemma:
\begin{theorem}[\cite{RS04}, \cite{RS06}, \cite{Gow07}, cf.~Theorem 11.27,
Proposition 11.28 and Exercise 11.6.3 in \cite{TV06}]
\label{thm:progressions}
If $|Z|$ is coprime to $(\ell-1)!$, then $\cP_{Z,\ell}$ is same-set hitting.
\end{theorem}

Taking $\ell = p$ and $Z = \mathbb{F}_p$ we obtain the classical
formulation of Szemerédi's theorem for progressions of length $p$
in the finite field model. The special case $\ell = 3$ is also known
as the \emph{capset problem}.
As is well known, the case $\ell=3$ follows from the arguments of
Roth~\cite{Rot53} applied to the finite field setup \cite{Mes95},
while the general case
follows a long line of work, starting by Szemerédi's regularity lemma
\cite{Sze75}, its proof by Furstenberg using the ergodic theorem \cite{Fur77}
as well as the finite group and multi-dimensional versions, see, e.g.,
\cite{Rot53, FK91, Gow01, Green05}.

It is natural to consider a generalization of the question where different
functions are applied to different $X^{(j)}$. This question was studied in the
theories of Gaussian noise stability and hyper-contraction as we explain next.

\subsection{Set hitting} 
The generalization to multiple sets is defined as follows. 
\begin{definition}
\label{def:hitting}
Let $\mu, \delta \in (0, 1]$. We say that a distribution
\emph{$\mathcal{P}$ is $(\mu, \delta)$-set hitting},
if, whenever functions $f^{(1)}, \ldots, f^{(\ell)}: \Omega^n \to [0, 1]$ satisfy
$\EE[f^{(j)}(\underline{X}^{(j)})] \allowbreak \ge \mu$
for every $j \in [\ell]$,
we have
\begin{align}
\label{eq:85a}
\EE \left[ \prod_{j=1}^\ell f^{(j)}(\underline{X}^{(j)})\right]
  \ge \delta \; .
\end{align}

We call $\mathcal{P}$ \emph{set hitting} if for every 
$\mu \in (0, 1]$ there exists $\delta \in (0, 1]$ such that
$\mathcal{P}$ is $(\mu, \delta)$-set hitting.
\end{definition}

Borell~\cite{Bor85} established the set hitting property in the Gaussian case where 
$(X_i,Y_i) \sim \mathcal{N}(0,\left( \begin{matrix} 1 & \rho \\ \rho & 1 \end{matrix} \right))$ are i.i.d and 
$\rho \in (0,1)$. 
In fact~\cite{Bor85} does much more: it finds the optimal $\delta$ in terms of $\mu$ and $\rho$ in this case. 
(Note that in this case $\Omega$ is infinite). 

In earlier work,~\cite{Bor82} Borell also proved some of the first
reverse hypercontractive inequalities. 
 These give a different proof that the Gaussian example above is set hitting but also imply 
  the same for the binary analog where 
$(X_i,Y_i) \in \{-1,1\}^2$ satisfy $E[X_i] = E[Y_i] = 0$ and $E[X_i Y_i] = \rho$.
See~\cite{MOR06} for a discussion of this result and some of its implications.   

The full classification of set hitting distributions can be deduced from
a paper on reverse hypercontractivity\footnote{
That $\mathcal{P}$ is set hitting if (\ref{eq:84a}) holds is a consequence
of Lemma 8.3 in \cite{MOS13}. If (\ref{eq:84a}) does not hold,
an appropriate combination of dictators establishes a counterexample.  
}
by Mossel, Oleszkiewicz and Sen
\cite{MOS13}:
\begin{theorem}[\cite{MOS13}]
\label{thm:different-sets-classification}
A finite probability space $\mathcal{P}$ is set hitting if and only if:
\begin{align}
\label{eq:84a}
  \beta(\cP) :=
  \min_{\substack{x^{(1)} \in \supp(X_i^{(1)}),\\
  \ldots,\\
  x^{(\ell)} \in \supp(X_i^{(\ell)})}}
  \mathcal{P}(x^{(1)}, \ldots, x^{(\ell)}) > 0 \; .
\end{align}
\end{theorem}

In many interesting settings, including the finite field models in
additive combinatorics, the distribution $\mathcal{P}$ does not have full
support. In these settings, as we discuss next,
the goal is to understand sufficient conditions
on the functions which imply that (\ref{eq:85a}) does hold. 

\subsection{Obstructions in additive combinatorics} 
In general much of the interest in additive combinatorics is in understanding
what conditions on functions $f$ imply (\ref{eq:85a}).
For example, the starting point of the proof of Roth's theorem~\cite{Rot53}
on arithmetic progressions of length three is that
if the functions $f^{(1)},f^{(2)},f^{(3)}$ all satisfy that 
$\| \widehat{f^{(j)} - E[f]}  \|_{\infty}$ is small
then (\ref{eq:85a}) holds. That is, the distribution of arithmetic progressions
of length three is set hitting for all functions $f^{(j)}$ with
all (positive degree) Fourier coefficients small in absolute value.
As a matter of fact,
in that case $f^{(1)}, f^{(2)}, f^{(3)}$ are known to be pseudorandom in the sense
that $\delta(\mu) \approx \mu^3$.

The proof of Roth's theorem then proceeds roughly as follows: If a function $f$
is pseudorandom, we are done. Otherwise, we are guaranteed a large Fourier
coefficient. This is then exploited in a density increment argument:
It turns out that a large Fourier coefficient implies that $f$ must
have increased relative density on an affine subspace of $\mathbb{F}_p^n$
of codimension one.
One iterates the density increment until $f$ becomes pseudorandom.

A similar situation arises in a more recent proof
for longer arithmetic progression by Gowers:
If the functions $f^{(j)}$ have low Gowers uniformity norm, then 
 (\ref{eq:85a}) holds, see e.g.~\cite{Green05}. 
 
 In one of our main results (see Section~\ref{sec:intro-fourier} below)
 we show that in a pretty general setup (which does not include the additive
 combinatorics setup), the only obstruction for $(\ref{eq:85a})$ to hold is for
 $f^{(j)}$ to have a large low-degree Fourier coefficient. 

\subsection{Basic example}
\label{sec:basic}

At this point we would like to introduce the simplest example that is not covered by either the theory of influences or techniques from additive combinatorics. 
Let $S \subseteq \{0,1,2\}^n$ be a non-empty set of density
$\mu = \frac{|S|}{3^n}$.  We pick a random vector
$\underline{X} = (X_1, \ldots, X_n)$ uniformly from $\{0,1,2\}^n$, and then
sample another vector $\underline{Y} = (Y_1,\ldots,Y_n)$ such that for each
$i$ independently, coordinate $Y_i$ is picked uniformly in
$\{X_i, X_i+1 \bmod{3}\}$.  Our goal is to show that:
\begin{align*}
\Pr[\underline{X} \in S \land \underline{Y} \in S] \ge c(\mu) > 0 \; .
\end{align*}
In other words, we want to bound away the probability from $0$
by an expression which only depends on $\mu$ and not on $n$.
Similarly, given sets $S$ and $T$ of density at least $\mu$,
we want to find under what conditions does it hold that the
probability $\Pr[\underline{X} \in S \land \underline{Y} \in S]$ can be lower
bounded effectively. We note that the support of the distribution on
$\{0,1,2\}^2$ is not full
(hence, Theorem~\ref{thm:different-sets-classification} does not apply)
and that the distribution is not of arithmetic nature. 
\subsection{Our results}

\subsubsection{Same-set hitting for two steps}

In case of $\ell = 2$ we establish the following theorem:
\begin{theorem}[cf.~Theorem~\ref{thm:main-two-variables}]
\label{thm:two-steps-classification}
A two-step probability distribution with equal marginals $\mathcal{P}$
is same-set hitting if and only if
$\alpha(\mathcal{P}) := \min_{x \in \Omega} \cP(x, x)
> 0$.
\end{theorem}

Of course, if $\beta(\cP) > 0$, then Theorem \ref{thm:two-steps-classification}
follows from Theorem \ref{thm:different-sets-classification}. 
Our work is novel in case $\beta(\mathcal{P}) = 0$,
i.e., when the distribution is same-set hitting but not set hitting.
In particular we establish same-set hitting
for the probability space from Section \ref{sec:basic}.

\subsubsection{Same-set hitting for more than two steps}

In a general case of an $\ell$-step distribution with equal marginals, it
is still clear that, letting
$\alpha(\mathcal{P}) := \min_{x\in\Omega} \cP(x,x,\ldots,x)$, the condition
$\alpha(\mathcal{P}) > 0$ is necessary. However, it remains
open if it is sufficient.

We provide the following partial results. Firstly, by a simple inductive argument
based on Theorem \ref{thm:main-two-variables}, we show that multi-step 
probability  spaces induced by Markov chains are same-set hitting
(cf.~Section \ref{sec:markov}).

Secondly, we show that $\mathcal{P}$ is same-set hitting 
if $\alpha(\mathcal{P}) > 0$ and its 
\emph{correlation} $\rho(\mathcal{P})$ is smaller than $1$. 
The opposite condition $\rho(\mathcal{P}) = 1$ is equivalent
to the following:
There exist $j \in [\ell]$, $S \subseteq \Omega$, 
$T \subseteq \Omega^{\ell-1}$ such that
$0 < |S| < |\Omega|$ and:
\begin{align*}
X_i^{(j)} \in S \iff
\left(X_i^{(1)}, \ldots, X_i^{(j-1)}, X_i^{(j+1)}, \ldots, X_i^{(\ell)}\right)
\in T \; .
\end{align*}
For the full definition of $\rho(\mathcal{P})$, see~Definition \ref{def:correlation}.

\begin{theorem}[cf.~Theorem \ref{thm:main-multiple}]
\label{thm:rho-hitting}
Let $\mathcal{P}$ be a probability distribution with equal marginals.
If $\alpha(\mathcal{P}) > 0$ and $\rho(\mathcal{P}) < 1$, then $\mathcal{P}$
is same-set hitting.
\end{theorem}

We are not aware of any general results in case $\rho(\mathcal{P}) = 1$.
In particular, let $\mathcal{P}$ be a three-step distribution over
$\Omega = \{0, 1, 2\}$ such that $X_i^{(1)},X_i^{(2)},X_i^{(3)}$ are uniform
over $\{000,111,222,012,120,201\}$. To the best of our knowledge,
it is an open question whether this distribution $\mathcal{P}$ is same-set 
hitting.
One might conjecture that $\alpha(\cP) > 0$ is the sole sufficient
condition for same-set hitting.
Unfortunately, the techniques used to prove Theorem \ref{thm:progressions}
do not seem to extend easily to spaces with less algebraic structure.

\subsubsection{Set hitting for functions  with no large Fourier coefficients}
\label{sec:intro-fourier}

The methods developed here also allow to obtain lower bounds on the 
probability of hitting multiple sets. In fact, we show that if $\rho(\cP) < 1$, 
then such lower bounds exist in terms of $\rho$, 
the measures of the sets and the largest non-empty Fourier coefficient. 

\begin{theorem}[Informal, cf.~Theorem \ref{thm:local-variance}]
\label{thm:local-variance-basic}
Let $\cP$ be a probability distribution with $\rho(\cP) < 1$. Then, $\cP$
is set-hitting for functions $f^{(1)}, \ldots, f^{(\ell)}: 
\underline{\Omega} \to [0, 1]$ that have both:
\begin{itemize}
  \item Noticeable expectations, i.e., 
$\EE[f^{(j)}(\underline{X}^{(j)})] \ge \Omega(1)$.
  \item No large Fourier coefficients, i.e.,
$\max_{\sigma} \left| \hat{f}^{(j)}(\sigma) \right| \le o(1)$.
\end{itemize}
\end{theorem}

\subsection{Other related work}

In the case of symmetric two-step spaces
(which can be thought of as product graphs) 
works by Dinur, Friedgut and Regev \cite{DFR08, FR18}
establish a removal lemma: They show
that if $\Pr[\underline{X} \in S \land \underline{Y} \in S]$ is small,
then it must be possible to remove a small number of elements from $S$
to obtain $S'$ with $\Pr[\underline{X} \in S' \land \underline{Y} \in S'] = 0$.
They go on to use this result to characterize all sets with
$\Pr[\underline{X} \in S \land \underline{Y} \in S] = 0$: It turns out
that every such set must be almost contained in a junta. 
Interestingly, \cite{FR18} obtain
a tower-type dependence between $\mu$ and $\delta$ in the removal lemma,
in contrast to ours which is ``merely'' triply exponential.

The case of $\rho < 1$ has also been studied in the context of extremal 
combinatorics and hardness of approximation.
In particular, Mossel \cite{Mos10} uses the 
invariance principle to prove that if $\rho(\mathcal{P}) < 1$, then $\mathcal{P}$ is
set hitting for low-influence functions. We use this result to establish
Theorem \ref{thm:rho-hitting}. Additionally, Theorem 
\ref{thm:local-variance-basic} can be seen as a strengthening of
\cite{Mos10}.

Furthermore, Austrin and Mossel \cite{AM13} establish the result equivalent to 
Theorem \ref{thm:local-variance-basic} assuming in addition to $\rho(\cP) < 1$
also that $\cP$ is pairwise independent (they also prove results for the case
$\rho(\cP)=1$ with pairwise independence but these involve only bounded degree 
functions).

Our work is related to problems and results in inapproximability in theoretical
computer science. 
For example, our theorem is related to the proof of hardness
for rainbow colorings of hypergraphs by Guruswami and Lee
\cite{GL15}. In particular, it is connected to their Theorem 4.3
and partially answers their Questions C.4 and C.6.

There are works in additive combinatorics that treat 
specific classes of distributions with $\rho = 1$. For example, one
can take $\cP$ to be uniform over solutions to a fixed full-rank
system of $r$ linear equations with $\ell$ variables over $\mathbb{F}_p$.
There is extensive work on removal lemmas (which imply same-set hitting)
for different cases in this setting, see, e.g.,
\cite{Gre05a, KSV09, Sha10, FLS18}.

\paragraph{Follow-up work}
There are two subsequent preprints by some of the authors:
\cite{Mos17} strengthens Theorem~\ref{thm:local-variance} to obtain
precise Gaussian bounds for functions with small low-degree
Fourier coefficients in case $\rho(\cP) < 1$ (one can also use the technique
from \cite{Mos17} to deduce an
alternative proof of Theorem~\ref{thm:main-multiple} with roughly
the same dependence).
Another author \cite{SymProg} shows
same-set hitting for symmetric sets for the distribution of
arithmetic progressions with restricted differences mentioned
in Section~\ref{sec:open}.

\subsection{Proof ideas: additive combinatorics and theory of influences}
Interestingly, the proof of our results interpolates between additive
combinatorics and the theory of influences. Results of \cite{Mos10} imply that
if a collection of functions have low influences then they are same-set
hitting. In the proof of Theorem~\ref{thm:main-multiple} we apply a variant of
a density increment argument to reduce to this case. First, we apply the standard
density increment argument to assume without loss of generality that
conditioning on a small number of coordinates does not change the measure of
the set by much. Then we show, under this assumption, by applying another variant
of density increment that we can additionally assume
w.l.o.g.~that all influences are small. 

\subsection{Outline of the paper}

The rest of the paper is organised as follows: the notation is introduced
in Section~\ref{sec:notation}, Section~\ref{sec:results} contains full
statements of our theorems, and Section~\ref{sec:proof-sketch} sketches
the proof of our main theorem.

The full proof of the multi-step theorem follows in 
Section~\ref{sec:main-proof}. The proof of the two-step theorem is in 
Section~\ref{sec:two-steps} and the proof for functions with
small Fourier coefficients in Section~\ref{sec:local-variance}.
A theorem for Markov chains is introduced in Section~\ref{sec:markov}
and better bounds for symmetric spaces in Section~\ref{sec:polynomial-hitting}.
Finally, the modified proof of the low-influence theorem from
\cite{Mos10} is presented in the appendix.
We note that an extended abstract of our results appeared 
in~\cite{HaHoMo:16}. 

\section{Notation and Preliminaries}
\label{sec:notation}

\subsection{Notation}

We will now introduce our setting and notation.
We refer the reader to Figure~\ref{fig:naming} for an overview.

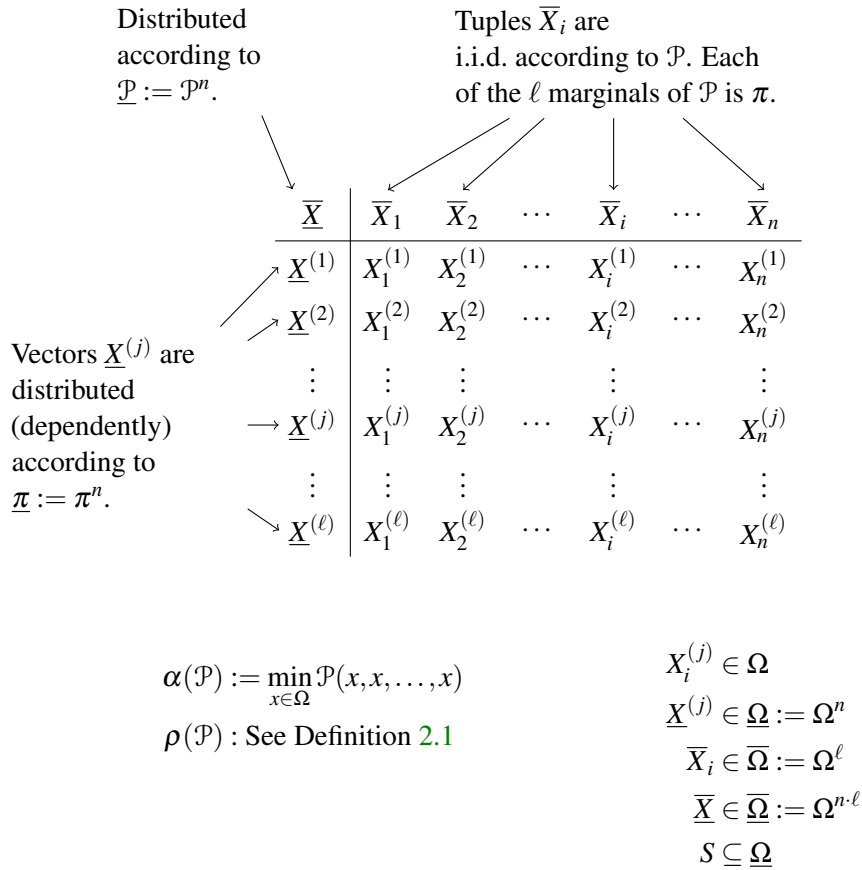
\begin{figure}
\begin{tikzpicture}[yscale=-0.7]
\node(X__) at (0,0) {$\underline{\overline{X}}$};
\node(X1_) at (1,0) {$\overline{X}_1$};
\node(X2_) at (2,0) {$\overline{X}_2$};
\node(X3_) at (3,0) {$\dots$};
\node(X4_) at (4,0) {$\overline{X}_i$};
\node(X5_) at (5,0) {$\dots$};
\node(X6_) at (6,0) {$\overline{X}_n$};

\node(X_1) at (0,1) {$\underline{X}^{(1)}$};
\node(X11) at (1,1) {$X_1^{(1)}$};
\node(X21) at (2,1) {$X_2^{(1)}$};
\node(X31) at (3,1) {$\cdots$};
\node(X41) at (4,1) {$X_i^{(1)}$};
\node(X51) at (5,1) {$\cdots$};
\node(X51) at (6,1) {$X_n^{(1)}$};

\node(X_2) at (0,2) {$\underline{X}^{(2)}$};
\node(X12) at (1,2) {$X_1^{(2)}$};
\node(X22) at (2,2) {$X_2^{(2)}$};
\node(X32) at (3,2) {$\cdots$};
\node(X42) at (4,2) {$X_i^{(2)}$};
\node(X52) at (5,2) {$\cdots$};
\node(X52) at (6,2) {$X_n^{(2)}$};

\node(X_3) at (0,3) {$\vdots$};
\node(X13) at (1,3) {$\vdots$};
\node(X23) at (2,3) {$\vdots$};
\node(X43) at (4,3) {$\vdots$};
\node(X63) at (6,3) {$\vdots$};

\node(X_4) at (0,4) {$\underline{X}^{(j)}$};
\node(X14) at (1,4) {$X_1^{(j)}$};
\node(X24) at (2,4) {$X_2^{(j)}$};
\node(X34) at (3,4) {$\cdots$};
\node(X44) at (4,4) {$X_i^{(j)}$};
\node(X54) at (5,4) {$\cdots$};
\node(X64) at (6,4) {$X_n^{(j)}$};

\node(X_5) at (0,5) {$\vdots$};
\node(X15) at (1,5) {$\vdots$};
\node(X25) at (2,5) {$\vdots$};
\node(X45) at (4,5) {$\vdots$};
\node(X65) at (6,5) {$\vdots$};

\node(X_6) at (0,6) {$\underline{X}^{(\ell)}$};
\node(X16) at (1,6) {$X_1^{(\ell)}$};
\node(X26) at (2,6) {$X_2^{(\ell)}$};
\node(X36) at (3,6) {$\cdots$};
\node(X46) at (4,6) {$X_i^{(\ell)}$};
\node(X56) at (5,6) {$\cdots$};
\node(X66) at (6,6) {$X_n^{(\ell)}$};

\draw[-] (0.5,-.5) to (0.5,6.5);
\draw[-] (-.5,0.5) to (6.5,0.5);

\node(iid) at (4,-3) {\begin{minipage}{4.2cm}\raggedright
Tuples $\overline{X}_{i}$ are 
i.i.d.~according to $\mathcal{P}$.
Each of the $\ell$ marginals of $\mathcal{P}$ is~$\pi$.\end{minipage}};

\draw[->] (iid) to (X1_.north);
\draw[->] (iid) to (X2_.north);
\draw[->] (iid) to (X4_.north);
\draw[->] (iid) to (X6_.north);

\node[rotate=0](rowdistr) at (-2.5, 4) 
{\begin{minipage}{3cm}\raggedright Vectors $\underline{X}^{(j)}$
are distributed (dependently) according to $\underline{\pi} := \pi^n$.
\end{minipage}};
\draw[->] (rowdistr) to (X_1.west);
\draw[->] (rowdistr) to (X_2.west);
\draw[->] (rowdistr) to (X_4.west);
\draw[->] (rowdistr) to (X_6.west);

\node(commonDist) at (-1, -3) 
{\begin{minipage}{3.2cm}\raggedright Distributed according to 
$\underline{\mathcal{P}}:=\mathcal{P}^n$.\end{minipage}};
\draw[->] (commonDist) to (X__.north west);

\node(propertiesofP) at (0,9) {\begin{minipage}{5cm}\begin{align*}
\alpha(\mathcal{P}) &:= \min_{x \in \Omega} \mathcal{P}(x,x,\ldots,x)\\
\rho(\mathcal{P}) &: \text{See Definition~\ref{def:correlation}}
\end{align*}\end{minipage}};

\node(domainsofxij) at (6,10) {\begin{minipage}{4cm}\begin{align*}
X_{i}^{(j)} &\in \Omega\\
\underline{X}^{(j)} &\in \underline{\Omega} := \Omega^{n}\\
\overline{X}_{i} & \in \overline{\Omega} := \Omega^{\ell}\\
\underline{\overline{X}} & \in \underline{\overline{\Omega}} 
:= \Omega^{n\cdot \ell}\\
S &\subseteq \underline{\Omega}
\end{align*}\end{minipage}};

\end{tikzpicture}
\caption{Naming of the random variables in the general case.
The columns $\overline{X}_i$ are distributed i.i.d~according to
$\mathcal{P}$. 
Each $X_{i}^{(j)}$ is distributed according to $\pi$.
The overall distribution of $\overline{\underline{X}}$ is
$\underline{\mathcal{P}}$.}
\label{fig:naming}
\end{figure}

We always assume that we have $n$ independent coordinates.
In each coordinate~$i$
we pick~$\ell$ values $X_{i}^{(j)}$ for $j \in [\ell] = \{1,\ldots,\ell\}$ 
at random using some distribution.
Each value $X_{i}^{(j)}$ is chosen from the same fixed set $\Omega$, and 
the distribution
of the tuple $\overline{X}_i = (X_{i}^{(1)},\ldots,X_{i}^{(\ell)})$ of
values from $\Omega^{\ell}$ is given by a distribution~$\mathcal{P}$.

This gives us values $X_{i}^{(j)}$ for $i \in \{1,\ldots,n\}$ and 
$j \in \{1,\ldots,\ell\}$.
Thus, we have $\ell$ vectors
$\underline{X}^{(1)},\ldots,\underline{X}^{(\ell)}$,
where $\underline{X}^{(j)} = (X_1^{(j)},\ldots,X_n^{(j)})$
represents the $j$-th step of the random process. In case $\ell = 2$,
we might call our two vectors $\underline{X}$ and $\underline{Y}$ instead.

For reasons outlined in Section \ref{sec:differentMarginals} we assume that
all of $X_i^{(1)}, \ldots, X_i^{(\ell)}$ have the same marginal distribution,
which we call $\pi$. We assume that $\Omega$ is the support of $\pi$.

Even though it is not necessary,
for clarity of the presentation
we assume that each coordinate $\overline{X}_i = 
(X_i^{(1)},\ldots,X_i^{(j)},\ldots,X_i^{(\ell)})$ has the 
same distribution $\mathcal{P}$.

\medskip

We consistently use index $i$ to index over the coordinates 
(from $[n]$)
and $j$ to index over the steps (from $[\ell]$).

As visible in Figure \ref{fig:naming},
we denote the aggregation across the coordinates by the underline
and the aggregation across the steps by the overline.
For example, we write $\underline{\Omega} = \Omega^n$,
$\overline{\Omega} = \Omega^\ell$, 
$\underline{\mathcal{P}} = \mathcal{P}^n$
and $\overline{\underline{X}} = 
(\overline{X}_1, \ldots, \overline{X}_n)
\allowbreak = (\underline{X}^{(1)}, \ldots, \underline{X}^{(\ell)})$.

We sometimes call $\underline{\mathcal{P}}$ a tensorized, 
multi-step probability distribution
as opposed to a tensorized, single-step distribution 
$\underline{\mathcal{\pi}}$ and single-coordinate, multi-step distribution
$\mathcal{P}$.

Furthermore, we extend the index notation to subsets of indices or steps.
For example, for $S \subseteq [\ell]$ we define $X^{(S)}$ to be the
collection of random variables $\left\{ X^{(j)}: j \in S  \right\}$.

We also use the set difference symbol to mark vectors with one element missing,
e.g., $\overline{X}^{\setminus j} := (X^{(1)}, \ldots, X^{(j-1)}, X^{(j+1)}, \ldots,
X^{(\ell)})$.

\medskip

One should think of $\ell$ and $|\Omega|$ as constants and of $n$ as large.
We aim to get bounds which are independent of $n$.

\subsection{Correlation}
\label{sec:correlation}
In case $\ell > 2$, the bound we obtain will depend on the
\emph{correlation} of the distribution $\mathcal{P}$.
This concept was used before in \cite{Mos10}.

\begin{definition}\label{def:correlation}
Let $\mathcal{P}$ be a single-coordinate distribution
and let $S, T \subseteq [\ell]$. We define the \emph{correlation}:
\begin{align*}
\rho(\mathcal{P}, S, T) &:=
\sup \Bigl\{ \Cov[f(X^{(S)}), g(X^{(T)})] \Bigm| 
f: \Omega^{(S)} \to \mathbb{R}, 
g: \Omega^{(T)} \to \mathbb{R}, \\
&\qquad\qquad
\Var[f(X^{(S)})] = \Var[g(X^{(T)})] = 1 
\Bigr\} \; .
\end{align*}
The correlation of $\mathcal{P}$ is
$\rho(\mathcal{P}) := \max_{j \in [\ell]} 
\rho\left(\mathcal{P}, \{j\}, [\ell]\setminus\{j\}\right)$.
\end{definition}

\subsection{Influence}

A crucial notion in the proof of Theorem~\ref{thm:rho-hitting} is
the \emph{influence} of a function. It expresses the average variance of a 
function, 
given that all but one of its $n$ inputs have been fixed to random values:
\begin{definition}
Let $\underline{X}$ be a random vector over alphabet $\underline{\Omega}$
and $f: \underline{\Omega} \to \mathbb{R}$ be a function
and $i \in [n]$. The \emph{influence of $f$ on the $i$-th coordinate} is:
\begin{align*}
  \Inf_i(f(\underline{X})) := \EE \left[ \Var\left[ f(\underline{X}) 
  \mid \underline{X}_{\setminus i} \right] \right]
  \; .
\end{align*}
The \emph{(total) influence of $f$} is 
$\Inf(f(\underline{X})) := \sum_{i=1}^n \Inf_i(f(\underline{X}))$.
\end{definition}
Note that the influence depends both on the function $f$ and
the distribution of the vector $\underline{X}$.

\section{Our Results}
\label{sec:results}

Here we give precise statements of our results presented in the introduction.

\subsection{The case of 
\texorpdfstring{$\ell = 2$}{l = 2}}

\begin{restatable}{theorem}{maintwovariables}
\label{thm:main-two-variables}
Let $\Omega$ be a finite set and
$\mathcal{P}$ a probability distribution over $\Omega^2$ 
with equal marginals $\pi$.
Let pairs $(X_i, Y_i)$ 
be i.i.d.~according to $\mathcal{P}$ for 
$i \in \{1,\ldots,n\}$.

Then, for every $f: \Omega^n \to [0,1]$ with 
$\EE[f(\underline{X})] = \mu > 0$:
\begin{align}
\EE [f(\underline{X})f(\underline{Y})] 
\geq c\left (\alpha(\mathcal{P}), \mu \right) \; ,
\end{align}
where the function $c()$ is positive whenever
$\alpha(\mathcal{P})> 0$.
\end{restatable}
We remark that Theorem~\ref{thm:main-two-variables} does not depend
on $\rho(\mathcal{P})$ in any way. 
This is in contrast to the case $\ell > 2$.
It is possible to obtain an inverse polynomial bound
$c(\mu) \ge \mu^C$
for symmetric two-step spaces (see Section~\ref{sec:polynomial-hitting}).

To prove Theorem \ref{thm:main-two-variables} we make
a convex decomposition argument and then apply
the multi-step Theorem \ref{thm:main-multiple}
(see Section~\ref{sec:two-steps}).
For completeness, we provide a proof of 
Theorem \ref{thm:two-steps-classification} assuming Theorem
\ref{thm:main-two-variables}.
\begin{proof}[Proof of Theorem \ref{thm:two-steps-classification}]
The ``if'' part follows from Theorem \ref{thm:main-two-variables}.
The ``only if'' can be seen by taking $f$ to be an appropriate dictator. 
\end{proof}

\subsection{The general case}

\begin{restatable}{theorem}{mainmultiple}
\label{thm:main-multiple}
Let $\Omega$ be a finite set and $\cP$ a distribution over $\Omega^{\ell}$ in
which all marginals are equal.  Let tuples
$\overline{X}_i = (X_i^{(1)}, \ldots, X_i^{(\ell)})$ be i.i.d.~according to
$\mathcal{P}$ for $i \in \{1,\ldots,n\}$.

Then, for every function $f: \Omega^n \to [0,1]$
with $\EE[f(\underline{X}^{(j)})] = \mu > 0$:
\begin{align}
  \EE \left[ \prod_{j=1}^\ell f(\underline{X}^{(j)}) \right]
  \geq c\left(\alpha(\mathcal{P}), \rho(\mathcal{P}), \ell, \mu \right) \; ,
\end{align}
where the function $c()$ is positive whenever $\alpha(\mathcal{P})>0$
and $\rho(\mathcal{P}) < 1$.

Furthermore, there exists some $D(\mathcal{P}) > 0$
(more precisely, $D$ depends on 
$\alpha$, $\rho$ and $\ell$) such that
if $\mu \in (0, 0.99]$, one can take:
\begin{align}
\label{eq:76a}
  c(\alpha, \rho, \ell, \mu) := 1 / \exp\left(\exp\left(\exp\left(
  \left(1/\mu\right)^D
  \right)\right)\right) \; .
\end{align}
\end{restatable}

Note that this bound \emph{does} depend on $\rho(\mathcal{P})$.
We also obtain a bound that does not depend on $\rho(\mathcal{P})$
for multi-step probability spaces generated by Markov chains
(see Section~\ref{sec:markov}).

\subsection{Hitting of different sets by uniform functions}

Finally, we state the generalization of low-influence theorem from
\cite{Mos10}. We assume that the reader is familiar with
Fourier coefficients $\hat{f}(\sigma)$ and the basics of discrete function 
analysis, for details see, e.g., Chapter 8 of \cite{Dol14}.
Note that this theorem requires neither equal marginals
nor $\alpha(\cP) > 0$.
For the proof see Section~\ref{sec:local-variance}.

\begin{restatable}{theorem}{localvariance}
\label{thm:local-variance}
Let $\overline{\underline{X}}$ be a random vector distributed according
to an $\ell$-step distribution $\cP$ with $\rho(\cP) \le \rho < 1$ and
let $\mu^{(1)}, \ldots, \mu^{(\ell)} \in (0, 1]$. 

There exist
$k \in \mathbb{N}$ and $\gamma > 0$ (both depending only
on $\cP$ and $\mu^{(1)}, \ldots, \mu^{(\ell)}$) such that for all functions
$f^{(1)}, \ldots, f^{(\ell)}: \underline{\Omega} \to [0, 1]$,
if $\EE[f^{(j)}(\underline{X}^{(j)})] = \mu^{(j)}$ and
$\max_{\sigma: 0 < |\sigma| \le k} 
|\hat{f}^{(j)}(\sigma)| \le \gamma$, then
\begin{align}
\label{eq:89a}
\EE\left[\prod_{j=1}^{\ell} f^{(j)}(\underline{X}^{(j)}) \right]
\geq c(\cP, \mu^{(1)}, \ldots, \mu^{(\ell)}) > 0 \; .
\end{align}
\end{restatable}

\subsection{Assumptions of the theorems}

\subsubsection{Equal distributions: unnecessary}
\label{sec:equal-distributions}
In Theorems~\ref{thm:main-two-variables},~\ref{thm:main-multiple}
and~\ref{thm:local-variance}  we assumed that the tuples 
$(X_i^{(1)},\ldots,X_i^{(\ell)})$ are distributed 
identically for each $i$.
It is natural to ask if it is indeed necessary.

This is not the case.
Instead, we made this assumption for simplicity of notation and presentation.
If one is interested in statements which are valid where coordinate
$i$ is distributed according to $\mathcal{P}_i$, one simply needs
to assume that there are  $\alpha > 0$ and $\rho < 1$ such 
that
$\alpha(\mathcal{P}_i) \geq \alpha$ and $\rho(\mathcal{P}_i) \leq \rho$.

\subsubsection{Equal marginals: necessary}
\label{sec:differentMarginals}

We quickly discuss the case when $\cP$ does not have equal marginals.
Recall that $\beta(\cP) = \min_{x^{(1)}, \ldots, x^{(\ell)} \in \Omega} 
\cP(x^{(1)}, \ldots, x^{(\ell)})$. If $\beta(\cP) > 0$, then,
by Theorem~\ref{thm:different-sets-classification}, $\cP$ is set hitting,
and therefore also same-set hitting.

In case $\beta(\cP) = 0$, we demonstrate an example which shows that 
$\EE\left[\prod_{j=1}^\ell f(\underline{X}^{(j)}) \right]$ 
can be exponentially small in $n$. 
For concreteness, we set $\ell := 2$ and $\Omega := \{0,1\}$ 
and consider $\mathcal{P}$ which picks uniformly
among $\{00, 01, 11\}$.
We then set \newcommand{\wt}{\mathrm{wt}}
\begin{align}
S_1 &:= \{(x_1,\ldots,x_n) \mid x_1 = 1 \land 
|\wt(x) - n/3| \leq 0.01n\}\\
S_2 &:= \{(x_1,\ldots,x_n) \mid x_1 = 0 \land 
|\wt(x) - 2n/3| \leq 0.01n\}
\end{align}
where $\wt(x)$ is the Hamming-weight of $x$, i.e., the number of ones
in~$x$.

For large enough $n$, a concentration bound 
implies that $\Pr[\underline{X}^{(1)} \in S_1] > \frac{1}{3}-0.01$ 
and $\Pr[\underline{X}^{(2)} \in S_2] > \frac{1}{3}-0.01$.
Hence, if we set $f$ to be the indicator function of $S := S_1 \cup S_2$, 
the assumption of Theorem~\ref{thm:main-multiple} holds.
However, because of the first coordinate we have
$\Pr[\underline{X}^{(1)} \in S \land \underline{X}^{(2)} \in S]
\leq \Pr[\underline{X}^{(1)} \in S_2] + \Pr[\underline{X}^{(2)} \in S_1]$,
and the right hand side is easily seen to be exponentially
small.

It is not difficult to extend this example to any distribution 
with $\beta(\cP) = 0$ that does not
have equal marginals.

\section{Proof Sketch}
\label{sec:proof-sketch}

In this section we briefly outline the proof of Theorem \ref{thm:main-multiple}. 
For simplicity, we assume that the probability space is the one from 
Section \ref{sec:basic}, i.e.,
$(X_i, Y_i)$ are distributed uniformly in $\{00, 11, 22, 01, 12, 20\}$.
Additionally, we assume that we are given a set $S \subseteq \{0,1,2\}^n$
with $\mu(S) = |S|/3^n > 0$, so that we want a bound of the form
\begin{align*}
  \Pr\left[ \underline{X} \in S \land \underline{Y} \in S \right] \ge c(\mu) > 0 \; .
\end{align*}

The proof consists of three steps.
Intuitively, in the first step we deal with dictator sets, e.g.,
$S_{\mathrm{dict}} = \{ \underline{x}: x_1 = 0\}$,
in the second step with linear sets, e.g.,
$S_{\mathrm{lin}} = \{ \underline{x}: \sum_{i=1}^n x_i \pmod 3 = 0 \}$
and in the third step with threshold sets, e.g.,
$S_{\mathrm{thr}} = \{ \underline{x}: |\{i: x_i = 0\}| \ge n/3 \}$.

\subsection{Step 1 --- making a set resilient}

We call a set resilient if $\Pr[\underline{X} \in S]$ does not change by more
than a (small) multiplicative constant factor whenever conditioned on 
$(X_{i_1} = x_{i_1}, \ldots, X_{i_s} = x_{i_s})$ on a constant number $s$ of 
coordinates.

In particular, $S_{\mathrm{dict}}$ is not resilient (because
conditioning on $x_1 = 0$ increases the measure of the set to $1$), 
while $S_{\mathrm{lin}}$ and $S_{\mathrm{thr}}$ are.

If a set is not resilient, using $\mathcal{P}(x, x) = 1/6$ for every 
$x \in \Omega$,
one can find an event $\cE :\equiv X_{i_1}=Y_{i_1}=x_{i_1} \land \ldots \land X_{i_s}=Y_{i_s}=x_{i_s}$
such that for some constant $\epsilon > 0$ we have
$\Pr[ \cE ] \ge \epsilon$ and, at the same time,
 $\Pr[\underline{X} \in S \mid \cE] \ge (1+\epsilon)\Pr[\underline{X} \in S]$.

Since each such conditioning increases the measure of the set $S$ by a constant 
factor, $S$ must become resilient after a constant number of iterations.
Furthermore,
each conditioning induces only a constant factor loss in
$\Pr[\underline{X} \in S \land \underline{Y} \in S]$.

It is worth noting that this is the only stage of the proof where
we assume the same-set property (and utilize the assumption
$\alpha(\mathcal{P}) > 0$).

\subsection{Step 2 --- eliminating high influences}

In this step, assuming that $S$ is resilient, we condition on a constant
number of coordinates to transform it into two
sets $S'$ and $T'$ such that:
\begin{itemize}
\item Both of them have low influences on all coordinates.
\item Both of them are supersets of $S$ (after conditioning).
\end{itemize}

The first property allows us to apply low-influence set hitting 
from $\cite{Mos10}$ to
$S'$ and $T'$. The second one, together with the resilience of $S$, 
ensures that $\mu(S'), \mu(T') \ge (1-\epsilon)\mu(S)$.

In fact, it is more convenient to assume that we are initially
given two resilient sets $S$ and $T$.

Assume w.l.o.g.~that $\Inf_1(T) \ge \tau$ for some $i \in [n]$.
Given $z \in \{0, 1, 2\}$, let 
$T_z := \{(x_1, x_2, \ldots, x_n): (z, x_2, \ldots, x_n) \in T\}$.
Furthermore, let $T^*_z := T_z \cup T_{z+1 \pmod 3}$.

Since $\Inf_1(T) \ge \tau$, we can show that
there exists $z \in \{0,1,2\}$ such that, after conditioning on
$X_1 = Y_1 = z$,
the sum $\mu(S_z)+ \mu(T^*_z)$ is strictly greater than
the sum $\mu(S) + \mu(T)$:
\begin{align}
\label{eq:83a}
  \Pr[\underline{X} \in S_z \mid X_1 = z] + 
  \Pr[\underline{Y} \in T^*_z \mid Y_1 = z]
  \ge \Pr[\underline{X} \in S] + \Pr[\underline{Y} \in T] + c(\tau) \; .
\end{align}

We choose to disregard the first coordinate and replace
$S$ with $S' := S_z$ and $T$ with $T' := T^*_{z}$.
Equation (\ref{eq:83a}) implies that after a constant number of such 
operations, neither $S$ nor $T$ has any remaining high-influence coordinates. 

Crucially, with respect to same-set hitting our set replacement
is essentially equivalent to conditioning on
$X_1 = z$ and $Y_1 = z \lor Y_1 = z + 1 \pmod 3$.
Therefore, each operation induces only a constant factor loss in 
$\Pr[\underline{X} \in S \land \underline{Y} \in T]$.

\subsection{Step 3 --- applying low-influence theorem from 
\texorpdfstring{\cite{Mos10}}{[Mos10]}}

Once we are left with two low-influence, somewhat-large sets $S$
and $T$, we 
obtain
$\Pr[\underline{X} \in S \land \underline{Y} \in T] \ge c(\mu) > 0$
by a straightforward application of a slightly
modified version of Theorem 1.14 from \cite{Mos10}.
The theorem gives that $\rho(\mathcal{P}) < 1$ implies that
the distribution $\mathcal{P}$ is set hitting for low-influence functions:

\begin{restatable}{theorem}{lowinfluence}
\label{thm:low-influence}
Let $\overline{\underline{X}}$ be a random vector distributed
according to $(\overline{\underline{\Omega}}, \underline{\mathcal{P}})$ 
such that $\cP$ has equal marginals, $\rho(\mathcal{P}) \le \rho < 1$
and $\min_{x \in \Omega} \pi(x) \ge \alpha > 0$.

Then, for all $\epsilon > 0$, there exists 
$\tau := \tau(\epsilon, \rho, \alpha, \ell) > 0$
such that if functions 
$f^{(1)}, \ldots, f^{(\ell)}: \underline{\Omega} \to [0, 1]$ satisfy
\begin{align}
\max_{i \in [n], j \in [\ell]} \Inf_i(f^{(j)}(\underline{X}^{(j)})) \le \tau\;,
\end{align}
then, for $\mu^{(j)} := \EE[f^{(j)}(\underline{X}^{(j)})]$:
\begin{align}
\EE \left[ \prod_{j=1}^{\ell} f^{(j)}(\underline{X}^{(j)}) \right] 
\ge \left( \prod_{j=1}^\ell \mu^{(j)} \right)^{\ell / (1-\rho^2)} - \epsilon \; .
\end{align}

Furthermore, there exists an absolute constant $C \ge 0$ such that
for $\epsilon \in (0, 1/2]$ one can take
\begin{align}\label{eq:38a}
\tau := \left(\frac{(1-\rho^2)\epsilon}{\ell^{5/2}}
\right)^{C \frac{\ell \ln(\ell/\epsilon)\ln(1/\alpha)}{(1-\rho)\epsilon}} 
\; .
\end{align}
\end{restatable}

The proof of Theorem \ref{thm:low-influence} can be found in
Appendix~\ref{sec:low-influence-proof}.
The first part of the appendix contains a short
explanation of differences
between \cite{Mos10} and our version.

\subsection{The case
\texorpdfstring{$\rho = 1$}{rho = 1}
: open question}
\label{sec:open}

Theorem~\ref{thm:main-multiple} requires that $\rho < 1$ in order to give a
meaningful bound.  It is unclear whether this is an artifact of our proof or if
it is necessary.  In particular, consider the three step distribution $\mathcal{P}$ 
which
picks a uniform triple from
$\{000, 111, 222, 012, 120, 201\}$.
In other words, sampling from $\underline{\cP}$ picks a random
arithmetic progression $x, x+d, x+2d$ with $x \in \mathbb{F}_3^n$
and $d \in \{0,1\}^n$.
One easily
checks that $\rho(\mathcal{P})=1$ and that all marginals are uniform.
We do not know if this distribution is same-set hitting.

However, the method of our proof breaks down. 
We illustrate the reason in the following lemma.
\begin{lemma}\label{lem:threeSetFail}
For every $n > n_0$ there exist three sets $S^{(1)}$, $S^{(2)}$, and $S^{(3)}$ such 
that for
the distribution $\mathcal{P}$ as described above we have
\begin{itemize}
\item $\forall j: \Pr[ \underline{X}^{(j)} \in S^{(j)} ] \geq 0.49$.
\item $\Pr[\forall j: \underline{X}^{(j)} \in S^{(j)}] = 0$.
\item The characteristic functions $\1_{S^{(j)}}$ of the three sets all
satisfy
\begin{align*}
\max_{i \in [n]} \Inf_{i}(\1_{S^{(j)}}(\underline{X}^{(j)})) \to 0
\text{ as } n \to \infty \; .
\end{align*}
\end{itemize}
\end{lemma}
While the lemma does not give information about whether $\mathcal{P}$ is
same-set hitting, it shows that
our proof fails (since the analogue of Theorem~\ref{thm:low-influence}
fails).
\begin{proof}
We let 
\begin{align*}
S^{(1)} &:= \{\underline{x}^{(1)} : \text{$\underline{x}^{(1)}$ has 
less than $n/3$ twos}\}\;,\\
S^{(2)} &:= \{\underline{x}^{(2)} : \text{$\underline{x}^{(2)}$ has 
less than $n/3$ ones}\}\;,\\
S^{(3)} &:= \{\underline{x}^{(3)} : \text{$\underline{x}^{(3)}$ has 
less than $n/3$ zeros}\}\;.
\end{align*}
Whenever we pick
$\underline{X}^{(1)}, \underline{X}^{(2)}, \underline{X}^{(3)}$, the number of
twos in $\underline{X}^{(1)}$ plus the number of ones in $\underline{X}^{(2)}$
plus the number of zeros in $\underline{X}^{(3)}$ always equals $n$ (there is a
contribution of one from each coordinate).  All three properties are now easy to
check.
\end{proof}

\section{Proof for General
\texorpdfstring{$\ell$}{l}
and
\texorpdfstring{$\rho(\mathcal{P}) < 1$}{rho(P) < 1}}
\label{sec:main-proof}

The goal of this section is to prove our second main result, which we
restate here for convenience.

\mainmultiple*

\subsection{Properties of the correlation}
\label{sec:correlation-properties}

Recall Definition \ref{def:correlation}.
We now give an alternative characterization of $\rho(\mathcal{P}, \{j\},
\allowbreak [\ell]\setminus \{j\})$ which will be useful later.
For this, we first define certain random process and an associated Markov chain.
\begin{definition}
\label{def:double-sample}
Let $\mathcal{P}$ be a single-coordinate distribution
and let $j \in [\ell]$.
We call a collection of random variables 
$(\overline{X}^{\setminus j}
= (X^{(1)}, \ldots, \allowbreak X^{(j-1)}, \allowbreak X^{(j+1)}, \ldots, X^{(\ell)}), 
\allowbreak Y, Z)$
a \emph{double sample on step $j$ from $\mathcal{P}$} if:
\begin{itemize}
\item
  $\overline{X}$ is first sampled according to $\mathcal{P}$, ignoring step $j$.
\item Assuming that $\overline{X}^{\setminus j} = \overline{x}^{\setminus j}$,
the random variables $Y$ and $Z$ are then sampled independently of each other 
according to the
$j$-th step of $\cP$ conditioned on 
$\overline{X}^{\setminus j} = \overline{x}^{\setminus j}$.
\end{itemize}
Sometimes we will omit $\overline{X}^{\setminus j}$ from the notation 
and refer as double sample to $(Y, Z)$ alone.
\end{definition}

An equivalent interpretation of a double sample
is that after sampling $(\overline{X}^{\setminus j}, Y)$ according to
$\cP$ we ``forget'' about $Y$ and sample $Z$ again from the same distribution
(keeping the same value of $\overline{X}^{\setminus j}$).
Therefore, both
$(\overline{X}^{\setminus j}, Y)$ and $(\overline{X}^{\setminus j}, Z)$
are distributed
according to $\mathcal{P}$.

If we let
\[
K(y,z) := \Pr[Z=z | Y=y] = \EE\left[\Pr\left[Z = z | Y = y, \overline{X}^{\setminus j}
\right]\right],
\]
we see that
\begin{align}
\pi(y)K(y,z) = \Pr[Y = y \land Z = z] = \Pr[Y = z \land Z = y]
= \pi(z)K(z,y) \;,
\end{align}
which means that $K$ is the kernel of a Markov chain that is reversible with
respect to $\pi$ (see e.g., \cite[Section 1.6]{LevinPW08}).  Thus, $K$ has an
orthonormal eigenbasis with eigenvalues
$1=\lambda_1(K) \geq \lambda_2(K) \geq \dots \geq \lambda_{|\Omega|}(K)\geq -1$, 
(e.g., \cite[Lemma 12.2]{LevinPW08}).
We will say that $K$ is the \emph{Markov kernel induced by the double sample
  $(Y, Z)$}.
  
A standard fact from the Markov chain theory
expresses $\lambda_2(K)$ in terms of covariance of functions
$f \in L^2(\Omega, \pi)$:
\begin{lemma}[Lemma 13.12 in \cite{LevinPW08}]
  Let $Y, Z$ be two consecutive steps of a reversible Markov chain with kernel
  $K$ such that both $Y$ and $Z$ are distributed according to a stationary
  distribution of $K$. Then,
\label{lem:lambda-lpw}
\begin{align}
\label{eq:72a}
\lambda_2(K) = \max_{\substack{f: \Omega \to \mathbb{R} \\\EE[f(Y)] = 0\\\Var[f(Y)]=1}} 
  \EE\left[ f(Y) f(Z) \right] \; .
\end{align}
\end{lemma}

\begin{lemma}
\label{lem:lambda-rho}
Let $\cP$ be a single-coordinate distribution and let 
$(\overline{X}^{\setminus j}, Y, Z)$ be a double sample from $\cP$
that induces a Markov kernel $K$.
Then,
\begin{align*}
\lambda_2(K) = \rho(\mathcal{P},\{j\},[\ell]\setminus\{j\})^2 \; .
\end{align*}
\end{lemma}
\begin{proof}
  For readability, let us write $\overline{X}$ instead of
  $\overline{X}^{\setminus j}$.

  Consider first two functions $f$ and $g$ as in
  Definition~\ref{def:correlation} and assume without loss of generality that
  $\EE[f(Y)] = \EE[g(\overline{X})] = 0$.
  Of course, we also assume that
  $\Var[f(Y)]=\Var[g(\overline{X})]=1$ as specified by
  Definition~\ref{def:correlation}. We will show that
\begin{align}
\label{eq:73a}
  \Cov\left[f(Y), g(\overline{X})\right]^2 \le \lambda_2(K) \; ,
\end{align}
and that there exists a choice of $f$ and $g$ that achieves equality
in (\ref{eq:73a}).

Let $h(\overline{x}) := \EE[f(Y) | \overline{X} = \overline{x}]$ and observe
that
\begin{IEEEeqnarray}{rCl}
\EE[f(Y)f(Z)] &=& \sum_{\overline{x},y,z} \Pr[\overline{X}=\overline{x}]
\Pr[Y=y\mid \overline{X}=\overline{x}]\Pr[Z=z\mid \overline{X}=\overline{x}]
f(y)f(z) \nonumber
\\ &=& 
\EE[h(\overline{X})^2] \; .
\label{eq:71a}
\end{IEEEeqnarray}

Now, by Cauchy-Schwarz, (\ref{eq:71a}) and Lemma \ref{lem:lambda-lpw} 
we see that
\begin{IEEEeqnarray*}{rCl}
\Cov[f(Y),g(\overline{X})]^2 &=& \EE[f(Y)g(\overline{X})]^2 = 
\EE[h(\overline{X})g(\overline{X})]^2
\le \EE[h(\overline{X})^2] \EE[g(\overline{X})^2] \\
&=& \EE[h(\overline{X})^2]
= \EE[f(Y)f(Z)] \le \lambda_2(K) \; . 
\end{IEEEeqnarray*}
The equality is obtained for $f$ that maximizes the right-hand side of 
(\ref{eq:72a}) and $g := c\cdot h$ for some $c > 0$.
\end{proof}

For later use, we make the following implication of Lemma \ref{lem:lambda-rho}.
\begin{corollary}\label{cor:varianceToEdgeVar}
Let $(Y, Z)$ be a double sample on step $j$ from a single-coordinate
distribution $(\Omega, \mathcal{P})$ with $\rho(\cP) = \rho$.
Then, for every function $f: \Omega \to \mathbb{R}$,
\begin{align}
\EE[(f(Y)-f(Z))^2] \geq 2(1-\rho^2) \Var[f(Y)].
\end{align}
\end{corollary}
\begin{proof}
Assume w.l.o.g.~that $\EE[f(Y)]=0$.
By Lemmas \ref{lem:lambda-lpw} and \ref{lem:lambda-rho},
\begin{align*}
\EE\left[ \left( f(Y)-f(Z) \right)^2 \right]
= 2\left( \Var[f(Y)] - \EE[f(Y)f(Z)] \right)
\ge 2(1-\rho^2)\Var[f(Y)] \; .
\end{align*}
\end{proof}

\subsection{Reduction to the resilient case}
In this section, we will prove that we can assume that 
the function $f$ is \emph{resilient} in the following sense:
whenever we fix a constant number of inputs to some 
value, the expected value of $f$ remains roughly the same.

The intuitive reason for this is simple: if there is some way to fix
the coordinates which changes the expected value of $f$, we
can fix these coordinates such that the expected value \emph{increases}, 
which only makes our task easier
(and can be done only a constant number of times).

\medskip

We first make the concept of ``fixing'' a subset of the coordinates formal.

\begin{definition}
Let $f: \underline{\Omega} \to [0,1]$ be a function.
A \emph{restriction} $\cR$ is a sequence $\cR = (r_1,\ldots,r_n)$
where each $r_i$ is either an element $r_i \in \Omega$,
or the special symbol $r_i = \star$.

The coordinates with $r_i = \star$ are \emph{unrestricted}, the coordinates
where $r_i \in \Omega$ are \emph{restricted}.
The \emph{size} of a restriction is the number of restricted coordinates.

A restriction $\cR$ operates on a function $f$ as
\begin{align}
(\cR f) (x_1,\ldots,x_n) := f(y_1,\ldots,y_n) \; 
\end{align}
where $y_i = r_i$ if $r_i \neq \star$ and $y_i = x_i$ otherwise.
\end{definition}

Next, we define what it means for a function to be resilient:
restrictions do not change the expectation too much.
\begin{definition}
Let $\underline{X}$ be a random vector distributed according
to a (single-step) distribution $(\underline{\Omega}, \underline{\pi})$.
A function $f: \underline{\Omega} \to [0,1]$ is \emph{$\epsilon$-resilient up to
size $k$} if for every restriction $\cR$ of size at most $k$ we have that
$(1-\epsilon)\EE[f(\underline{X})] \leq \EE[\cR f(\underline{X})] 
\leq (1+\epsilon)\EE[f(\underline{X})]$.
\end{definition}

The function is upper resilient if the expectation cannot increase too much.
\begin{definition}
Let $\underline{X}$ be a random vector distributed according
to a distribution $(\underline{\Omega}, \underline{\pi})$.
A function $f: \underline{\Omega} \to [0,1]$ is \emph{$\epsilon$-upper resilient
up to size $k$} if for every 
restriction $\cR$ of size at most $k$ we have that
$\EE[\cR f(\underline{X})] \leq (1+\epsilon) \EE[f(\underline{X})]$.
\end{definition}

Resilience and upper resilience are equivalent up to a multiplicative
factor which depends only on $k$ and the smallest probability in the
marginal distribution  $\alpha(\pi)$.
Intuitively the reason is that if there is some restriction which
decreases the $1$-norm, then some other restriction on the same
coordinates must increase the $1$-norm somewhat.

\begin{lemma}
\label{lem:stability-equivalence}
Suppose that a function $f$ is $\epsilon$-upper resilient up to size $k$.
Then, $f$ is $\epsilon'$-resilient up to size $k$, where 
$\epsilon' = \epsilon/(\alpha(\pi))^k$.
\end{lemma} 
\begin{proof}
Fix a subset $S \subseteq [n]$ of the coordinates of size $|S| \le k$.
We consider a random variable $\cR$ whose values are
restrictions with restricted coordinates being exactly $S$.
The elements $r_i \in \Omega$ for $i \in S$ are picked according to the 
distribution $\pi$.
We let $p(\cR')$ be the probability a certain restriction $\cR'$ is picked, and
get
\begin{align}
  \EE[f(\underline{X})] =
  \sum_{{\cR}'} p({\cR}') \cdot \EE \left[ \cR' f(\underline{X}) \right] 
  \; ,
\end{align}
where we sum over all restrictions $\cR'$ that restrict exactly the coordinates
in $S$.

Let now $\cR^*$ be one of the possible choices for $\cR$.
Then,
\begin{align*}
p(\cR^*) \cdot \EE[\cR^* f(\underline{X})]
&= \EE[f(\underline{X})] - \sum_{\cR' \neq \cR^*}p(\cR')\cdot
\EE[\cR' f(\underline{X})]
\\
&\geq
\EE[f(\underline{X})] - (1+\epsilon)
\sum_{\cR' \neq \cR^*}p(\cR')\cdot \EE[f(\underline{X})]
\\
&=
\left(1 - (1+\epsilon)(1-p(\cR^*))\right) \cdot \EE[f(\underline{X})]
\\
&\geq \left(p(\cR^*) -\epsilon\right) \cdot \EE[f(\underline{X})] \; ,
\end{align*}
and hence:
\begin{align*}
\EE[ \cR^* f(\underline{X}) ]
&\geq \left(1 - \frac{\epsilon}{p(\cR^*)}\right) \cdot \EE[f(\underline{X})]\;.
\end{align*}
Since $p(\cR^*)  \geq \alpha(\pi)^k$ we get the bound
for the restriction $\cR^*$, which was chosen arbitrarily.
\end{proof}

\begin{lemma}\label{lem:restrictToMakeStable}
Let $\overline{\underline{X}}$ be a random vector distributed according to
a distribution with equal marginals 
$(\underline{\Omega}, \underline{\mathcal{P}})$
and $f: \underline{\Omega} \to [0, 1]$ be a function
with $\EE[f(\underline{X}^{(1)})] = \mu > 0$.

Let $\epsilon \in (0,1], k \in \mathbb{N}$. Then, there exists
a restriction $\cR$ such that $g := (\cR f)$ is $\epsilon$-resilient
up to size $k$ and
\begin{align}
\EE[g(\underline{X}^{(1)})] &\geq \mu 
                              \label{eq:241a}\;,\\
\EE \left[ \prod_{j=1}^{\ell} f(\underline{X}^{(j)}) \right] &\ge 
c \cdot \EE \left[ \prod_{j=1}^{\ell} g(\underline{X}^{(j)}) \right] \; ,
\label{eq:3a}
\end{align}
where $c := \exp\left(-\frac{2 \ln 1/\mu}{\alpha^{2k} \cdot \epsilon} \right)$ 
with $\alpha := \alpha(\mathcal{P}) > 0$.
\end{lemma}

In particular, $c$ depends only on $\epsilon, k, \alpha(\mathcal{P})$ and $\mu$
(requiring $\epsilon, \alpha(\mathcal{P}), \mu > 0$).

\begin{proof}
Let $\epsilon' := \alpha^k\cdot\epsilon$ and
choose a restriction $\cR$ such that 
$\EE[\cR f(\underline{X}^{(1)})] \geq \EE[f(\underline{X}^{(1)})] 
\cdot (1+\epsilon')$.
We repeat this, replacing $f$ with $(\cR f)$, until there is no such 
restriction.

Since the expectation of $f$ only increases, we get (\ref{eq:241a}).
Finally, once the process stops, the resulting function is $\epsilon$-resilient
due to Lemma~\ref{lem:stability-equivalence} 
(note that $\alpha(\pi) \ge \alpha$).

It remains to argue that (\ref{eq:3a}) holds for the resulting function.
Note first that the expectation cannot exceed~$1$, and
hence the process will be repeated at 
most $p := \ln(1/\mu)/\ln(1+\epsilon') \le \frac{2\ln(1/\mu)}{\epsilon'}$ times.
Therefore, the final restriction $\cR$ obtained after at most $p$ iterations
of the process above is of size at most $pk$.

Define $g := (\cR f)$ and
let $\mathcal{E}$ be the event that all 
strings $\underline{X}^{(1)},\ldots,\underline{X}^{(\ell)}$
agree with the restriction $\cR$ in its restricted coordinates.
We will use $\1(\cE)$ to denote the function which is 
$1$ if event $\cE$ happens and $0$ otherwise.
We see that
\begin{align*}
\EE \left[ \prod_{j=1}^{\ell} f(X^{(j)}) \right] &\ge 
\EE\left[ \prod_{j=1}^{\ell} f(X^{(j)}) \cdot \1(\mathcal{E}) \right]
=
\EE\left[ \prod_{j=1}^{\ell} g(X^{(j)}) \cdot \1(\mathcal{E}) \right]\\
&\geq 
\alpha^{pk} 
\cdot \EE \left[ \prod_{j=1}^{\ell} g(X^{(j)}) \right] \; .
\end{align*}
Finally,
\begin{align*}
\alpha^{pk} \ge 
\exp\left(-\frac{2 k \ln (1/\alpha) \ln (1/\mu)}{\alpha^k \cdot \epsilon}\right)
\ge \exp \left( - \frac{2 \ln 1/\mu}{\alpha^{2k} \cdot \epsilon}  \right) 
\; .
\end{align*}
\end{proof}

\subsection{Reduction to the low-influence case}
We next show that if $f$ is resilient, we can also assume that it has
only low influences.
However, this part of the proof actually produces a collection of functions 
$g^{(1)},\ldots, g^{(\ell)}$ such that each of them
has small influences: it operates differently on each function.
In turn, it is more convenient to do this part of the proof also starting 
from a collection $f^{(1)},\ldots,f^{(\ell)}$, as long as all of
them are sufficiently resilient.

\medskip

As in the previous section, we use restrictions.
Here, however, we are only interested in restrictions of size one.
Consequently, we write $\cR[i,a]$ to denote 
the restriction $\cR = (r_1,\ldots,r_n)$
with $r_i = a$ and $r_{i'} = \star$ for $i' \neq i$.

Furthermore, we require a new operator.
\begin{definition}
Let $f: \underline{\Omega} \to [0, 1]$, $i \in [n]$,
and fix values $y, z \in \Omega$.

We define the operator $\mathcal{M}[i,y,z]$ as
\begin{align*}
(\mathcal{M}[i,y,z]f)(x_1, \ldots, x_n) :=
	\max \bigl( f(&x_1, \ldots, x_{i-1}, y, x_{i+1},\ldots, x_n), \\
                    f(&x_1, \ldots, x_{i-1}, z, x_{i+1}, \ldots, x_n) \bigr) \; .
\end{align*}
\end{definition}

The operator $\mathcal{M}[i,y,z]$ is useful for two reasons.
First, if $\Inf_i(f^{(j)})$ is ``large'', then
$\EE\left[\mathcal{M}[i,y,z]f^{(j)}(\underline{X}^{(j)})\right] \geq
\EE[f^{(j)}(\underline{X}^{(j)})]+c$
for some $y,z \in \Omega$ and $c > 0$.  This implies that we can use this
operator to increase the expectation of a function unless all of its
influences are small. We will prove this property later.

Second, fix a step $j^* \in [\ell]$ and assume that for some values
$\overline{x}^{\setminus j^*} = (x^{(1)}, \ldots, x^{(j^*-1)}, \allowbreak
x^{(j^*+1)}, \ldots,
x^{(\ell)}), y, z \in \Omega$
both conditional probabilities
$\Pr[X_i^{(j^*)} = y \mid \overline{X}_i^{\setminus j^*} = \overline{x}^{\setminus
  j^*}]$
and
$\Pr[X_i^{(j^*)} = z \mid \overline{X}_i^{\setminus j^*} = \overline{x}^{\setminus
  j^*}]$ are ``somewhat large'' (larger than some constant).
We imagine now that 
$\overline{X}_i^{(\setminus j^*)} = \overline{x}^{\setminus j^*}$
and that we have also picked all values
$\underline{X}^{(j^*)}_{\setminus i} = (X_{1}^{(j^*)},\ldots,X_{i-1}^{(j^*)},
X_{i+1}^{(j^*)},\ldots,X_{n}^{(j^*)})$.
We then hope that $X^{(j^*)}_i$ is picked among $y$ and $z$ such that it
maximizes~$f^{(j^*)}$.  Since this happens with constant probability, we
conclude the following: Suppose we replace $f^{(j^*)}$ with
$\mathcal{M}[i,y,z]f^{(j^*)}$ and then prove that afterwards
$\EE[\prod f^{(j)}(\underline{X}^{(j)})]$ is large.  Then,
$\EE[\prod f^{(j)}(\underline{X}^{(j)})]$ was large before.

This second point is formalized in the following lemma:
\begin{lemma}\label{lem:maximumKeepsProbability}
  Let $\overline{\underline{X}}$ be a random vector distributed according to
  $(\overline{\underline{\Omega}}, \underline{\mathcal{P}})$.  Fix $i \in [n]$,
  $j^* \in [\ell]$
  and
  $\overline{x}^{\setminus j^*} = (x^{(1)},\ldots,x^{(j^*-1)},
  x^{(j^*+1)},\ldots,x^{(\ell)}), y, z \in \Omega$.  Suppose that:
\begin{align}
\cP(\overline{x}^{\setminus j^*},y) &\ge \beta \;,
\label{eq:845a}\\
\cP(\overline{x}^{\setminus j^*},z) &\ge \beta
\label{eq:846a}\;.
\end{align}
Let $f^{(1)},\ldots,f^{(\ell)}: \underline{\Omega} \to [0,1]$, and
for $j \in [\ell]$ define:
\begin{align}
g^{(j)} := 
\begin{cases}
\cR[i,x^{(j)}]f^{(j)} & \text{if $j \neq j^*$,}\\
\cM[i,y,z]f^{(j)} & \text{if $j = j^*$.}
\end{cases}
\end{align}
Then:
\begin{align}\label{eq:6a}
\EE\left[\prod_{j=1}^{\ell} f^{(j)}(\underline{X}^{(j)})\right]
\geq
\beta \cdot \EE\left[\prod_{j=1}^{\ell} g^{(j)}(\underline{X}^{(j)})\right]\;.
\end{align}
\end{lemma}
\begin{proof}
We first define a random variable $A$, 
which is the value among $y$ and $z$ which $X_{i}^{(j^*)}$ needs to take 
in order to maximize $f^{(j^*)}$. 
Formally,
\begin{align}
  A = \begin{cases}
    y & \text{if 
      $f\left(\underline{X}^{(j^*)}_{\setminus i}, y\right) 
      > f\left(\underline{X}^{(j^*)}_{\setminus i}, z\right)$,} \\
    z & \text{otherwise.}
\end{cases}
\end{align}
Consider now the event $\cE$ which occurs if 
$\overline{X}_i = (\overline{x}^{\setminus j}, A)$.
We get
\begin{align*}
\EE\left[\prod_{j=1}^{\ell} f^{(j)}(\underline{X}^{(j)})\right]
&\geq
\EE\left[
\prod_{j=1}^{\ell} f^{(j)}(\underline{X}^{(j)}) \cdot \1(\mathcal{E})
 \right]
\\
&=
\EE\left[
\prod_{j=1}^{\ell} g^{(j)}(\underline{X}^{(j)}) \cdot \1(\mathcal{E})
 \right]\\
&=
\EE\left[
\EE\left[
\prod_{j=1}^{\ell} g^{(j)}(\underline{X}^{(j)}) \cdot \1(\mathcal{E})
\middle | 
\overline{\underline{X}}_{\setminus i}
\right]
\right]\\
&=
\EE\left[
\prod_{j=1}^{\ell} g^{(j)}(\underline{X}^{(j)}) \cdot
\EE\left[
\1(\mathcal{E})
\middle | 
\overline{\underline{X}}_{\setminus i}
\right]
\right]\\
&\geq
\beta \cdot
\EE\left[
\prod_{j=1}^{\ell} g^{(j)}(\underline{X}^{(j)})
\right] \; .
\end{align*}

The equality from the first to the second line follows because
if the event $\cE$ happens, then the functions 
$f^{(j)}(\underline{X}^{(j)})$ 
and $g^{(j)}(\underline{X}^{(j)})$ are equal.
From the third to the fourth line we use that conditioned on 
$\underline{\overline{X}}_{\setminus i}$ the 
functions $g^{(j)}(\underline{X}^{(j)})$ are constant.
Finally, the last inequality follows because by (\ref{eq:845a})
and~(\ref{eq:846a}), for every choice of
$\underline{\overline{X}}_{\setminus i} = 
(\overline{X}_1,\ldots,\overline{X}_{i-1},
\overline{X}_{i+1},\ldots,\overline{X}_{n})$ event $\mathcal{E}$
has probability at least $\beta$.
\end{proof}

The obvious idea for the next step would be to 
find values $\overline{x}^{\setminus j}, y, z$
such that 
\begin{align*}
\EE\left[\cM[i,y,z]f^{(j^*)}(\underline{X}^{(j^*)})\right] \geq 
\EE\left[f^{(j^*)}(\underline{X}^{(j^*)})\right]+c
\end{align*}
and fix them.

Unfortunately, there is a problem with this strategy.  To replace 
the function $f^{(j^*)}$ with
$\cM[i,y,z] f^{(j^*)}$, Lemma~\ref{lem:maximumKeepsProbability} also
replaces $f^{(j)}$ with $\cR[i,x^{(j)}]f^{(j)}$ for $j \ne j^*$ (and this is
required for the proof to work).  Unfortunately, it is possible that
$\EE\Big[\cR[i,x^{(j)}]f^{(j)}(\underline{X}^{(j)})\Big] \ll
\EE\Big[f^{(j)}(\underline{X}^{(j)})\Big]$.
We remark that we \emph{cannot} use that $f^{(j)}$ is resilient here: while
$f^{(j)}$ is resilient the first time we condition, the functions
$\cM[i,y,z] f^{(j)}$ obtained in the subsequent steps are not resilient in general,
so later steps will not have the guarantee.

\medskip

Our solution is to pick the values $(\overline{X}^{\setminus j^*}, Y, Z)$ at
random, as a double sample on coordinate $j^*$ 
(cf.~Definition \ref{def:double-sample}).
Let:
\begin{align*}
  G^{(j)} := \begin{cases}
    \cR[i,X^{(j)}]f^{(j)} &\text{if $j \ne j^*$,}\\
    \cM[i, Y, Z]f^{(j)}  &\text{if $j = j^*$.}
  \end{cases}
\end{align*}
Note that the random variable $X^{(j)}$ is part of the double sample
$(\overline{X}^{\setminus j^*}, Y, Z)$ and sampled separately (and independently)
from the random vector $\overline{\underline{X}}$. In particular, it should
not be confused with the ``input'' random variable $X_i^{(j)}$.
We prove that (in expectation over 
$\overline{X}^{(\setminus j)},Y,Z$) the sum of expectations
$\sum_{j=1}^\ell \EE[G^{(j)}(\underline{X}^{(j)})]$ is greater by a constant than
the sum $\sum_{j=1}^\ell \EE[f^{(j)}(\underline{X}^{(j)})]$.
To argue that the sum of expectations increases, the key part is to show that
$\EE\left[ G^{(j^*)}(\underline{X}^{(j^*)}) \right]$ increases by a constant.

\begin{lemma}\label{lem:fjNormIncrease}
Let $(\overline{X}^{\setminus j^*}, Y, Z)$ be a double sample from a
single-coordinate distribution $\mathcal{P}$.

Let $\underline{X}$ be a random vector, independent of this double sample
and distributed according to a single-step
distribution $(\underline{\Omega}, \underline{\pi})$ such that
$\pi$ is the $j^*$-th marginal distribution of $\mathcal{P}$.

Then, for every $i \in [n]$ and every function 
$f: \underline{\Omega} \to [0,1]$ we have
\begin{align}
\EE\left[
   \mathcal{M}[i,Y,Z]f(\underline{X}) \right]
\geq
\EE[f(\underline{X})] + \tau (1-\rho^2(\mathcal{P})) \; ,
\end{align}
where $\tau = \Inf_{i}(f(\underline{X}))$.
\end{lemma}
Recall that the distribution of $(Y,Z)$ depends on $j^*$.
We do not need to
consider the full multi-step process in this lemma, but when applying it later
we will set $\underline{X} = \underline{X}^{(j^*)}$ and $f = f^{(j^*)}$.

\begin{proof}
Fix a vector $\underline{x}_{\setminus i}$ for 
$\underline{X}_{\setminus i}$,
and define the function $h: \Omega \to [0,1]$ as
$h(x) := f(\underline{x}_{\setminus i}, x)$.
By Corollary~\ref{cor:varianceToEdgeVar},
\begin{align*}
\EE[|h(Y) - h(Z)|]
\geq
\EE[(h(Y) - h(Z))^2]
\geq
2(1-\rho^2) \Var[h(Y)] \; ,
\end{align*}
and hence, averaging over $\underline{X}_{\setminus i}$,
\begin{align}\label{eq:19a}
\EE\Bigl[\Bigl|f(\underline{X}_{\setminus i}, Y) - 
f(\underline{X}_{\setminus i}, Z) \Bigr|\Bigr] 
\geq 2(1-\rho^2) \Inf_i(f(\underline{X}_{\setminus i}, Y))
= 2\tau(1-\rho^2) \; .
\end{align}

Since $Y$ and $Z$ are symmetric (i.e., they define a reversible Markov chain,
cf.~remarks after Definition \ref{def:double-sample}) and by (\ref{eq:19a}),
\begin{align*}
\EE \left[ \left(\cM[i, Y, Z]f-f\right)(\underline{X}) \right]
&=
\EE \left[ \max(f(\underline{X}_{\setminus i}, Y), f(\underline{X}_{\setminus i}, Z))
- f(\underline{X}_{\setminus i}, Y) \right]
\\&= 
\frac12 \EE \left[ \left| f(\underline{X}_{\setminus i}, Y) - 
  f(\underline{X}_{\setminus i}, Z) \right| \right]
\ge
\tau(1-\rho^2) \; ,
\end{align*}
as claimed.
\end{proof}

\begin{lemma}\label{lem:averageNormIncrease}
Let a random vector $\overline{\underline{X}}$ be distributed according to
$(\overline{\underline{\Omega}}, \underline{\mathcal{P}})$ and functions
$f^{(1)}, \ldots, f^{(\ell)}: \underline{\Omega} \to [0, 1]$.
Let $i$, $j^*$ and $\tau$ be such that
$\Inf_{i}(f^{(j^*)}) \ge \tau \ge 0$ and let $\rho(\mathcal{P}) \le \rho \le 1$.

Pick a double sample $(\overline{X}^{\setminus j^*}, Y, Z)$ from $\mathcal{P}$ and let:
\begin{align}\label{eq:9a}
G^{(j)} := 
\begin{cases}
\cR[i,X^{(j)}]f^{(j)} & \text{if $j \neq j^*$}\\
\cM[i,Y,Z]f^{(j^*)} & \text{if $j = j^*$}.
\end{cases}
\end{align}
Then:
\begin{align}
\EE\left[\sum_{j=1}^{\ell} \EE[G^{(j)}(\underline{X}^{(j)}) \mid G^{(j)}] \right]
 \geq 
\sum_{j=1}^{\ell} \EE\left[ f^{(j)}(\underline{X}^{(j)})\right] + \tau \cdot (1-\rho^2)
\; .
\end{align}
\end{lemma}
Note that (\ref{eq:9a}) defines the functions $G^{(j)}$ as random variables
which is why we use capital letters.
\begin{proof}
If $j \neq j^*$ we have
\begin{align}
\EE \left[ \EE[G^{(j)}(\underline{X}^{(j)}) \mid G^{(j)}] \right]  = 
\EE[f^{(j)}(\underline{X}^{(j)})]\;,
\end{align}
since the marginal distribution of $\underline{X}^{(j)}$ is exactly as in the 
marginal $\pi$ of $\mathcal{P}$.
Hence, it suffices to show that
\begin{align*}
\EE\left[ \EE[G^{(j^*)}(\underline{X}^{(j^*)})
 \mid G^{(j^*)} ] \right]
&= 
\EE\left[ \cM[i,Y,Z]f^{(j^*)}(\underline{X}^{(j^*)}) 
 \right]
\\&\geq 
\EE[f^{(j^*)}(\underline{X}^{(j^*)})] + \tau(1-\rho^2)\;,
\end{align*}
but this is exactly Lemma~\ref{lem:fjNormIncrease}.
\end{proof}

\begin{lemma}
\label{lem:second-reduction-single-step}
Let $\overline{\underline{X}}$ be a random vector distributed according to
$(\overline{\underline{\Omega}}, \underline{\mathcal{P}})$ and also let
$f^{(1)},\ldots,f^{(\ell)}: \underline{\Omega} \to [0,1]$,
$i \in [n]$, $j^* \in [\ell]$, $\Inf_{i}(f^{(j^*)}) \geq \tau \ge 0$,
$\rho(\mathcal{P}) \le \rho \le 1$.

Then, there exist values 
$\overline{x}^{\setminus j^*} = (x^{(1)}, \ldots, x^{(j^*-1)}, x^{(j^*+1)}, \ldots,
x^{(\ell)}), y, z$
such that the functions
\begin{align}\label{eq:17a}
g^{(j)} := \begin{cases}
\cR[i,x^{(j)}]f^{(j)} & \text{if $j \neq j^*$}\\
\cM[i,y,z]f^{(j)} & \text{if $j = j^*$}
\end{cases}
\end{align}
satisfy
\begin{align}
\sum_{j=1}^{\ell} \EE[ g^{(j)}(\underline{X}^{(j)})]
&\geq 
\sum_{j=1}^{\ell} \EE[f^{(j)}(\underline{X}^{(j)})] +
\tau(1-\rho^2)/2
\; , \label{eq:15a}\\
\EE\left[\prod_{j=1}^{\ell} f^{(j)}(\underline{X}^{(j)})\right]
&\geq
\frac{\tau(1-\rho^2)}{2\ell|\Omega|^{\ell+1}} 
\cdot \EE\left[\prod_{j=1}^{\ell} g^{(j)}(\underline{X}^{(j)})\right]\label{eq:16a}
\; .
\end{align}
\end{lemma}
While (\ref{eq:15a}) is immediate from
Lemma~\ref{lem:averageNormIncrease}, we have to do a little bit of work 
to guarantee that it holds simultaneously with (\ref{eq:16a}).
\begin{proof}
Choose $(\overline{X}^{\setminus j^*}, Y, Z)$ as a double sample
from $\mathcal{P}$ and let $G^{(j)}$ be defined as in
(\ref{eq:9a}).

Define 
$p(\overline{x}^{\setminus j^*}, y, z) := 
\Pr[\overline{X}^{\setminus j^*} = \overline{x}^{\setminus j^*} \land
Y = y \land Z = z]$,
$\beta := \frac{\tau(1-\rho^2)}{2\ell|\Omega|^{\ell+1}}$,
an event $\cE :\equiv p(\overline{X}^{\setminus j^*}, Y, Z) < \beta$
and a random variable
\begin{align*}
A := \sum_{j=1}^\ell \EE[ G^{(j)}(\underline{X}^{(j)}) \mid G^{(j)} ] - 
\EE[ f^{(j)}(\underline{X}^{(j)})] \; .
\end{align*}
By Lemma \ref{lem:averageNormIncrease}, we have $\EE[A] \ge \tau(1-\rho^2)$.

Since there are $|\Omega|^{\ell+1}$ possible tuples
$(\overline{x}^{\setminus j}, y, z)$, by union bound we have
$\Pr[\cE] \le |\Omega|^{\ell+1}\beta = \tau(1-\rho^2)/2\ell$.
Bearing in mind the above and that
$A \in [-\ell, \ell]$,
\begin{align*}
\EE[A \cdot \1(\lnot \cE)]
\ge 
\EE[A] - \ell\Pr[\cE]
\ge
\tau (1-\rho^2)/2 \; .
\end{align*}

As a consequence, we can choose $(\overline{x}^{\setminus j^*}, y ,z)$
such that $A \ge \tau(1-\rho^2)/2$ and
$\cE$ does not happen. (\ref{eq:15a}) is now immediate, while for
(\ref{eq:16a}) observe that $\lnot \cE$ implies
$\mathcal{P}(\overline{X}^{\setminus j^*}, Y) \ge \beta$
and $\mathcal{P}(\overline{X}^{\setminus j^*}, Z) \ge \beta$
and apply Lemma
\ref{lem:maximumKeepsProbability}.
\end{proof}

We can now repeat the process from Lemma~\ref{lem:second-reduction-single-step}
multiple times to get the result of this section.

\begin{corollary}
\label{cor:low-influence-reduction}
Let $\overline{\underline{X}}$ be a random vector distributed according
to $(\overline{\underline{\Omega}}, \underline{\mathcal{P}})$ with 
$\rho(\mathcal{P}) \le \rho < 1$.
Then, for every $\tau > 0$ there exist $k \in \mathbb{N}$ and 
$\beta > 0$ such that:

For every $\epsilon \in [0, 1]$ and functions
$f^{(1)}, \ldots, f^{(\ell)}: \underline{\Omega} \to [0, 1]$ such that each
$f^{(j)}$ is $\epsilon$-resilient up to size $k$, there exist
$g^{(1)}, \ldots, g^{(\ell)}: \underline{\Omega} \to [0, 1]$
with the following properties:
\begin{enumerate}
\item
  $\max_{j \in [\ell]} \max_{i \in [n]} \Inf_i(g^{(j)}(\underline{X}^{(j)}))
  \leq \tau$.
\item $\EE \left[ \prod_{j=1}^{\ell} f^{(j)}(\underline{X}^{(j)}) \right] \ge 
\beta \cdot 
\EE \left[ \prod_{j=1}^{\ell} g^{(j)}(\underline{X}^{(j)}) \right]$.
\item For all $j \in [\ell]$:
  $\EE[g^{(j)}(\underline{X}^{(j)})] \geq 
  (1-\epsilon)\EE[f^{(j)}(\underline{X}^{(j)})]$.
\end{enumerate}
Furthermore, one can take $k := \lfloor \frac{2\ell}{\tau(1-\rho^2)} \rfloor$
and $\beta := \left(\frac{\tau(1-\rho^2)}{2\ell|\Omega|^{\ell+1}}\right)^k$.
\end{corollary}
In particular, both $k$ and $\beta$ depend only on $\tau$ and $\mathcal{P}$
(requiring $\tau > 0$ and $\rho(\mathcal{P}) < 1$).
\begin{proof}
We repeat the process from Lemma~\ref{lem:second-reduction-single-step},
always replacing the collection of functions $f^{(1)},\ldots,f^{(\ell)}$
with $g^{(1)},\ldots,g^{(\ell)}$ until condition $1$ is satisfied.
Since $\sum_{j=1}^{\ell} \EE[f^{(j)}(\underline{X}^{(j)})]$ cannot exceed $\ell$
and every time it increases by $\tau(1-\rho^2)/2$, we have to do this at most
$\frac{2 \ell}{ \tau (1-\rho^2)}$ times.

The first point is then obvious, and the second point follows
from Lemma~\ref{lem:second-reduction-single-step}.

Finally, the third point follows 
because the functions $f^{(j)}$ are all $\epsilon$-resilient up to size
$k$, and each of the functions $g^{(j)}$ can be written as 
a maximum of restrictions
of size at most $k$ of $f^{(j)}$. Since the maximum only increases expectations, the proof follows. 
\end{proof}

\subsection{Finishing the proof}
\begin{proof}[Proof of Theorem~\ref{thm:main-multiple}]
  Let us assume that $\mu \in (0, 0.99]$, the computations being only easier if
  this is not the case. To establish (\ref{eq:76a}), whenever we say
  ``constant'', in the $O()$ notation or otherwise, we mean ``depending only on
  $\mathcal{P}$ (in particular, on $\alpha$, $\rho$, $|\Omega|$ and $\ell$),
  but not on $\mu$''.

The proof consecutively applies Lemma \ref{lem:restrictToMakeStable},
Corollary \ref{cor:low-influence-reduction} and
Theorem \ref{thm:low-influence}.

Given $f: \underline{\Omega} \to [0, 1]$ with
$\EE[f(\underline{X}^{(1)})] = \mu$, first apply Lemma
\ref{lem:restrictToMakeStable} to $f$ with $\epsilon := 1/2$ and 
$k := \exp\left(\left(1/\mu\right)^D\right)$ for a constant $D$ large enough
(where ``large enough'' will depend on another constant $D'$ to be defined
later).
This gives us a function $g: \underline{\Omega} \to [0, 1]$ such that:
\begin{itemize}
\item $g$ is $\epsilon$-resilient up to size $k$.
\item $\EE[g(\underline{X}^{(1)})] \ge \mu$.
\item 
\begin{align}
\label{eq:77a}
\EE\left[ \prod_{j=1}^\ell f(\underline{X}^{(j)}) \right]
\ge
c \cdot \EE \left[ \prod_{j=1}^\ell g(\underline{X}^{(j)}) \right] \; ,
\end{align}
where:
\begin{IEEEeqnarray*}{rCl}
  c &=& 1/\exp\left( \left(1/\alpha\right)^{2k} \cdot 4\ln 1/\mu \right) \ge
  1/\exp\left( \exp\left(O\left(k\right)\right) \cdot 4\ln 1/\mu \right) 
  \\ &\ge&
  1/\exp\left(\exp\left(\exp\left(\left(1/\mu\right)^{O(1)}\right)\right) \cdot
    4\ln 1/\mu\right) 
  \\ &\ge& 
  1/\exp\left(\exp\left(\exp\left(\left(1/\mu\right)^{O(1)}\right)\right)\right) \; .
\end{IEEEeqnarray*}
\end{itemize}

Next, apply Corollary \ref{cor:low-influence-reduction}.
Set $g^{(1)} := \ldots := g^{(\ell)} := g$ and 
$\tau := 1/\exp\left( \left(1/\mu\right)^{D'}\right)$ for a constant $D'$ large 
enough.
We need to check if $k$ we have chosen satisfies the assumption
of Corollary \ref{cor:low-influence-reduction}:
\begin{align*}
\frac{2\ell}{\tau(1-\rho^2)}
\le
O\left(\exp\left( \left(1/\mu\right)^{D'}\right)\right)
\le
\exp\left( \left(1/\mu\right)^{O(1)}\right) \le k \; .
\end{align*}
Therefore, Corollary \ref{cor:low-influence-reduction} is applicable
and yields $h^{(1)}, \ldots, h^{(\ell)}: \underline{\Omega} \to [0, 1]$
such that:
\begin{itemize}
\item
  $\max_{j \in [\ell]} \max_{i \in [n]} \Inf_i(h^{(j)}(\underline{X}^{(j)})) \le
  \tau$.
\item
  $\forall j \in [\ell]: \EE[h^{(j)}(\underline{X}^{(j)})] \ge \mu/2$.
\item
\begin{align}
\label{eq:78a}
\EE\left[ \prod_{j=1}^\ell g(\underline{X}^{(j)}) \right] \ge \beta \cdot \EE
  \left[ \prod_{j=1}^\ell h^{(j)}(\underline{X}^{(j)}) \right] \; ,
\end{align}
where:
\begin{IEEEeqnarray*}{rCl}
  \beta &=& 
  \left(\frac{\tau(1-\rho^2)}{2\ell|\Omega|^{\ell+1}}\right)^k
  \ge
  1/O\left(\exp\left(\left(1/\mu\right)^{D'}\right)\right)^k
  \\ &\ge&
  1/\exp\left(\left(1/\mu\right)^{O(1)} \cdot k\right)
  \ge 1/\exp\left(\exp\left(\left(1/\mu\right)^{O(1)}\right)\right) \; .
\end{IEEEeqnarray*}
\end{itemize}

Finally, we need to apply Theorem \ref{thm:low-influence}. To this end, set:
\begin{align*}
\epsilon := \left(\mu/2\right)^{\ell^2/(1-\rho^2)}/2 \ge \mu^{O(1)}
\end{align*}
and verify (\ref{eq:38a}):
\begin{IEEEeqnarray*}{rCl}
  \left(\frac{(1-\rho^2)\epsilon}{\ell^{5/2}}
  \right)^{O\left(\frac{\ln(\ell/\epsilon) \ln(1/\alpha)}{(1-\rho)\epsilon}\right)}
  &\ge&
  \Omega\left(\epsilon\right)^{\left(O(1) + \ln 1/\epsilon\right)
    \cdot O\left(1/\epsilon\right)}
  \\ &\ge&
  1/\exp\left( \left(O(1) + \ln 1/\epsilon \right)^2 \cdot 
    O\left(1/\epsilon\right)  \right)
  \\ &\ge&
  1/\exp\left( \left(1/\epsilon\right)^{O(1)} \right)
    \ge 1/\exp\left(\left(1/\mu\right)^{O(1)}\right) \; .
\end{IEEEeqnarray*}
Hence, from Theorem \ref{thm:low-influence}:
\begin{align}
\label{eq:79a}
  \EE\left[ \prod_{j=1}^\ell h^{(j)}(\underline{X}^{(j)}) \right]
  \ge \epsilon / 2 \ge \mu^{O(1)} \; .
\end{align}
(\ref{eq:77a}), (\ref{eq:78a}) and (\ref{eq:79a}) put together give:
\begin{align*}
  \EE \left[ \prod_{j=1}^\ell f(\underline{X}^{(j)}) \right]
  \ge c \cdot \beta \cdot \mu^{O(1)}
  \ge 1/\exp\left(\exp\left(\exp\left(\left(\left(1/\mu\right)^{O(1)}
  \right)\right)\right)\right)
  \; ,
\end{align*}
as claimed.
\end{proof}

\section{Proof for Two Steps}
\label{sec:two-steps}

Our goal in this section is to prove Theorem \ref{thm:main-two-variables}
assuming Theorem \ref{thm:main-multiple}.

In the following  we will sometimes drop the assumption that $\Omega$ is 
necessarily the support of a probability distribution $\mathcal{P}$. 
One can check that this will not cause problems.

\subsection{Correlation of a cycle}

Assume we are given a support set $\Omega$ of size
$|\Omega| = k$. Let $s \ge 2, p \in (0, 1)$ and let $(x_0, \ldots, x_{s-1})$
be a sequence of distinct $x_i \in \Omega$.

\begin{definition}
We call a probability distribution $\mathcal{C}$ over $\Omega$
an \emph{$(s, p)$-cycle} if
\begin{align*}
\mathcal{C}(x, y) = \begin{cases}
p / s & \text{if } x=y=x_i \text{ for } i \in \{0, \ldots, s-1\} \; , \\
(1-p) / s &\text{if } x=x_i \land y=x_{(i+1) \bmod s} \text{ for } i \in \{0, \ldots, s-1\} \; ,\\
0 & \text{otherwise.}
\end{cases}
\end{align*}
\end{definition}

\begin{lemma}
\label{lem:cycle-rho}
Let $\mathcal{C}$ be an $(s, p)$-cycle. Then
\begin{align*}
  \rho(\mathcal{C}) \le 1-\frac{7p(1-p)}{s^2} \; .
\end{align*}
\end{lemma}

\begin{proof}
Let $K$ be the Markov kernel induced by a double sample on $\mathcal{C}$
($K$ is the same whether a sample is on the first or the second step,
cf.~Section \ref{sec:correlation-properties}). Observe that
\begin{align*}
K(y, z) := \begin{cases}
  p^2 + (1-p)^2 & \text{if $y=z=x_i$,}\\
  p(1-p) & \text{if $y=x_i$ and $z=x_{(i\pm 1) \bmod s}$.}
\end{cases}
\end{align*}
Let $\alpha_k := \frac{2\pi k}{s}$.
One can check that the eigenvalues
of $K$ are $\lambda_0, \ldots, \lambda_{s-1}$
with $\lambda_k := 1-2p(1-p)(1-\cos \alpha_k)$. 
This is easiest if one knows the respective (complex) eigenvectors 
$v_k := (1, \exp(\alpha_k \imath), \ldots, \exp((s-1)\alpha_k \imath))$ 
(where $\imath$ is the imaginary unit).

Using $\cos x \le 1-x^2/5$ for $x \in [0, \pi]$
and $\sqrt{1-x} \le 1-x/2$ for $x \in [0, 1]$ we obtain that if $k > 0$,
then
\begin{align*}
  \sqrt{\lambda_k} \le \sqrt{1-2p(1-p)(1-\cos \alpha_1)}
 \le 
 \sqrt{1-2p(1-p)\frac{4\pi^2}{5s^2}}
 \le 1 - \frac{7p(1-p)}{s^2} \; .
\end{align*}

The bound on $\rho(\mathcal{C})$ now follows from Lemma \ref{lem:lambda-rho}.
\end{proof}

\subsection{Convex decomposition of 
\texorpdfstring{$\mathcal{P}$}{P}}

In this section we show that if a distribution $\mathcal{P}$ 
can be decomposed into a convex combination
of distributions $\mathcal{P} = \sum_{k=1}^{r} \alpha_k \cP_k$ and each 
distribution $\cP_k$ is same-set hitting, 
then also $\mathcal{P}$ is same-set hitting.

\begin{definition}
We say that a probability distribution with equal marginals $\cP$ has an
\emph{$(\alpha, \rho)$}-convex decomposition if there exist 
$\beta_1, \ldots, \beta_r > 0$ with $\sum_{k=1}^r \beta_k = 1$
and distributions with equal marginals
$\cP_1, \ldots, \cP_r$ such that
\begin{align*}
  \cP = \sum_{k=1}^r \beta_k \cdot \cP_k \; .
\end{align*}
and $\alpha(\cP_k) \ge \alpha$ and $\rho(\cP_k) \le \rho$ for every $k \in [r]$.
\end{definition}

\begin{lemma}
\label{lem:convex-decomposition}
Let an $\ell$-step distribution $\mathcal{P}$ with equal marginals have an
$(\alpha, \rho)$-convex decomposition
for some $\alpha > 0$ and $\rho < 1$. 

Then, for every function $f: \underline{\Omega} \to [0, 1]$ with 
$\EE[f(\underline{X}^{(1)})] = \mu > 0$:
\begin{align*}
  \EE \left[ \prod_{j=1}^\ell f(\underline{X}^{(j)}) \right] \ge 
c(\alpha, \rho, \ell, \mu) > 0 \; .
\end{align*}
\end{lemma}

\begin{proof}
Let us write the relevant decomposition as $\cP = \sum_{k=1}^r \beta_k \cP_k$.
The existence of this decomposition implies that there exists a random vector
$\underline{Z} = (Z_1, \ldots, Z_n)$ such that:
\begin{itemize}
  \item The variables $Z_i \in [r]$ are i.i.d.~with $\Pr[Z_i = k] = \beta_k$.
  \item For every $i \in [n]$ and $k \in [r]$, conditioned on $Z_i = k$,
    the tuple $\overline{X}_i$ is distributed according to $\cP_k$.
\end{itemize}

Let $\underline{z}$ be an arbitrary assignment to $\underline{Z}$ and let
$\mu_{\underline{z}} := \EE[f(\underline{X}^{(1)}) \mid \underline{Z}=\underline{z}]$.
If $\mu_{\underline{z}} \ge \mu / 2$, by Theorem \ref{thm:main-multiple}\footnote{
  Technically, Theorem \ref{thm:main-multiple} requires the distributions to
be the same for each coordinate, which is not the case in our setting. However,
this is not a problem, cf.~Section \ref{sec:equal-distributions}.
}
\begin{align}
\label{eq:86a}
  \EE \left[ \prod_{j=1}^\ell f(\underline{X}^{(j)}) \mid 
  \underline{Z}=\underline{z} \right] \ge
  c(\alpha, \rho, \ell, \mu) > 0 \; .
\end{align}

Since $\EE[\mu_{\underline{Z}}] = 
\EE[\EE[f(\underline{X}^{(1)}) \mid \underline{Z}]] = 
\EE[f(\underline{X}^{(1)})] = \mu$, by Markov
\begin{align}
\label{eq:87a}
	\Pr \left[ \mu_{\underline{Z}} \ge \mu / 2 \right] \ge \mu/2 \; .
\end{align}
(\ref{eq:86a}) and (\ref{eq:87a}) together give
\begin{align*}
  \EE \left[ \prod_{j=1}^\ell f(\underline{X}^{(j)}) \right]
  \ge \mu/2 \cdot c(\alpha, \rho, \ell, \mu) > 0 \; .
\end{align*} 
\end{proof}

\subsection{Decomposition of 
\texorpdfstring{$\cP$}{P} 
into cycles}

\begin{definition}
Let us consider weighted directed graphs with non-negative weights 
over a vertex set $\Omega$. We will identify such a digraph $G$
with its weight matrix.

We say that such a weighted digraph is \emph{regular}, if for every vertex
the total weight of the incoming edges is equal to the total weight of
the outgoing edges.

We call a weighted digraph a \emph{weighted cycle}, if it is a directed cycle 
over a subset of $\Omega$ with all edges of the same weight $w > 0$.
We call $w$ the \emph{weight} of the cycle and number of its edges $s$
the \emph{size} of the cycle.

We say that a weighted digraph $G$ can be 
\emph{decomposed into $r$ weighted cycles}
if there exist weighted cycles $C_1, \ldots, C_r$  such that 
$G = \sum_{k=1}^r C_k$.
\end{definition}

\begin{lemma}
\label{lem:digraph-decomposition}
Every regular weighted digraph $G$ over a set $\Omega$ of size $k$ 
can be decomposed into at most $k^2$ weighted cycles.
\end{lemma}
\begin{proof}
Since the digraph is regular, it must have a cycle. Remove it from the graph
(taking as weight $w$ the minimum weight of the edge on this cycle).

Since the resulting graph is still regular, 
proceed by induction until the graph is empty.

At each step at least one edge is completely removed from the graph,
therefore there will be at most $k^2$ steps.
\end{proof}

To see that a two-step distribution $\cP$ can be decomposed into cycles, 
it will be useful to take $\cP' := \cP - \alpha \cdot \Id$ and
look at it as a weighted directed graph $(\Omega, \cP')$,
where $\cP'$ is interpreted as a weight function 
$\cP': \Omega \times \Omega \to \mathbb{R}_{\ge 0}$.

\begin{lemma}
\label{lem:cycle-decomposition}
Let $\cP$ be a two-step distribution 
with equal marginals 
over an alphabet $\Omega$ with size $t$.

Then, $\cP$ has a convex decomposition 
$\cP = \sum_{k=1}^r \beta_k \cP_k$ such that each $\cP_k$ either
has support of size $1$ or is
an $(s, p)$-cycle with $2 \le s \le t$ and $p \in [\alpha(\cP)^3, 1/2]$. 

Consequently, $\cP$ has an $(\alpha, \rho)$-convex decomposition with
$\alpha := \alpha(\cP)^4$ and $\rho := 1-3\alpha(\cP)^5$.
\end{lemma}

\begin{proof}
Throughout this proof we will treat $\mathcal{P}$ as a weight matrix of
a digraph. Since $\cP$ has equal marginals, this weighted digraph is regular.
Use Lemma \ref{lem:digraph-decomposition} to decompose
$\cP - \alpha(\cP) \cdot \Id$ into weighted cycles, which allows us to write
\begin{align*}
	\cP = \alpha(\cP) \cdot \Id + \sum_{k=1}^{r} C_k \; ,
\end{align*}
where $C_k$ is a weighted cycle with weight $w_k$ and size $s_k$
and $r \le t^2$.
Take $\beta_k := \min(w_k, \alpha(\cP) / t^2)$ and let
$\Id_k$ be the identity matrix restricted to the support of $C_k$. Now we
can write $\cP$ as
\begin{align*}
\cP = \left( \alpha(\cP) \cdot \Id - \sum_{k=1}^{r} \beta_k \Id_k \right) + 
	\left( \sum_{k=1}^{r} s_k (w_k + \beta_k) \cdot 
  \frac{\beta_k \Id_k + C_k}{s_k(w_k + \beta_k)} \right) \; .
\end{align*}

Firstly, $\left( \alpha(\cP) \cdot \Id - \sum_{k=1}^{r} \beta_k \Id_k \right)$
can be decomposed into distributions with support size $1$.

As for the other term, note that 
$\mathcal{C}_k := \frac{\beta_k \Id_k + C_k}{s_k(w_k + \beta_k)}$ is a probability
distribution that
either has support of size $1$ (iff $C_k$ has
support of size $1$) or is an $(s, p)$-cycle with 
$2 \le s \le t$ and $p = \beta_k/(\beta_k+w_k)$.

If $\beta_k = w_k$, then $p  = 1/2$. If $\beta_k < w_k$, then
$1/2 \ge p = \beta_k / (\beta_k+w_k) \ge \beta_k = \alpha(\cP) / t^2 \ge 
\alpha(\cP)^3$.
Therefore, $p \in [\alpha(\cP)^3, 1/2]$, as stated.

Consequently, $\alpha(\mathcal{C}_k) = p/s_k \ge \alpha(\cP)^4$
and, by Lemma \ref{lem:cycle-rho}, 
$\rho(\mathcal{C}_k) \ge 1-3\alpha(\cP)^5$ and, since
every $(s,p)$-cycle has equal marginals,
we obtained
an $(\alpha, \rho)$-convex decomposition of $\cP$.
\end{proof}

\subsection{Putting things together}

\begin{proof}[Proof of Theorem \ref{thm:main-two-variables}]
From Lemmas \ref{lem:cycle-decomposition} and
\ref{lem:convex-decomposition}.
\end{proof}

\begin{remark}
One can see that see that, as in Theorem \ref{thm:main-multiple}, we obtain
a triply exponential explicit bound, i.e, there exists 
$D(\alpha(\cP)) > 0$ such that if $\mu \in (0, 0.99]$, then
\begin{align*}
\EE\left[ f(\underline{X})f(\underline{Y}) \right]
  \ge 1/\exp\left(\exp\left(\exp\left( \left(1/\mu\right)^D\right) \right)\right)\;.
\end{align*}
\end{remark}

\section{Local Variance}
\label{sec:local-variance}

In this section we state and prove a generalization of the low-influence theorem 
from \cite{Mos10}. We assume that the reader is familiar with
Fourier coefficients $\hat{f}(\sigma)$ and the basics of discrete function 
analysis, for details see, e.g., Chapter 8 of \cite{Dol14}. 

\cite{Mos10} shows that $\rho(\cP) < 1$ implies that $\cP$ is set hitting
for low-influence functions. We extend this result to
a weaker notion of influence. In particular, we show that $\cP$ is set hitting
for functions with $\Omega(1)$ measure and $o(1)$ largest Fourier coefficient.
The main result of this section is Theorem~\ref{thm:local-variance}.

We remark that Theorem~\ref{thm:local-variance} 
does \emph{not} require equal marginals.
The rest of this section contains the proof of Theorem~\ref{thm:local-variance}.
First, from Corollary \ref{cor:low-influence-reduction} and
Theorem \ref{thm:low-influence} it is easy to establish\footnote{
One needs to check that the assumption about equal marginals is not necessary,
but that turns out to be the case (the bound in Theorem \ref{thm:low-influence}
then depends on $\min_{j \in [\ell], x \in \supp(X^{(j)})} \pi^{(j)}(x)$).
} the following:
\begin{theorem}\label{thm:stabilityExpectation}
Let $\overline{\underline{X}}$ be a random vector distributed according to
an $\ell$-step distribution $\cP$ with $\rho(\cP) \le \rho < 1$ and
let $\epsilon \in [0, 1)$.

Then, for
all $\mu^{(1)}, \ldots, \mu^{(\ell)} \in (0, 1]$ there exists
$k(\cP, \epsilon, \mu^{(1)}, \ldots, \mu^{(\ell)}) \in \mathbb{N}$ 
such that for all functions
$f^{(1)}, \ldots, f^{(\ell)}: \underline{\Omega} \to [0, 1]$,
if $\EE[f^{(j)}(\underline{X}^{(j)})] = \mu^{(j)}$ and
if $f^{(1)}, \ldots, f^{(\ell)}$ are all $\epsilon$-resilient up to size
$k$, then
\begin{align}
\EE\left[\prod_{j=1}^{\ell} f^{(j)}(\underline{X}^{(j)}) \right]
\geq c(\cP, \epsilon, \mu^{(1)}, \ldots, \mu^{(\ell)}) > 0 \; .
\end{align}
\end{theorem}

\begin{definition}
Let $\pi$ be a single-step distribution and let 
$f: \underline{\Omega} \to \bbR$ be a 
function. Let $S \subseteq [n]$ with $|S| = k$.
We define $f^{\subseteq S}: \underline{\Omega} \to \bbR$ as
\begin{align}
f^{\subseteq S}(\underline{x}) :=
\EE[f(\underline{x}_S,\underline{X}_{\overline S})] \; ,
\end{align}
where $\overline S := [n] \setminus S$, $\underline{x}_{S}$
is the vector $\underline{x}$ restricted to coordinates in $S$,
and $\underline{X}_{\overline{S}}$ is a random vector of
$n-k$ elements with each coordinate distributed i.i.d.~in $\pi$.
\end{definition}

A proof of the following claim can be found, e.g., in~\cite{Dol14}:
\begin{claim}
\label{cl:fourier-vs-variance}
Let $\pi$ be a single-step distribution and let 
$f: \underline{\Omega} \to \bbR$, $S \subseteq [n]$.
If a random vector $\underline{X}$ 
is distributed according to $\underline{\pi}$
and $\phi_0, \ldots, \phi_{m-1}$ form a Fourier basis for $\pi$ and
$f = \sum_{\sigma \in \mathbb{N}^n_{< m}} 
\hat{f}(\sigma) \phi_{\sigma}$, then
$f^{\subseteq S} = \sum_{\sigma: \supp(\sigma) \subseteq S}
\hat{f}(\sigma) \phi_{\sigma}$.
In particular,
\begin{align*}
\Var\left[f^{\subseteq S}(\underline{X})\right] = 
\sum_{\substack{\sigma: \supp(\sigma) \subseteq S,\\ 
\sigma \ne 0^n}} 
\left|\hat{f}(\sigma)\right|^2 \; .
\end{align*}
\end{claim}

\begin{lemma}\label{lem:localInfluenceToStability}
Let a random vector $\underline{X}$ be distributed according to a single-step
distribution $\underline{\pi}$ with 
$\min_{x \in \Omega} \pi(x) \ge \alpha$ and
let $\epsilon \in [0, 1]$, $k \in \mathbb{N}$.

Then, for every $f: \underline{\Omega} \to \mathbb{R}_{\ge 0}$
with $\EE[f(\underline{X})] = \mu$, if for every
$S \subseteq[n]$ with $|S| = k$ it holds that
\begin{align*}
  \Var\left[f^{\subseteq S}(\underline{X})\right] \le \alpha^{k}(\epsilon\mu)^2 \; ,
\end{align*}
then $f$ is $\epsilon$-resilient up to size $k$.
\end{lemma}
\begin{proof}
We prove the contraposition. 

If $f$ is not $\epsilon$-resilient up to size $k$, by definition
of $f^{\subseteq S}$ it implies that there exist $S \subseteq [n]$ with
$|S| = k$ and $\underline{x}$ such that
\begin{align*}
  \left| f^{\subseteq S}(\underline{x}) -
  \EE[f^{\subseteq S}(\underline{X})] \right| >
  \epsilon \EE[f^{\subseteq S}(\underline{X})]
  = \epsilon\mu \; .
\end{align*}
But this gives
\begin{align*}
\Var\left[ f^{\subseteq S}(\underline{X}) \right] 
\ge \alpha^k 
\left(f^{\subseteq S}(\underline{x}) - \EE[f^{\subseteq S}(\underline{X})] 
\right)^2
> \alpha^{k} (\epsilon\mu)^2 \; ,
\end{align*}
as required.
\end{proof}

Using Lemma~\ref{lem:localInfluenceToStability} we can weaken
the assumption in Theorem~\ref{thm:stabilityExpectation}
such that it only requires that all Fourier coefficients of degree at most $k$
are small:

\begin{proof}[Proof of Theorem~\ref{thm:local-variance}]
From Theorem \ref{thm:stabilityExpectation}, there exists
$k := k(\cP, \mu^{(1)}, \ldots, \mu^{(\ell)})$ such that if
$f^{(1)}, \ldots, f^{(\ell)}$ are all $1/2$-resilient up to size $k$,
then (\ref{eq:89a}) holds. Therefore, it is sufficient to show that
the functions $f^{(j)}$ are indeed $1/2$-resilient up to size $k$
if the parameter $\gamma$ is chosen small enough.

By Claim~\ref{cl:fourier-vs-variance}, if $\max_{\sigma: 0<| \sigma|\le k} 
|\hat{f}^{(j)}(\sigma)| \le \gamma$, then for any
$S \subseteq [n]$ with $|S| = k$ we have
$\Var\left[ (f^{(j)})^{\subseteq S}(\underline{X}^{(j)})\right]
\le |\Omega|^k \gamma^2$. With that in mind it is easy to choose $\gamma$
such that Lemma 
\ref{lem:localInfluenceToStability} 
can be applied to each $f^{(j)}$.
\end{proof}

\section{Multiple Steps of a Markov Chain}
\label{sec:markov}

Next, we consider the case where the distribution $\cP$
is such that the random variables
$X^{(1)},X^{(2)},\ldots,X^{(\ell)}$ form a Markov chain.

\begin{definition}
Let $\cP$ be a an $\ell$-step distribution with equal marginals and let
$\overline{X} = (X^{(1)}, \ldots, X^{(\ell)})$ be a random variable distributed
according to $\cP$. We say that
\emph{$\mathcal{P}$ is generated by Markov chains}\footnote{
Note that our definition allows for different Markov chains
in different steps.
} if for every $j \in \{2, \ldots, \ell\}$ and
$x^{(1)}, \ldots, x^{(j)} \in \Omega$ we have
\begin{IEEEeqnarray*}{rCl}
\IEEEeqnarraymulticol{3}{l}{
  \Pr[{X^{(j)} = x^{(j)}}  | { X^{(1)} = x^{(1)}} \land 
  \dots \land {X^{(j-1)} = x^{(j-1)}}]
}
\\ \qquad &=& 
\Pr[X^{(j)} = x^{(j)} | X^{(j-1)} = x^{(j-1)}] \; .
\end{IEEEeqnarray*}
\end{definition}

Observe that since we still require $\mathcal{P}$ to have equal marginals,
the mar\-ginal $\pi$ is then simply a stationary distribution of the chain.

In this case, we give a reduction to Theorem~\ref{thm:main-two-variables}
to prove a bound that does not depend on $\rho(\cP)$:

\begin{restatable}{theorem}{mainmarkov}
\label{thm:main-markov}
Let $\Omega$ be a finite set and $\cP$ a probability distribution over 
$\Omega^{\ell}$ with equal marginals
generated by Markov chains.
Let tuples $\overline{X}_i = (X_i^{(1)}, \ldots, X_i^{(\ell)})$ be 
i.i.d.~according to $\mathcal{P}$ for $i \in \{1,\ldots,n\}$.

Then, for every $f: \Omega^n \to [0,1]$ with 
$\EE[f(\underline{X}^{(1)})] = \mu > 0$:
\begin{align}
\EE \left[ \prod_{j=1}^\ell f(\underline{X}^{(j)}) \right]
\geq c \left( \alpha(\mathcal{P}), \ell, \mu \right) 
\; ,
\end{align}
where the function $c()$ is positive whenever 
$\alpha(\mathcal{P})> 0$.
\end{restatable}

\begin{proof}
Let $\cP$ be a distribution generated by Markov chains
with $\alpha := \alpha(\cP) > 0$ and let $f: \underline{\Omega} \to [0, 1]$
with $\EE[f(\underline{X}^{(1)})] = \mu > 0$.

The proof is by induction on $\ell$. For $\ell = 2$, apply 
Theorem \ref{thm:main-two-variables} directly. For $\ell > 2$, define
the function $g: \underline{\Omega} \to [0, 1]$ as
\begin{align*}
g(\underline{x}) := 
\EE \left[ f(\underline{X}^{(\ell-1)})f(\underline{X}^{(\ell)}) \mid 
  \underline{X}^{(\ell-1)} = \underline{x} \right]
  = f(\underline{x}) \cdot \EE \left[ f(\underline{X}^{(\ell)}) \mid  
  \underline{X}^{(\ell-1)} = \underline{x} \right] \; . 
\end{align*}
Applying Theorem \ref{thm:main-two-variables} for the distribution
of the last two steps,
\begin{align}
\label{eq:88a}
  \EE[g(\underline{X}^{(1)})] = 
  \EE[g(\underline{X}^{(\ell-1)})] =
  \EE[ f(\underline{X}^{(\ell-1)}) f(\underline{X}^{(\ell)}) ] \ge 
  c(\alpha, \mu) > 0 \; .
\end{align}

Now we have
\begin{IEEEeqnarray}{rCl}
\EE\left[ \prod_{j=1}^{\ell} f(\underline{X}^{(j)}) \right]
	& = &
	\EE \left[ \left( \prod_{j=1}^{\ell-2} f(\underline{X}^{(j)}) \right) 
          g(\underline{X}^{(\ell-1)}) \right] 
		\label{eq:07a} \\
	& \ge &
	\EE \left[ \prod_{j=1}^{\ell-1} g(\underline{X}^{(j)}) \right] 
		\label{eq:08a} \\
	& \ge &
	c\left(\alpha, \ell-1, c(\alpha, \mu) \right) = c(\alpha, \ell, \mu) > 0,
		\label{eq:09a}
\end{IEEEeqnarray}
where (\ref{eq:07a}) holds since $\cP$ is generated by Markov chains,
(\ref{eq:08a}) is due to $f \ge g$ pointwise and 
(\ref{eq:09a}) is an application of the induction and (\ref{eq:88a}).
\end{proof}

\begin{remark}
Unfortunately, this proof worsens the explicit bound. One can check that
for a Markov-generated distribution with $\ell$ steps the dependence on $\mu$
is a tower of exponentials of height $3(\ell-1)$.
\end{remark}

\section{Polynomial Same-Set Hitting}
\label{sec:polynomial-hitting}

The property of set hitting establishes a lower bound on
$\EE\left[ \prod_{j=1}^\ell f^{(j)}(\underline{X}^{(j)})\right]$ that is
independent of $n$. However, it might be the case that this bound is very
small, perhaps far from the best possible one. In particular, our bound
from Theorem \ref{thm:main-multiple} is triply exponentially small, and
the bound from Theorem \ref{thm:progressions} is not even primitive recursive.

\begin{definition}
A distribution $\mathcal{P}$ is \emph{polynomially set hitting}
(resp.~polynomially same-set hitting) if there exists $C \ge 0$
such that $\mathcal{P}$ is $(\mu, \mu^C)$-set hitting (resp.~same-set hitting)
for every $\mu \in (0, 1]$.
\end{definition}

As a matter of fact, \cite{MOS13} 
(cf.~Theorem \ref{thm:different-sets-classification}) establishes that
all distributions that are set hitting are also polynomially set hitting.
We suspect that this is also the case for two-step same-set hitting,
but this remains an open problem.

However, it is possible to harness reverse hypercontractivity to show
that all \emph{symmetric} two-step distributions are polynomially same-set
hitting:
\begin{theorem}
\label{thm:symmetric}
Let a two-step probability distribution with equal marginals 
$\mathcal{P}$ be symmetric, i.e.,
$\mathcal{P}(x, y) = \mathcal{P}(y, x)$ for all $x, y \in \Omega$.
If $\alpha(\mathcal{P}) > 0$, then $\mathcal{P}$ is polynomially same-set
hitting.
\end{theorem}

We omit the proof of Theorem \ref{thm:symmetric}, 
noting that the idea is similar as in Section \ref{sec:two-steps}:
one performs an obvious convex decomposition of $\mathcal{P}$ into cycles
of length two and applies
the result of \cite{MOS13} to each term of this decomposition.

\appendix 
\section{Appendix: Proof of Theorem \ref{thm:low-influence}}
\label{sec:low-influence-proof}

Our proof of Theorem \ref{thm:low-influence} follows in this appendix. 
It is only a slight adaptation of the argument from \cite{Mos10}, but
we include it in full for the sake of completeness.

We first restate the theorem and discuss the differences between
our proof and the one in \cite{Mos10}:
\lowinfluence*

Theorem \ref{thm:low-influence} is very similar to a subcase of Theorem 1.14
from \cite{Mos10}.  We make a stronger claim with one respect: in
\cite{Mos10} the influence threshold $\tau$ depends among others on:
\begin{align}
\label{eq:74a}
  \alpha^* := \min_{(x^{(1)}, \ldots, x^{(\ell)}) \in \supp(\mathcal{P})} 
  \mathcal{P}(x^{(1)}, \ldots, x^{(\ell)}) \; ,
\end{align}
while our bound depends only on the smallest marginal probability:
\begin{align}
\label{eq:75a}
\alpha = \min_{x \in \Omega} \pi(x) \; .
\end{align}

The main differences to the proof in \cite{Mos10} are:
\begin{itemize}
\item \cite{Mos10} proves the base case $\ell=2$ and then obtains the result
for general $\ell$ by an inductive argument 
(cf., Theorem 6.3 and Proposition 6.4 in \cite{Mos10}). Since the induction
is applied to functions $f^{(1)}$ and $g := \prod_{j=2}^{\ell} f^{(j)}$,
where $g$ is viewed as a function on a single-step space, the information 
on the smallest marginal is lost in the case of $g$. To avoid this, our proof
proceeds directly for general $\ell$. However, the structure 
and the main ideas are really the same as in \cite{Mos10}.
\item In Section~\ref{sec:hyper}, in hypercontractivity bounds for Gaussian
  and discrete spaces
(Theorem~\ref{thm:hypercontractivity-degree} and 
Lemma~\ref{lem:hypercontractivity-technical}) we are slightly more careful
to obtain bounds which depend on $\alpha$ rather than $\alpha^*$
(as defined in (\ref{eq:75a}) and (\ref{eq:74a})). This better
bound is then propagated in the proof of the invariance principle.
\item Another change is not related to the dependency on the smallest 
  marginal. In Section~\ref{sec:gaussian-hyper},
  in the Gaussian reverse hypercontractivity bound
(Theorem \ref{thm:gaussian-hypercontractivity-main}) instead of
using the result of Borell (\cite{Bor85}, Theorem 5.1 in \cite{Mos10}) 
for a bound expressed in terms of the cdf
of bivariate Gaussians, we utilize the results of \cite{CDP13} and \cite{Led14}
 for a more convenient bound
of the form $\left(\prod_{j=1}^\ell \mu^{(j)}\right)^{c(\rho, \ell)}$.
\end{itemize}

The proof can be generalized in several directions, but for the sake
of clarity we present the simplest version sufficient for our purposes.

\subsection{Preliminaries --- the general framework}

We start with explaining the notation of random variables and $L^2$ spaces 
that we will use throughout the proof.

\begin{definition}
\label{def:l2}
Let $(\Omega, \mathcal{F}, \mathcal{P})$ be a probability space.
We define the real inner product space $L^2(\Omega, \mathcal{P})$
as the set of all square-integrable functions 
$f: \Omega \to \mathbb{R}$, i.e., the functions that satisfy
\begin{align}
  \int_{\Omega} f^2 \, \mathrm{d}\mathcal{P} < +\infty \; ,
\label{eq:44a}
\end{align}
with inner product defined as
\begin{align}
  \langle f, g \rangle := \int_{\Omega} fg \, \mathrm{d}\mathcal{P} \; .
\label{eq:45a}
\end{align}
\end{definition}

\begin{remark}
As we will see shortly,
if $X$ is a random variable sampled from $\Omega$ according to
$\mathcal{P}$,
the equations (\ref{eq:44a}) and (\ref{eq:45a}) can be written as
\begin{IEEEeqnarray*}{rCl}
\EE[f^2(X)] & < & +\infty \; , \\
\langle f, g \rangle & = & \EE[f(X) g(X)] \; .
\end{IEEEeqnarray*}
\end{remark}

\begin{remark}
We omitted the event space $\mathcal{F}$ in the definition
of $L^2(\Omega, \mathcal{P})$. This is because $\mathcal{F}$ is always
implicit in the choice of the measure~$\mathcal{P}$. 

In particular,  when $\mathcal{P}$ is discrete,
of course we choose $\mathcal{F}$ to be the powerset
of $\Omega$.
When $\mathcal{P}$ is continuous over $\mathbb{R}^n$,
we use the ``standard'' real event space, i.e.,~the 
completion of the Borel algebra.  
\end{remark}

While this will not be our usual way of thinking,
at this point it makes sense to introduce the formal
definition of a random variable: a function from a probability
space to some set.
\begin{definition}
Let $(\Sigma, \mathcal{F}, \mathcal{P})$ be a probability space. 
We say that $X$ is a random variable over
a set  $\Sigma'$
if it is a measurable function $X: \Sigma \to \Sigma'$.
\end{definition}
As usual, we will assume throughout the
proof that all random variables are induced by some underlying probability 
space $(\Sigma,\mathcal{F},\mathcal{P})$.

Using this,
a random variable induces some distribution, which we can study.

\begin{definition}
We say that a random variable $X$ over a set $\Omega$ is \emph{distributed
according to a probability space $(\Omega, \mathcal{P})$} 
if for every event $A \in \mathcal{F}$:
\begin{align*}
  \Pr[X \in A] = \mathcal{P}(A) \; .
\end{align*}
\end{definition}

\begin{definition}
\label{def:l2-rv}
Let $X$ be a random variable distributed over $\Omega$.
By $L^2(X)$ we denote the inner product space of random variables that
correspond to square-integrable functions $f: \Omega \to \mathbb{R}$:
\begin{align*}
  L^2(X) := \{ Z \mid Z = f \circ X \text{ for some } f:\Omega \to \mathbb{R}  
  \text{ with } \EE[f(X)^2] < +\infty  \} \; ,
\end{align*}
with the inner product given as
\begin{align*}
  \langle Z_1, Z_2 \rangle := \EE[ Z_1 \cdot Z_2 ] \; .
\end{align*}
\end{definition}

\begin{remark}
We consider the formal setting again, i.e., 
suppose $(\Sigma, \mathcal{F}, \mathcal{P})$ is the underlying probability 
space, and $X: \Sigma \to \Omega$ a random variable. 
Then, $L^{2}(X)$ is a subspace of $L^2(\Sigma,\mathcal{P})$.
Intuitively, it contains all real valued functions 
which ``depend only on $X$''.
\end{remark}

\begin{example}
Fix $(\Omega, \mathcal{P})$ to be the uniform distribution on 
$\Omega := \{0, 1, 2\}$ and let $X$ be distributed according
to $(\Omega, \mathcal{P})$. Then $L^2(X)$ has dimension three
and one of its orthonormal bases is
\begin{IEEEeqnarray*}{rCl}
Z_0 & :\equiv & 1 \\
Z_1 &:=& \begin{cases}
  \sqrt{6}/2 & \text{if $X=0$,}\\
  -\sqrt{6}/2 & \text{if $X=1$,}\\
  0 & \text{if $X=2$.}
\end{cases}\\
Z_2 &:=& \begin{cases}
  \sqrt{2}/2 & \text{if $X \in \{0,1\}$,}\\
  -\sqrt{2} & \text{if $X=2$.}\\
\end{cases}
\end{IEEEeqnarray*}  
\end{example}

After this point, we will have no need to refer explicitly 
to the underlying probability space $(\Sigma,\mathcal{F},\mathcal{P})$ anymore.
Nevertheless, it will be useful to remember that random variables are 
functions of this underlying space.

It immediately follows from the definitions that:

\begin{lemma}
Let $X$ be a random variable distributed according to $(\Omega, \mathcal{P})$.
Then $L^2(X)$ is isomorphic to $L^2(\Omega, \mathcal{P})$.
\end{lemma}

\subsection{Preliminaries --- orthonormal ensembles and multilinear polynomials}

In this section we introduce orthonormal ensembles
and multilinear polynomials over them.

\begin{definition}
We call a finite family $(\mathcal{X}_0, \ldots, \mathcal{X}_p)$ 
of random variables \emph{orthonormal}
if they satisfy $\EE[\mathcal{X}_k^2] = 1$ for every $k$
and $\EE[\mathcal{X}_j \mathcal{X}_k ] = 0$ for every $j \ne k$.
\end{definition}

\begin{definition}
We call a finite family of orthonormal random variables
$\mathcal{X} = (\mathcal{X}_{\star,0} = 1, \mathcal{X}_{\star,1}, 
\ldots, \mathcal{X}_{\star,p})$
an $\emph{orthonormal ensemble}$. 
We call $p$ the \emph{size} of the ensemble.

An $\emph{ensemble sequence}$ is a sequence of independent families of 
random variables
$\underline{\mathcal{X}} = (\mathcal{X}_1, \ldots, \mathcal{X}_n)$
such that each $\mathcal{X}_i$ is an orthonormal ensemble
$\mathcal{X}_i = (\mathcal{X}_{i,0} = 1, \mathcal{X}_{i,1}, 
\ldots, \mathcal{X}_{i,p})$ of the same size $p$.
We call $n$ the \emph{size} of the sequence. 
\end{definition}
The notation $\mathcal{X}_{\star, k}$ is a little awkward, but we do not
need to use it often.
The reason for it is that we want to to make sure that one cannot confuse 
one of the random variables $\mathcal{X}_{\star,k}$ within an orthonormal
ensemble with the orthonormal ensemble $\mathcal{X}_i$ itself.
Whenever a random variable $\mathcal{X}_{i,k}$ is part of an ensemble
$\mathcal{X}_i$, 
there is no reason to use the $\star$-symbol.
Instead we use the index of the ensemble.

Note that in an orthonormal ensemble for $k>0$ we have $\EE[\mathcal{X}_{\star,k}]
= \EE[\mathcal{X}_{\star,k}\mathcal{X}_{\star,0}] = 0$.

\begin{definition}
We call two ensemble sequences 
$\underline{\mathcal{X}} = (\mathcal{X}_1, \ldots, \mathcal{X}_n)$
and $\underline{\mathcal{Y}} = (\mathcal{Y}_1, \ldots, \mathcal{Y}_m)$
\emph{compatible} if $n = m$ and the sizes of the individual ensembles
$\mathcal{X}_i$ and $\mathcal{Y}_i$ are the same.
\end{definition}

\begin{definition}
Let $\underline{\mathcal{X}} = (\mathcal{X}_1, \ldots, \mathcal{X}_n)$
be an ensemble sequence such that each ensemble $\mathcal{X}_i$ is of size $p$.

A \emph{monomial compatible with $\underline{\mathcal{X}}$} is a term
\begin{align*}
  x_\sigma := \prod_{i=1}^n x_{i,\sigma_i} \; ,
\end{align*}
where $\sigma = (\sigma_1, \ldots, \sigma_n)$ with 
$\sigma_i \in \{0, \ldots, p\}$.

A (formal) \emph{multilinear polynomial
compatible with $\underline{\mathcal{X}}$} is
a sum of compatible monomials, i.e., a polynomial $P$ of the form
\begin{align*}
P(\underline{x}) =
	\sum_{\sigma \in \{0,\ldots,p\}^n} \alpha(\sigma) x_{\sigma} =
	\sum_{\sigma \in \{0,\ldots,p\}^n}
		\alpha(\sigma) \prod_{i=1}^n x_{i, \sigma_i} \; ,
\end{align*}
where the sum goes over all tuples $\sigma = (\sigma_1,\ldots,\sigma_n)$ 
as above, and  $\alpha(\sigma) \in \mathbb{R}$.

For a tuple $\sigma$ we define its \emph{support}
as $\supp(\sigma) := \{ i \in [n]: \sigma_i \ne 0 \}$
and its
 \emph{degree} as the size of its support:
$|\sigma| := |\supp(\sigma)|$.
Also, we will write the tuple $(0, \ldots, 0)$ as $0^n$.
\end{definition}

Let a multilinear polynomial $P$ compatible with $\underline{\mathcal{X}}$
be given. Then, $P(\underline{\mathcal{X}})$ is what one expects:
the random variable obtained by evaluating the polynomial on the given input.
Analogously, if $\sigma$ is a tuple as above we write $\mathcal{X}_\sigma$ for 
the random variable corresponding to the evaluation of the monomial $x_{\sigma}$.

\begin{lemma}\label{lem:orthonormality}
Let $\underline{\mathcal{X}}$ be an ensemble sequence and $\sigma$, $\tau$ 
two tuples whose monomials $x_{\sigma}$, $x_{\tau}$ are compatible with 
$\underline{\mathcal{X}}$.
Then, 
\begin{align}\label{eq:68a}
\EE[{\mathcal{X}}_{\sigma}{\mathcal{X}}_\tau]
= \begin{cases}
1 & \text{if $\sigma = \tau$}\\
0 & \text{otherwise}
\end{cases}
\end{align}
and
\begin{align}\label{eq:69a}
\EE[{\mathcal{X}}_{\sigma}]
= \begin{cases}
1 & \text{if $\sigma = 0^n$}\\
0 & \text{otherwise.}
\end{cases}
\end{align}
\end{lemma}
\begin{proof}
By independence of the coordinates we have
$\EE[{\mathcal{X}}_{\sigma}{\mathcal{X}}_\tau] =
\prod_{i=1}^n \EE [\mathcal{X}_{i,\sigma_i} \cdot \mathcal{X}_{i, \tau_i}]$
and now we can use the orthonomality of each ensemble $\mathcal{X}_i$.
For the second part, we apply the first on $\tau=0^n$.
\end{proof}

\begin{definition}
Given a multilinear polynomial 
$P(\underline{x}) = \sum_{\sigma}\alpha(\sigma)x_{\sigma}$
we define its following properties:
\begin{IEEEeqnarray}{rCl}
\deg(P) & := & \begin{cases}
\max_{\sigma: \alpha_\sigma \ne 0} |\sigma| \label{eq:27a} & \text{if $P$ is non-zero}\\
-\infty & \text{if $P$ is the zero polynomial} \\
\end{cases}\\
\EE[P] & := & \alpha(0^n) \label{eq:28a} \\
\EE[P^2] & := & \sum_{\sigma} \alpha(\sigma)^2 \label{eq:29a}  \\
\Var[P] & := & \EE[P^2] - {\EE}^2[P] \label{eq:30a} \\
\Inf_i(P) & := & \sum_{\sigma: \sigma_i \ne 0} \alpha(\sigma)^2 \label{eq:31a} \\
\Inf(P) & := & \sum_{i=1}^n \Inf_i(P) \label{eq:32a} 
\end{IEEEeqnarray}
\end{definition}

The next lemma states that the formal expressions defined above
are consistent with the corresponding probabilistic interpretations
for every ensemble sequence.

\begin{lemma}
\label{lem:poly-expectation-equiv}
For an ensemble sequence $\underline{\mathcal{X}}$
and a multilinear polynomial $P$ compatible with it we have
\begin{IEEEeqnarray}{rCl}
\EE[P] & = & \EE[P(\underline{\mathcal{X}})] \label{eq:14a} \\
\EE[P^2] & = & \EE[(P(\underline{\mathcal{X}}))^2] \label{eq:33a} \\ 
\Var[P] & = & \Var[P(\underline{\mathcal{X}})] \label{eq:34a} \; .
\end{IEEEeqnarray}
Furthermore, if all random variables in $\underline{\mathcal{X}}$
are discrete, then
\begin{align}
\Inf_i(P)  =  \EE \left[ \Var\left[ P(\underline{\mathcal{X}})  \mid 
	\mathcal{X}_1, \ldots, \mathcal{X}_{i-1}, 
        \mathcal{X}_{i+1}, \ldots, \mathcal{X}_n \right] \right] \; .
    \label{eq:35a}
\end{align}
\end{lemma}

\begin{proof}
Linearity of expectation
and (\ref{eq:69a}) yield
$\EE[P(\underline{\mathcal{X}})] = 
\sum_{\sigma} \alpha(\sigma) \EE[\mathcal{X}_{\sigma}] = \alpha(0^n)$, which
is (\ref{eq:14a}).
Next, (\ref{eq:68a}) gives
$\EE[P^2(\underline{\mathcal{X}})] = 
\sum_{\sigma,\tau} \alpha(\sigma)\alpha(\tau) \mathcal{X}_\sigma\mathcal{X}_\tau
= \sum_{\sigma} \alpha(\sigma)^2$, 
i.e.~(\ref{eq:33a}), and hence
(\ref{eq:34a}) by the definition of the variance.

As for (\ref{eq:35a}), fix an assignment
$\underline{x}_{\setminus i} = (x_1, \ldots, x_{i-1}, x_{i+1}, \ldots, x_n)$
to the ensemble sequence 
$\underline{\mathcal{X}}_{\setminus i} = 
(\mathcal{X}_1, \ldots, \mathcal{X}_{i-1}, \mathcal{X}_{i+1}, \ldots, 
\mathcal{X}_n)$.\footnote{
Note that each entry in this tuple is itself a tuple:
$x_{i} = (x_{i,0} = 1, x_{i,1}, \ldots, x_{i,p})$,
where $p$ is the size of the ensemble. 
}
We suppose that this tuple has a non-zero probability of occurence.
Since $\mathcal{X}_i$ is an orthornormal ensemble,
\begin{IEEEeqnarray*}{rCl}
  \Var[P(\underline{\mathcal{X}}) \mid 
  \underline{\mathcal{X}}_{\setminus i} = \underline{x}_{\setminus i} ]
& = & \sum_{k = 1}^p \left( \sum_{\sigma: \sigma_i = k}
    \alpha(\sigma) \cdot \prod_{j \ne i} x_{j, \sigma_j}
    \right)^2
\end{IEEEeqnarray*}
From Lemma~\ref{lem:orthonormality}, for a fixed $k \in \{1, \ldots, p\}$,
\begin{align*}
\EE \left[ \left(
  \sum_{\sigma: \sigma_i = k} \alpha(\sigma) \cdot \prod_{j \ne i} 
  \mathcal{X}_{j,\sigma_j}
\right)^2 \right] = \sum_{\sigma: \sigma_i = k} \alpha(\sigma)^2 \; .
\end{align*}
Together this gives
\begin{align*}
\EE\left[ \Var\left[ P(\underline{\mathcal{X}}) \mid 
\underline{\mathcal{X}}_{\setminus i} \right] \right]
= \sum_{\sigma: \sigma_i \ne 0} \alpha(\sigma)^2 \; ,
\end{align*}
as claimed.
\end{proof}

\begin{definition}
For a multilinear polynomial $P(\underline{x}) = 
\sum_{\sigma} \alpha(\sigma)x_\sigma$ and
$S \subseteq [n]$ we let $P_S$ be $P$
restricted to tuples $\sigma$ with $\supp(\sigma) = S$,
i.e., $P_S := \sum_{\sigma:\supp(\sigma)=S} \alpha(\sigma)x_\sigma$.

Then, let $P^{>d} := \sum_{S: |S| > d} P_S$ be $P$
restricted to tuples with the degree greater than $d$. 
We also define $P^{=d}$, $P^{\le d}$ etc.~in the analogous way.
\end{definition}

\begin{lemma}
\label{lem:orthogonal-decomposition}
Let $P$ and $Q$ be multilinear polynomials compatible with an ensemble
sequence $\underline{\mathcal{X}}$. Then,
\begin{align*}
  \EE\left[ P(\underline{\mathcal{X}}) Q(\underline{\mathcal{X}}) \right]
  = \sum_{S \subseteq [n]} \EE\left[ P_S(\underline{\mathcal{X}})
  Q_S(\underline{\mathcal{X}}) \right] \; .
\end{align*}
\end{lemma}

\begin{proof}
It is enough to show that for $S \ne T$
\begin{align*}
  \EE\left[ P_S(\underline{\mathcal{X}}) 
  Q_T(\underline{\mathcal{X}}) \right] = 0 \;.
\end{align*}
Let $P(\underline{\mathcal{X}}) = 
\sum_\sigma \alpha(\sigma) \cdot \mathcal{X}_\sigma$
and $Q(\underline{\mathcal{X}}) = \sum_\sigma \beta(\sigma) \cdot
\mathcal{X}_\sigma$.
Assume w.l.o.g.~that there exists $i^* \in S\setminus T$.
Then,
\begin{IEEEeqnarray*}{l}
\EE\left[ P_S(\underline{\mathcal{X}}) Q_T(\underline{\mathcal{X}}) \right] = 
  \\
\qquad = \sum_{\substack{\sigma: \supp(\sigma) = S\\ \sigma': \supp(\sigma') = T}}
\alpha(\sigma) \beta(\sigma') \EE \left[ \mathcal{X}_{i^*, \sigma_{i^*}} \right]
\EE \left[ \prod_{i \ne i^*} \mathcal{X}_{i,\sigma_i} \mathcal{X}_{i,\sigma'_i}  
\right] = 0 \; .
\end{IEEEeqnarray*}
\end{proof}

\begin{corollary}
\label{cor:orthogonal-esquared}
Let $P$ be a multilinear polynomial.
Then, $\EE[P^2] \allowbreak = \allowbreak \sum_{S \subseteq [n]}
\allowbreak \EE[P_S^2]$.
\end{corollary}
\begin{proof}
Taking any ensemble sequence $\underline{\mathcal{X}}$ compatible with $P$,
\begin{align*}
\EE[P^2] = \EE[P(\underline{\mathcal{X}})^2] = \sum_{S\subseteq[n]}
\EE[P_S(\underline{\mathcal{X}})^2] = \sum_{S\subseteq [n]} \EE[P_S^2] \; .
& \qedhere
\end{align*}
\end{proof}

\begin{claim}
\label{cl:orthogonal-variance}
Let $P$ be a multilinear polynomial. Then, 
$\Var[P] = \sum_{S \subseteq[n]} \Var[P_S]$.
\end{claim}
\begin{proof}
Observing that $\Var[P_\emptyset] = 0$, $\EE[P_\emptyset^2] = \alpha(0^n)^2$
and $\Var[P_S] = \EE[P_S^2]$ for $S \ne \emptyset$, by Corollary
\ref{cor:orthogonal-esquared}
\begin{align*}
\Var[P] = \EE[P^2] - \alpha(0^n)^2 = \sum_{S\subseteq[n], S \ne \emptyset}
\EE[P_S^2] = \sum_{S\subseteq[n]} \Var[P_S] \; . & \qedhere
\end{align*}
\end{proof}

\begin{lemma}
\label{lem:influence-vs-variance}
Let $P$ be a multilinear polynomial with $\deg(P) \le d$. Then,
\begin{align*}
\Inf(P) \le d \cdot \Var[P] \; .
\end{align*}
\end{lemma}
\begin{proof}
\begin{align*}
\Inf(P) &= \sum_{\sigma} |\sigma| \cdot \alpha(\sigma)^2 
\le d \cdot \sum_{\sigma \ne 0^n} \alpha(\sigma)^2 = d \cdot \Var[P] \; .
\qedhere
\end{align*}
\end{proof}

\begin{definition}
\label{def:t-rho}
Let $\rho \in \mathbb{R}$. We define the operator $T_\rho$ as follows:
let $P(\underline{x}) = \sum_{\sigma} \alpha(\sigma) x_\sigma$ be a multilinear 
polynomial. Then,
\begin{align*}
(T_\rho P)(\underline{x}) := 
\sum_{\sigma} \rho^{|\sigma|} \alpha(\sigma) x_\sigma \; .
\end{align*}
\end{definition}

We will mostly use the operator $T_\rho$ with $\rho \in [0, 1]$.

\begin{definition}
\label{def:gaussian-ensemble}
We call an orthonormal ensemble $\mathcal{G}_\star$ of size $p$
\emph{Gaussian} if random variables 
$\mathcal{G}_{\star,1}, \ldots, \mathcal{G}_{\star, p}$
are independent $\mathcal{N}(0, 1)$ Gaussians.

We say that an ensemble sequence
$\underline{\mathcal{G}} = (\mathcal{G}_1, \ldots, \mathcal{G}_n)$ 
is Gaussian if for each $i \in [n]$ the ensemble $\mathcal{G}_i$
is Gaussian.
\end{definition}

We remark than as in all ensemble sequences, in a Gaussian ensemble sequence
we have $\mathcal{G}_{i,0} \equiv 1$ for all $i$.

\begin{definition}
For tuples of multilinear polynomials $\overline{P} = (P^{(1)}, \ldots, \allowbreak P^{(\ell)})$
such that each polynomial $P^{(j)}$ is 
compatible with an ensemble sequence $\underline{\mathcal{X}}$ we write
$\overline{P}(\underline{\mathcal{X}})$ for the tuple 
$(P^{(1)}(\underline{\mathcal{X}}), \ldots, P^{(\ell)}(\underline{\mathcal{X}}))$.

Similarly, given multilinear polynomials 
$\overline{P} = (P^{(1)}, \ldots, P^{(\ell)})$ and a collection of ensemble
sequences
$\overline{\underline{\mathcal{X}}} = (\underline{\mathcal{X}}^{(1)}, \ldots
\underline{\mathcal{X}}^{(\ell)})$ such that
$P^{(j)}$ is compatible with $\underline{\mathcal{X}}^{(j)}$ we write
$\overline{P}(\overline{\underline{\mathcal{X}}})$ for
$(P^{(1)}(\underline{\mathcal{X}}^{(1)}), \ldots, 
P^{(\ell)}(\underline{\mathcal{X}}^{(\ell)}))$.
\end{definition}

\subsection{Preliminaries --- ensemble collections}
\label{sec:app-pre-2}

In this section we recall the setting of Theorem \ref{thm:low-influence}
and introduce some other concepts we will need throughout the proof.

From now on we will always implicitly assume that all multi-step distributions 
$\cP$ have equal marginals (denoted as $\pi$). 
This assumption is not necessary, but sufficient for our main purpose,
while making the notation easier.

\begin{definition}
Let $X$ be a random variable distributed according to a 
single-step, single-coordinate distribution $(\Omega, \pi)$.
We say that an orthonormal ensemble $\mathcal{X}_\star$ is
\emph{constructed from $X$} if the elements
of $\mathcal{X}_\star$ form an orthonormal basis of $L^2(X)$.

Similarly, let $\underline{X}$ be a random vector distributed
according to $(\underline{\Omega}, \underline{\pi})$.
We say that an ensemble sequence $\underline{\mathcal{X}} = 
(\mathcal{X}_1, \ldots, \mathcal{X}_n)$ is constructed 
from~$\underline{X}$ if for each $i \in [n]$ 
the ensemble $\mathcal{X}_i$ is constructed from $X_i$.
\end{definition}

The definition of ensemble sequences requires that $\mathcal{X}_{i,0}
\equiv 1$ for every $i$; of course we can find a basis 
of $L^2({X}_i)$ which satisfies this requirement, so that 
ensemble sequences constructed from $\underline{X}$ indeed exist.

\begin{lemma}
\label{lem:product-base}
Let $\underline{\mathcal{X}}$ be an ensemble sequence constructed
from a random vector 
${\underline{X}}$ distributed according to
$({\underline{\Omega}}, \underline{\pi})$.
Assume that the size of each ensemble $\mathcal{X}_i$ is $p$.
Then the set of monomials
\begin{align*}
  \underline{\mathcal{B}} :=
  \left\{ \mathcal{X}_\sigma \mid \sigma = (\sigma_1, \ldots, \sigma_n),
  \sigma_i \in \{0, \ldots, p\} \right\}
\end{align*}
is an orthonormal basis of $L^2({\underline{X}})$.
\end{lemma}

\begin{proof}
Observe that the dimension of $L^2({X}_i)$ is $p+1$, (note that
it is the support size of the single-coordinate distribution 
$(\Omega, \pi)$).
Hence, the dimension of $L^2(\underline{X})$ is $(p+1)^n$,
which equals the size of $\underline{\mathcal{B}}$.
Therefore, it is enough to check that $\underline{\mathcal{B}}$ is
orthonormal, which is done in Lemma~\ref{lem:orthonormality}. 
\end{proof}

\begin{definition}
Let $\underline{\mathcal{X}}$ be an ensemble sequence constructed from a 
random vector $\underline{X}$ distributed according to 
$(\underline{\Omega}, \underline{\pi})$.

For a function $f: \underline{\Omega} \to \mathbb{R}$
and a multilinear polynomial $P$ compatible with
$\underline{\mathcal{X}}$
we say that $f(\underline{X})$
is \emph{equivalent} to $P$ if it always holds that
\begin{align*}
  f(\underline{X}) = P(\underline{\mathcal{X}}) \; .
\end{align*}
\end{definition}

\medskip

Recall the operator $T_\rho$ from Definition \ref{def:t-rho}.
We show that it has a natural counterpart in 
$L^2({\underline{\Omega}}, \underline{\pi})$.
\begin{definition}
\label{def:t-rho-function}
Let $\rho \in [0, 1]$ and let
$({\underline{\Omega}}, \underline{\pi})$
be a single-step probability space (with
$({\Omega}, \pi)$ a corresponding 
single-coordinate probability space). 

We define a linear operator 
$T_\rho:  L^2({\underline{\Omega}}, \underline{\pi}) \to 
	L^2({\underline{\Omega}}, \underline{\pi})$ as
\begin{align*}
	T_\rho f({\underline{x}}) := \EE \left[ f\left( 
  {\underline{Y}}^{\rho, {\underline{x}}} \right) \right] \; ,
\end{align*}
where ${\underline{Y}}^{\rho, {\underline{x}}} = 
({Y}^{\rho, {\underline{x}}}_1, \ldots, 
{Y}^{\rho, {\underline{x}}}_n)$ 
is a random vector with independent coordinates
distributed such that ${Y}^{\rho, {\underline{x}}}_i = 
{x}_i$ 
with probability $\rho$ 
and ${Y}^{\rho, {\underline{x}}}_i$ is 
(independently)
distributed according to $({\Omega}, \pi)$ with probability 
$(1-\rho)$.
\end{definition}

The next lemma states that taking operator $T_\rho$ preserves
the equivalence of functions and polynomials:
\begin{lemma}
Let $\underline{\mathcal{X}}$ be an ensemble sequence constructed from
a random vector ${\underline{X}}$ distributed according
to $({\underline{\Omega}}, \underline{\pi})$.

Let $\rho \in [0, 1]$, $f: {\underline{\Omega}} \to \mathbb{R}$
and $P$ be a multilinear polynomial equivalent to $f$.
Then, $T_\rho P$ and $T_\rho f$ are equivalent, i.e.,
\begin{align*}
  T_\rho f ({\underline{X}})
  = T_\rho P (\underline{\mathcal{X}}) \; .
\end{align*}
\end{lemma}

\begin{proof}
Fix an input ${\underline{x}} \in {\underline{\Omega}}$
in the support of $\underline{\mathcal{P}}$.
Let 
$\underline{\mathcal{Y}}^{\rho, {\underline{x}}} = 
(\mathcal{Y}_1^{\rho, {\underline{x}}}, \ldots, 
\mathcal{Y}_n^{\rho, {\underline{x}}})$ 
be the random sequence where for each coordinate $i \in [n]$, independently
\begin{align*}
  \mathcal{Y}_i^{\rho, {\underline{x}}} := \begin{cases}
    \mathcal{X}_i({x_i}) & \text{with probability $\rho$,}\\
    \text{a random ensemble distributed as $\mathcal{X}_i$} 
    & \text{with probability $1-\rho$.}
  \end{cases}
\end{align*}
Note that $\underline{\mathcal{Y}}^{\rho, {\underline{x}}}$ is not an 
ensemble sequence, but this will not cause problems. 

Writing $P(\underline{x}) = 
\sum_\sigma \alpha(\sigma) \cdot x_\sigma$ we can calculate
\begin{IEEEeqnarray*}{rCl}
  T_\rho f ({\underline{x}}) & = & \EE [ 
  f(\underline{{Y}}^{\rho, {\underline{x}}}) ]
  = \EE [ P(\underline{\mathcal{Y}}^{\rho, {\underline{x}}})]
  = \sum_\sigma \alpha(\sigma)
    \EE[ \mathcal{Y}_\sigma^{\rho, {\underline{x}}} ] \\
  & = & \sum_\sigma \rho^{|\sigma|} \alpha(\sigma) \cdot 
    \mathcal{X}_\sigma({\underline{x}})
  = T_\rho P ({\underline{x}}) \; .
\end{IEEEeqnarray*}
Since ${\underline{x}}$ was arbitrary, the claim is proved.
\end{proof}

Recall Definition \ref{def:gaussian-ensemble}.
In the proof we will construct a tuple of ensemble sequences
$\overline{\underline{\mathcal{X}}} = 
(\underline{\mathcal{X}}^{(1)}, \ldots, \underline{\mathcal{X}}^{(\ell)})$
from a random vector $\overline{\underline{X}}$
and consider relations between those sequences
and compatible Gaussian ensemble sequences. 
To this end, we need to introduce the Gaussian equivalent of
marginal ensemble sequences $\underline{\mathcal{X}}^{(j)}$.

\begin{definition}
\label{def:vg}
Let $\mathcal{G_\star} = (\mathcal{G}_{\star,0}, \ldots, \mathcal{G}_{\star, p})$ be a Gaussian 
orthonormal ensemble of size $p$. We define
an inner product space $V(\mathcal{G})$ as
\begin{align*}
  V({\mathcal{G}}) := \left\{
    \sum_{k=0}^p \alpha_k \cdot \mathcal{G}_k \mid 
    \alpha_0, \ldots, \alpha_k \in \mathbb{R}
  \right\}
\end{align*}
with the inner product of $A, B \in V(\mathcal{G})$ given by
$\langle A, B \rangle := \EE[ A \cdot B ]$.

Similarly, given a Gaussian ensemble sequence 
$\underline{\mathcal{G}}$ such that each of its ensembles is of size $p$
we let
\begin{align*}
  V(\underline{\mathcal{G}}) := \left\{
    \sum_{\sigma} \alpha(\sigma) \cdot \mathcal{G}_\sigma
    \mid \sigma = (\sigma_1, \ldots, \sigma_n) \in \{0, \ldots, p\}, 
    \alpha(\sigma) \in \mathbb{R}
  \right\} \; ,
\end{align*}
with the inner product $\langle A, B \rangle := \EE[ A \cdot B ]$.
\end{definition}

\begin{lemma}\label{lem:cov-xg}
Let a random tuple $\overline{X} = 
(X^{(1)},\ldots,X^{(\ell)})$ be distributed according
to a single-coordinate distribution 
$(\overline{\Omega}, \mathcal{P})$.
Let $\overline{\mathcal{X}}_\star = 
(\mathcal{X}_\star^{(1)},\ldots,\mathcal{X}_\star^{(\ell)})$ be such
that $\mathcal{X}_\star^{(j)}$ is an orthonormal ensemble constructed 
from $X^{(j)}$.

Then, there exist Gaussian orthonormal ensembles
$\overline{\mathcal{G}}_\star = 
(\mathcal{G}_\star^{(1)},\ldots,\mathcal{G}_\star^{(\ell)})$ 
compatible with $\overline{\mathcal{X}}_\star$ such
that for all $j_1,j_2 \in [\ell]$, and all $k_1,k_2 \geq 0$ we have
\begin{align}\label{eq:70a}
  \Cov\left[ \mathcal{X}_{\star, k_1}^{(j_1)}, \mathcal{X}_{\star, k_2}^{(j_2)} 
  \right]
  = \Cov\left[ 
  \mathcal{G}_{\star, k_1}^{(j_1)}, \mathcal{G}_{\star, k_2}^{(j_2)} \right] \; .
\end{align}
\end{lemma}
\begin{proof}
Consider
$(\overline{\Omega}, \mathcal{P})$ 
as a single-step probability space, and let $\overline{X}$ 
be the corresponding random variable.
Let now $\mathcal{Z}_\star$ be an orthonormal ensemble constructed
from $\overline{X}$. Recall that this means that
the elements of $\mathcal{Z}_\star$ form an orthonormal basis
of $L^2(\overline{X})$.

Let $\mathcal{H}_\star$ be a Gaussian ensemble sequence compatible
with $\mathcal{Z}_\star$.
Define the map
$\Psi: L^2(\overline{X}) \to V(\mathcal{H}_\star)$ by linearly extending
$\Psi(\mathcal{Z}_{\star,k}) := \mathcal{H}_{\star,k}$.
In this way $\Psi$ becomes an isomorphism between $L^2(\overline{X})$
and $V(\mathcal{H}_\star)$ (and as such it preserves inner products).

Since $L^2(X^{(j)})$ is a subspace of $L^2(\overline{X})$, we can define
$\mathcal{G}_{\star,k}^{(j)}$ as $\mathcal{G}_{\star,k}^{(j)} := 
\Psi(\mathcal{X}_{\star,k}^{(j)})$.
Since $\Psi$ preserves inner products we get (\ref{eq:70a}).

We still need to argue that for each $j \in [\ell]$ the
orthonormal ensemble $\mathcal{G}_\star^{(j)}$ is Gaussian.
The fact that $\mathcal{G}_\star^{(j)}$ is an ensemble sequence
follows from (\ref{eq:70a}) for $j_1 = j_2 = j$ (note that $\Psi(1) = 1$).

The variables $\mathcal{G}_{\star,k}^{(j)}$ are clearly jointly Gaussian, 
since they can be written as sums of independent Gaussians.
By (\ref{eq:70a}), their covariance matrix is identity.
This finishes the proof, since joint Gaussians with 
the identity covariance matrix must be independent.
\end{proof}

Since the proof of Lemma \ref{lem:cov-xg} is somewhat abstract, 
we illustrate the construction of $\overline{\mathcal{G}}_\star$ 
with an example.

\begin{example}
Consider $(X^{(1)}, X^{(2)})$ distributed according to $\cP$ 
over $\Omega = \{0,1\}$ with
$\cP(0, 0) = \cP(1, 1) = 1/8$ and $\cP(0, 1) = \cP(1, 0) = 3/8$. We can take
the following for the ensemble $\mathcal{Z}_\star$:

\begin{tabular}{|l||c|c|c|c|}
\hline
$(X^{(1)}, X^{(2)}) :=$ & (0,0) & (0, 1) & (1, 0) & (1, 1) 
\\ \hline \hline
$\mathcal{Z}_{\star, 0}$ & 1 & 1 & 1 & 1
\\ \hline
$\mathcal{Z}_{\star, 1}$ & 2 & 0 & 0 & -2
\\ \hline
$\mathcal{Z}_{\star, 2}$ & 0 & $2\sqrt{3}/3$ & $-2\sqrt{3}/3$ & 0
\\ \hline
$\mathcal{Z}_{\star, 3}$ & $\sqrt{3}$ & $-\sqrt{3}/3$ & $-\sqrt{3}/3$ 
                                                  & $\sqrt{3}$
\\ \hline
\end{tabular}

For the marginal ensemble $\mathcal{X}_\star^{(1)}$ we can take

\begin{tabular}{|l||c|c|}
\hline
$X^{(1)} :=$ & 0 & 1 
\\ \hline \hline
$\mathcal{X}_{\star, 0}^{(1)}$ & 1 & 1
\\ \hline
$\mathcal{X}_{\star, 1}^{(1)}$ & $1$ & $-1$
\\ \hline
\end{tabular}

Now one can check that $\mathcal{X}^{(1)}_{\star, 0} = \mathcal{Z}_{\star, 0}$
and 
$\mathcal{X}^{(1)}_{\star, 1} = 1/2 \cdot \mathcal{Z}_{\star, 1}
+\sqrt{3}/2 \cdot \mathcal{Z}_{\star,2}$. 
Defining the ensemble $\mathcal{X}_{\star}^{(2)}$ in the same way we get
$\mathcal{X}^{(2)}_{\star,0} = \mathcal{Z}_{\star,0}$ and
$\mathcal{X}^{(2)}_{\star,1} = 1/2 \cdot \mathcal{Z}_{\star, 1} -\sqrt{3}/2 \cdot
\mathcal{Z}_{\star, 2}$.

Let $\mathcal{H}_{\star} = 
(\mathcal{H}_{\star, 0}\equiv 1, \mathcal{H}_{\star, 1}, \mathcal{H}_{\star, 2}, 
\mathcal{H}_{\star, 3})$ be a Gaussian ensemble sequence compatible with
$\mathcal{Z}$. One easily checks that our construction gives
\begin{align*}
\mathcal{G}_{\star, 0}^{(1)} = \mathcal{G}_{\star, 0}^{(2)} & =  \mathcal{H}_{\star, 0}
\\
\mathcal{G}^{(1)}_{\star, 1} & = 
1/2 \cdot \mathcal{H}_{\star, 1} +\sqrt{3}/2 \cdot \mathcal{H}_{\star,2}
\\
\mathcal{G}^{(2)}_{\star,1} & = 
1/2 \cdot \mathcal{H}_{\star, 1} -\sqrt{3}/2 \cdot \mathcal{H}_{\star, 2} \; .
\end{align*}
\end{example}

Since the covariances between independent coordinates are always zero,
Lemma \ref{lem:cov-xg} applied to each coordinate separately gives:
\begin{corollary}
\label{cor:cov-xg}
Let a random vector $\overline{\underline{X}} = 
(\underline{X}^{(1)},\ldots,\underline{X}^{(\ell)})$ be distributed according
to a distribution 
$(\underline{\overline{\Omega}}, \underline{\mathcal{P}})$.
Let $\underline{\overline{\mathcal{X}}} = 
(\underline{\mathcal{X}}^{(1)},\ldots,\underline{\mathcal{X}}^{(\ell)})$ be such
that $\underline{\mathcal{X}}^{(j)}$ is an ensemble sequence constructed 
from $\underline{X}^{(j)}$.

Then, there exist Gaussian ensemble sequences
$\underline{\overline{\mathcal{G}}} = 
(\underline{\mathcal{G}}^{(1)},\ldots,\underline{\mathcal{G}}^{(\ell)})$ 
compatible with $\overline{\underline{\mathcal{X}}}$ such
that for all $i_1, i_2 \in [n]$, $j_1,j_2 \in [\ell]$, and all $k_1,k_2 \geq 0$ we have
\begin{align}\label{eq:70aa}
  \Cov\left[ \mathcal{X}_{i_1, k_1}^{(j_1)}, \mathcal{X}_{i_2, k_2}^{(j_2)} 
  \right]
  = \Cov\left[ 
  \mathcal{G}_{i_1, k_1}^{(j_1)}, \mathcal{G}_{i_2, k_2}^{(j_2)} \right] \; .
\end{align}
\end{corollary}

\begin{definition}
An \emph{ensemble collection for 
$(\overline{\underline{\Omega}}, \underline{\mathcal{P}})$}
is a tuple \begin{align*}
\left(\overline{\underline{X}},
\underline{\overline{\mathcal{X}}} = (\underline{\mathcal{X}}^{(1)}, \ldots,
\underline{\mathcal{X}}^{(\ell)}),
\underline{\overline{\mathcal{G}}} = (\underline{\mathcal{G}}^{(1)}, \ldots,
\underline{\mathcal{G}}^{(\ell)})\right)
\end{align*}
where 
\begin{itemize}
\item $\overline{\underline{X}}$ is a random vector
distributed according to $(\overline{\underline{\Omega}}, 
\underline{\mathcal{P}})$,
\item 
$\underline{\mathcal{X}}^{(1)}, \ldots, \underline{\mathcal{X}}^{(\ell)}$
are ensemble sequences constructed from 
$\underline{X}^{(1)}, \ldots, \underline{X}^{(\ell)}$, respectively, 
\item and 
$\underline{\mathcal{G}}^{(1)}, \ldots, \underline{\mathcal{G}}^{(\ell)}$
are obtained from Corollary~\ref{cor:cov-xg}.
\end{itemize}
\end{definition}

\subsection{Hypercontractivity}
\label{sec:hyper}

In this section we develop a version of hypercontractivity for products
of multilinear polynomials. Our goal is to prove Lemma
\ref{lem:hypercontractivity-technical}.

Recall the operator $T_\rho$ from Definition \ref{def:t-rho}.

\begin{definition}
Let $\underline{\mathcal{X}}$ be an ensemble sequence
and let $1 \le p \le q < \infty$ and $\rho \in [0, 1]$.
We say that the sequence $\underline{\mathcal{X}}$ is
\emph{$(p, q, \rho)$-hypercontractive} if for
every multilinear polynomial $P$ compatible with $\underline{\mathcal{X}}$
we have
\begin{align*}
\EE\left[ \left| T_\rho P (\underline{\mathcal{X}}) \right|^q \right]^{1/q}
\le
\EE\left[ \left| P (\underline{\mathcal{X}}) \right|^p \right]^{1/p}
\end{align*}
\end{definition}

\begin{definition}
Let $\mathcal{X}$ be an orthonormal ensemble
and let $1 \le p \le q < \infty$ and $\rho \in [0, 1]$.
We say that the ensemble $\mathcal{X}$ is
\emph{$(p, q, \rho)$-hypercontractive} if
the one-element ensemble sequence
$\underline{\mathcal{X}} := (\mathcal{X})$ is
$(p, q, \rho)$-hypercontractive.
\end{definition}

We start with stating without proofs the hypercontractivity of
orthonormal ensembles that we use in
the invariance principle:

\begin{theorem}[\cite{Bon70, Nel73, Gro75, Bec75}]
\label{thm:hypercontractivity-single-gauss}
Let $\mathcal{G}$ be a Gaussian orthonormal
ensemble and $\rho \in \allowbreak 
[0, \sqrt{2}/2]$. Then,
$\mathcal{G}$ is $(2, 3, \rho)$-hypercontractive.
\end{theorem}

\begin{theorem}[Special case of Theorem 3.1 in \cite{Wol07}]
\label{thm:hypercontractivity-single-discrete}
Let $\mathcal{X}$ be an orthonormal ensemble constructed from a random variable
$X$ distributed according to a 
(single-coordinate, single-step) probability space 
$(\Omega, \pi)$ with 
$\min_{x \in \Omega} \pi(x) \ge \alpha \ge 0$.

Then, $\mathcal{X}$ is $(2, 3, \alpha^{1/6}/2)$-hypercontractive.
\end{theorem}

Subsequently, we observe that
an ensemble sequence constructed from
hypercontractive ensembles is itself hypercontractive:

\begin{theorem}
\label{thm:tensorization}
Let $1 \le p \le q < \infty$, $\rho \in [0, 1]$ and let
$\underline{\mathcal{X}} := (\mathcal{X}_1, \ldots, \mathcal{X}_n)$
be an ensemble sequence such that
for every $i \in [n]$, the ensemble $\mathcal{X}_i$
is $(p, q, \rho)$-hypercontractive. Then,
the sequence $\underline{\mathcal{X}}$ is also
$(p, q, \rho)$-hypercontractive.
\end{theorem}

Yet again, we omit the proof of Theorem \ref{thm:tensorization}.
We remark that it is well-known as the \emph{tensorization argument}.
The argument can be found, e.g., in the proof of Proposition 3.11
in \cite{MOO10}.

\begin{definition}
Let $\underline{X}$ be a random vector distributed
according to a (single-step, tensorized) probability space 
$(\underline{\Omega}, \underline{\pi})$.
We say that an ensemble sequence
$\underline{\mathcal{X}} = (\mathcal{X}_1, \ldots, \mathcal{X}_n)$
is \emph{$\underline{X}$-Gaussian-mixed} if
for each $i \in[n]$:
\begin{itemize}
\item Either $\mathcal{X}_i$ is constructed
from the random variable $X_i$,
\item or $\mathcal{X}_i$ is a Gaussian ensemble.
\end{itemize}
\end{definition}

Theorems \ref{thm:hypercontractivity-single-gauss},
\ref{thm:hypercontractivity-single-discrete} and
\ref{thm:tensorization} immediately imply:

\begin{corollary}
\label{cor:hypercontractivity}
Let $\underline{X}$ be a random vector distributed according to
a probability space
$(\underline{\Omega}, \underline{\pi})$ with
$\min_{x \in \Omega} \pi(x) \ge \alpha \ge 0$
and let $\underline{\mathcal{X}}$ be an 
$\underline{X}$-Gaussian-mixed
ensemble sequence.

Then, $\underline{\mathcal{X}}$ is $(2, 3, \alpha^{1/6}/2)$-hypercontractive.
\end{corollary}

\begin{theorem}
\label{thm:hypercontractivity-degree}
Let $\underline{X}$ be a random vector distributed according to a
probability space $(\underline{\Omega}, \underline{\pi})$
with $\min_{x \in \Omega} \pi(x) \ge \alpha > 0$
and let $\underline{\mathcal{X}}$ be an $\underline{X}$-Gaussian-mixed
ensemble sequence.
Let $P$ be a multilinear
polynomial compatible with $\underline{\mathcal{X}}$ of degree at most $d$.\
Then,
\begin{align*}
\EE\left[ \left| P(\underline{\mathcal{X}})  \right|^3  \right]^{1/3}
\le \left(\frac{2}{\alpha^{1/6}}\right)^d 
\sqrt{\EE\left[ P^2 \right]} \; .
\end{align*}
\end{theorem}

\begin{proof}
Let $\rho := \alpha^{1/6}/2$ and write
$P(\underline{\mathcal{X}}) = \sum_{\sigma} \beta(\sigma) \mathcal{X}_{\sigma}$. 
By Corollary \ref{cor:hypercontractivity}, definitions of $T_\rho$ and
$E[P^2]$, and the degree bound on $P$,
\begin{IEEEeqnarray*}{l}
\EE \left[ \left| P(\underline{\mathcal{X}}) \right|^3 \right]^{1/3} 
= \EE \left[ \left| T_\rho T_{1/\rho} P(\underline{\mathcal{X}}) 
\right|^3 \right]^{1/3}
\le \sqrt{ \EE \left[ (T_{1/\rho} P)^2 \right] }
\\ \qquad = \sqrt{ \sum_{\sigma} \rho^{-2|\sigma|} \beta(\sigma)^2}
\le \sqrt{ \sum_{\sigma} \rho^{-2d} \beta(\sigma)^2}
= \rho^{-d} \sqrt{ \EE[P^2] } \; .
\end{IEEEeqnarray*}
\end{proof}

\begin{lemma}
\label{lem:hypercontractivity-technical}
Let $\overline{\underline{X}}$ be a random vector distributed according to
a (multi-step)  probability space with equal marginals
$(\overline{\underline{\Omega}}, \underline{\mathcal{P}})$
with $\min_{x \in \Omega} \pi(x) \ge \alpha > 0$.

Let
$\underline{\mathcal{S}}^{(1)}, \ldots, \underline{\mathcal{S}}^{(\ell)}$
be ensemble sequences such that $\underline{\mathcal{S}}^{(j)}$ is
$\underline{X}^{(j)}$-Gaussian-mixed.
Let $P^{(1)}, \ldots, P^{(\ell)}$ be multilinear polynomials
such that $P^{(j)}$ is
compatible with $\underline{\mathcal{S}}^{(j)}$ 
and also $\deg(P^{(j)}) \le d$.

Then, for every triple $j_1, j_2, j_3 \in [\ell]$:
\begin{align*}
\EE \left[ \left| \prod_{k=1}^3 P^{(j_k)}(\underline{\mathcal{S}}^{(j_k)})
\right| \right]
\le \left( \frac{8}{\sqrt{\alpha}} \right)^d \cdot
\sqrt{ \prod_{k=1}^3 \EE \left[ (P^{(j_k)})^2\right]}
\; .
\end{align*}
\end{lemma} 

\begin{proof}
Let $\rho := \alpha^{1/6}/2$. By Hölder's inequality
and  Theorem 
\ref{thm:hypercontractivity-degree},
\begin{IEEEeqnarray*}{rCl}
\EE \left[ \left| \prod_{k=1}^3 P^{(j_k)}(\underline{\mathcal{S}}^{(j_k)})
\right| \right]
& \le &
\prod_{k=1}^3 \EE \left[ \left| P^{(j_k)}(\underline{\mathcal{S}}^{(j_k)}) 
\right|^3 \right]^{1/3}
\\ & \le & \rho^{-3d} \cdot 
\sqrt{ \prod_{k=1}^3 \EE \left[ (P^{(j_k)})^2\right] } \; .
\end{IEEEeqnarray*}
\end{proof}
\subsection{Invariance principle}

In this section we prove a basic version of invariance principle for multiple
polynomials.

We say that a function is $B$-smooth if all of its third-order partial
derivatives are uniformly bounded by $B$:
\begin{definition}
For $B \ge 0$ we say that a function $\Psi: \mathbb{R}^\ell \to \mathbb{R}$ is
\emph{$B$-smooth} if $\Psi \in \mathcal{C}^3$ and
for every $j_1, j_2, j_3 \in [\ell]$ and every 
$\overline{x} = (x^{(1)}, \ldots, x^{(\ell)}) \in \mathbb{R}^\ell$ we have
\begin{align*}
  \left| \frac{\partial^3}
  {\partial x^{(j_1)} \partial x^{(j_2)} \partial x^{(j_3)}} 
  \Psi(\overline{x}) \right| \le B \; .
\end{align*}
\end{definition}

\begin{theorem}[Invariance Principle]
\label{thm:invariance-main}
Let $(\overline{\underline{X}}, 
\overline{\underline{\mathcal{X}}},
\overline{\underline{\mathcal{G}}})$ be an ensemble collection
for a probability space 
$(\overline{\underline{\Omega}}, \underline{\mathcal{P}})$
with $\min_{x \in \Omega} \pi(x) \ge \alpha > 0$.

Let $\overline{P} = (P^{(1)}, \ldots, P^{(\ell)})$ be
such that $P^{(j)}$ is a multilinear polynomial compatible with 
the ensemble sequence $\underline{\mathcal{X}}^{(j)}$.

Let $d \in \mathbb{N}$ and $\tau \in [0, 1]$ and assume that
$\deg(P^{(j)}) \le d$ and $\Var[P^{(j)}] \le 1$ for each $j \in [\ell]$, 
and that $\sum_{j=1}^\ell \Inf_i(P^{(j)}) \le \tau$ for each $i \in [n]$.

Finally, let $\Psi: \mathbb{R}^\ell \to \mathbb{R}$
be a $B$-smooth function. Then,
\begin{align*}
\left| \EE \left[
  \Psi(\overline{P}(\overline{\underline{\mathcal{X}}}))
  - \Psi(\overline{P}(\overline{\underline{\mathcal{G}}}))
  \right] \right|
  \le
  \frac{\ell^{5/2}d B}{3} \left( \frac{8}{\sqrt{\alpha}} \right)^d 
  \sqrt{\tau} \; .
\end{align*}
\end{theorem}

\begin{remark}
A typical setting of parameters for which Theorem \ref{thm:invariance-main}
might be successfully applied is constant $\ell$, $d$, $B$, and $\alpha$,
while $\tau = o(1)$ (as $n \to \infty$).
\end{remark}

The rest of this section is concerned with proving Theorem 
\ref{thm:invariance-main}.

For $i \in \{0, \ldots, n\}$ and $j \in [\ell]$ let the ensemble sequence 
$\underline{\mathcal{U}}_{(i)}^{(j)}$ be
defined as $\underline{\mathcal{U}}_{(i)}^{(j)} := 
(\mathcal{G}_1^{(j)}, \ldots, \mathcal{G}_i^{(j)}, 
\mathcal{X}_{i+1}^{(j)}, \ldots, \mathcal{X}_n^{(j)})$.

\begin{claim}
\label{cl:invariance-triangle-inequality}
\begin{align*}
\left| \EE \left[
	\Psi(\overline{P}(\overline{\underline{\mathcal{X}}}))
	- \Psi(\overline{P}(\overline{\underline{\mathcal{G}}}))
	\right] \right|
  \le
\sum_{i=1}^n \left| \EE \left[
	\Psi(\overline{P}(\overline{\underline{\mathcal{U}}}_{(i-1)}))
	- \Psi(\overline{P}(\overline{\underline{\mathcal{U}}}_{(i)}))
	\right] \right| \; .
\end{align*}
\end{claim}

\begin{proof}
By the triangle inequality.
\end{proof}

Due to Claim 
\ref{cl:invariance-triangle-inequality}, we will estimate
\begin{align*}
\left| \EE \left[
	\Psi(\overline{P}(\overline{\underline{\mathcal{U}}}_{(i-1)}))
	- \Psi(\overline{P}(\overline{\underline{\mathcal{U}}}_{(i)}))
	\right] \right| 
\end{align*}
for every $i \in [n]$. Fix $i \in [n]$ and write
$\underline{\mathcal{T}}^{(j)} := \underline{\mathcal{U}}_{(i-1)}^{(j)}$ and
$\underline{\mathcal{U}}^{(j)} := \underline{\mathcal{U}}_{(i)}^{(j)}$ for 
readability.
For $j \in [\ell]$ we can write
\begin{align}
\label{eq:62a}
  P^{(j)}(\underline{\mathcal{T}}^{(j)}) =
  A^{(j)} + \sum_{k>0} \mathcal{X}_{i,k}^{(j)} \cdot B^{(j)}_k =
  A^{(j)} + P_i^{(j)}(\underline{\mathcal{T}}^{(j)}) \; ,
\end{align}
where $A^{(j)}$ and $B^{(j)}_k$ do not depend on the coordinate $i$
and, if $P^{(j)}(\underline{\mathcal{T}}^{(j)}) = 
\sum_\sigma \alpha(\sigma) \mathcal{T}_\sigma^{(j)}$,
then $P_i^{(j)}(\underline{\mathcal{T}}^{(j)}) = \sum_{\sigma: i \in \supp(\sigma)} 
\alpha(\sigma) \mathcal{T}_\sigma^{(j)}$.
At the same time, since $A^{(j)}$ and $B_k^{(j)}$ do not depend on the 
$i$-th coordinate,
\begin{align*}
  P^{(j)}(\underline{\mathcal{U}}^{(j)}) =
  A^{(j)} + \sum_{k>0} \mathcal{G}_{i,k}^{(j)} \cdot B^{(j)}_k =
  A^{(j)} + P_i^{(j)}(\underline{\mathcal{U}}^{(j)})
  \; .
\end{align*}

We note for later use that the construction gives us
\begin{align}\label{eq:36a}
\deg(P_i^{(j)}) &\le d\\
\EE\left[ \left( P_i^{(j)}\right)^2\right] &= \Inf_i\left(P^{(j)}\right)\;.
\label{eq:37a}
\end{align}

\medskip

The rest of the proof proceeds as follows:
we calculate the multivariate second order Taylor expansion (i.e., with 
the third-degree rest) of the expression, getting
\begin{IEEEeqnarray*}{l}
\Psi(\overline{P}(\overline{\underline{\mathcal{T}}}))
- \Psi(\overline{P}(\overline{\underline{\mathcal{U}}})) =
\\ \qquad =
\Psi\left(A^{(1)} + \sum_{k>0} \mathcal{X}_{i,k}^{(1)} B_{k}^{(1)}, \ldots,
A^{(\ell)} + \sum_{k>0} \mathcal{X}_{i,k}^{(\ell)} B_{k}^{(\ell)}\right) \\
\qquad \quad - \: \Psi\left(A^{(1)} + \sum_{k>0} 
\mathcal{G}_{i,k}^{(1)} B_{k}^{(1)}, \ldots,
A^{(\ell)} + \sum_{k>0} \mathcal{G}_{i,k}^{(\ell)} B_{k}^{(\ell)}\right)
\end{IEEEeqnarray*}
around the point $\overline{A} := (A^{(1)}, \ldots, A^{(\ell)})$.
We will see that:
\begin{itemize}
  \item All the terms up to the second degree cancel in expectation 
    due to the properties of ensemble sequences.
  \item The remainder, which is of the third degree,
    can be bounded using that $\Psi$ is $B$-smooth,
    properties of $P_{i}^{(j)}$, and
    hypercontractivity, in particular
    Lemma \ref{lem:hypercontractivity-technical}.
\end{itemize}

We proceed with a detailed description.
The first result we will need is multivariate Taylor's theorem
for $B$-smooth functions:
\begin{theorem}
\label{thm:taylor}
Let $\Psi: \mathbb{R}^\ell \to \mathbb{R}$ be a $B$-smooth function
and let 
$\overline{x} = (x^{(1)}, \ldots, x^{(\ell)}), 
\overline{\epsilon} = (\epsilon^{(1)}, \ldots, \epsilon^{(\ell)})
\in \mathbb{R}^\ell$.
Then,
\begin{IEEEeqnarray*}{l}
\Bigg| 
\Psi \left( x^{(1)}+\epsilon^{(1)}, \ldots, x^{(\ell)}+\epsilon^{(\ell)} \right)
- \\
\quad \left( \Psi(\overline{x}) + 
\sum_{j \in [\ell]} \epsilon^{(j)} 
\frac{\partial}{\partial x^{(j)}}\Psi(\overline{x})
+ \frac12 \sum_{j_1,j_2 \in [\ell]} \epsilon^{(j_1)} \epsilon^{(j_2)}
\frac{\partial^2}{\partial x^{(j_1)} \partial x^{(j_2)}} \Psi(\overline{x})
\right) \Bigg|
\\ \qquad \le \frac{B}{6} \sum_{j_1, j_2, j_3 \in [\ell]} \left|
\epsilon^{(j_1)} \epsilon^{(j_2)} \epsilon^{(j_3)} \right| \; .
\end{IEEEeqnarray*}
\end{theorem}

We omit the proof of Theorem \ref{thm:taylor}.

\begin{lemma}
Fix $i \in [n]$ and write 
$\underline{\mathcal{T}}^{(j)} := \underline{\mathcal{U}}^{(j)}_{(i-1)}$
and $\underline{\mathcal{U}}^{(j)} := \underline{\mathcal{U}}^{(j)}_{(i)}$.
Then,
\begin{IEEEeqnarray}{l}
\EE\left[ \Psi(\overline{P}(\overline{\underline{\mathcal{T}}})) \right]
= \nonumber  \\
= \EE \left[
\Psi(\overline{A}) + \frac12 \sum_{j_1, j_2 \in [\ell]} \left(
\sum_{k_1, k_2 > 0} \mathcal{X}_{i,k_1}^{(j_1)} \mathcal{X}_{i,k_2}^{(j_2)}
B_{k_1}^{(j_1)} B_{k_2}^{(j_2)}
\frac{\partial^2}{\partial A^{(j_1)} \partial A^{(j_2)}} \Psi(\overline{A}) 
\right)
+ R_{\underline{\mathcal{T}}} \right] 
\; , \label{eq:60a} \IEEEeqnarraynumspace
\end{IEEEeqnarray}
and
\begin{IEEEeqnarray}{l}
\EE\left[ \Psi(\overline{P}(\overline{\underline{\mathcal{U}}})) \right]
= \nonumber  \\
= \EE \left[
\Psi(\overline{A}) + \frac12 \sum_{j_1, j_2 \in [\ell]} \left(
\sum_{k_1, k_2 > 0} \mathcal{G}_{i,k_1}^{(j_1)} \mathcal{G}_{i,k_2}^{(j_2)}
B_{k_1}^{(j_1)} B_{k_2}^{(j_2)}
\frac{\partial^2}{\partial A^{(j_1)} \partial A^{(j_2)}} \Psi(\overline{A}) 
\right)
+ R_{\underline{\mathcal{U}}} \right] 
\; , \label{eq:61a} \IEEEeqnarraynumspace
\end{IEEEeqnarray}
where random variables $R_{\underline{\mathcal{T}}}$ and $R_{\underline{\mathcal{U}}}$ 
are such that
\begin{IEEEeqnarray}{l}
\EE\left[
\left| R_{\underline{\mathcal{T}}} \right|
\right],
\EE\left[
\left| R_{\underline{\mathcal{U}}} \right|
\right] 
\le \frac{\ell^{3/2} B}{6} \left(\frac{8}{\sqrt{\alpha}}\right)^d
\left( \sum_{j=1}^\ell \Inf_i(P^{(j)}) \right)^{3/2} \; .
\label{eq:65a}
\end{IEEEeqnarray}
\end{lemma}

\begin{proof}
We show only (\ref{eq:60a}) and the bound on $\EE[|R_{\underline{\mathcal{T}}}|]$,
the proofs for the ensemble sequence $\underline{\mathcal{U}}$ being
analogous.

As a preliminary remark,
note that since all the random ensembles we are dealing with are hypercontractive,
and since $\Psi$ is $B$-smooth,
all the terms in the expressions above have finite expectations.

Keeping in mind both decompositions from (\ref{eq:62a}), 
by Theorem \ref{thm:taylor}
\begin{IEEEeqnarray}{l}
\Psi(\overline{P}(\overline{\underline{\mathcal{T}}}))
=
\Psi(\overline{A}) + 
\sum_{j \in [\ell]} \left(
\sum_{k>0} \mathcal{X}_{i,k}^{(j)} B_k^{(j)}
\frac{\partial}{\partial A^{(j)}} \Psi(\overline{A})
\right) + \nonumber \\
\qquad + \: \frac12 \sum_{j_1, j_2 \in [\ell]} \left(
\sum_{k_1,k_2>0} \mathcal{X}_{i,k_1}^{(j_1)} \mathcal{X}_{i,k_2}^{(j_2)} 
B_{k_1}^{(j_1)} B_{k_2}^{(j_2)} 
\frac{\partial^2}{\partial A^{(j_1)} \partial A^{(j_2)}} \Psi(\overline{A})
\right)
+ R_{\underline{\mathcal{T}}} \; ,
\IEEEeqnarraynumspace \label{eq:63a}
\end{IEEEeqnarray}
where
\begin{IEEEeqnarray}{l}
\label{eq:64a}
\EE[ | R_{\underline{\mathcal{T}}} | ] \le \frac{B}{6}
\sum_{j_1, j_2, j_3 \in [\ell]} \EE\left[ \left|
\prod_{k=1}^3 P_i^{(j_k)}(\underline{\mathcal{T}}^{(j_k)})
\right| \right] \; .
\end{IEEEeqnarray}
Since $\EE[\mathcal{X}_{i,k}^{(j)}] = 0$, and all other terms
are independent of coordinate $i$, we have
\begin{align*}
\EE \left[ \sum_{j \in [\ell]}
\sum_{k>0} \mathcal{X}_{i,k}^{(j)} B_k^{(j)}
\frac{\partial}{\partial A^{(j)}} \Psi(\overline{A})
\right] = 0 \; ,
\end{align*}
which together with (\ref{eq:63a}) yields (\ref{eq:60a}).

As for the bound on $\EE[|R_{\underline{\mathcal{T}}}|]$, since
$T^{(j)}$ is $\underline{X}^{(j)}$-Gaussian-mixed ensemble sequence,
due to (\ref{eq:64a}), Lemma \ref{lem:hypercontractivity-technical}
(note that the degree is bounded due to (\ref{eq:36a})), and~(\ref{eq:37a}),
\begin{IEEEeqnarray*}{rCl}
\EE[|R_{\underline{\mathcal{T}}}|]
&\le&
\frac{B}{6} \left(\frac{8}{\sqrt{\alpha}}\right)^d
\sum_{j_1,j_2,j_3 \in [\ell]} \sqrt{
\prod_{k=1}^3 \EE\left[\left(P_i^{(j_k)}\right)^2\right]} \\
&=&
\frac{B}{6} \left(\frac{8}{\sqrt{\alpha}}\right)^d
\sum_{j_1,j_2,j_3 \in [\ell]} \sqrt{
\prod_{k=1}^3 \Inf_i(P^{(j_k)})} \\
&\le&
\frac{\ell^{3/2} B}{6} \left(\frac{8}{\sqrt{\alpha}}\right)^d
\left( \sum_{j=1}^\ell \Inf_i(P^{(j)}) \right)^{3/2} \; ,
\end{IEEEeqnarray*}
where the last inequality uses
$\sum_{j_1,j_2,j_3} \nu(j_1,j_2,j_3) \leq\allowbreak 
\sqrt{\ell^3} \sqrt{\sum \nu^2(j_1,j_2,j_3)}$
for the vector $\nu$ with entries 
$\nu(j_1,j_2,j_3) = \sqrt{\prod_{k=1}^3 \Inf_i(P^{(j_k)})}$.
\end{proof}

\begin{lemma}
\label{lem:invariance-single-diff}
Fix $i \in [n]$ and write 
$\underline{\mathcal{T}}^{(j)} := \underline{\mathcal{U}}^{(j)}_{(i-1)}$
and $\underline{\mathcal{U}}^{(j)} := \underline{\mathcal{U}}^{(j)}_{(i)}$.
Then,
\begin{align*}
\left| \EE\left[
\Psi(\overline{P}(\overline{\underline{\mathcal{T}}}))
- \Psi(\overline{P}(\overline{\underline{\mathcal{U}}}))
\right] \right| \le
\frac{\ell^{3/2} B}{3} \left(\frac{8}{\sqrt{\alpha}}\right)^d
\left( \sum_{j=1}^\ell \Inf_i(P^{(j)}) \right)^{3/2} \; .
\end{align*}
\end{lemma}

\begin{proof}
First, we need to show that the second-order terms in 
(\ref{eq:60a}) and (\ref{eq:61a}) cancel out.
Since
by Lemma \ref{lem:cov-xg}
for every $j_1, j_2 \in [\ell]$ and $k_1, k_2 > 0$:
\begin{align*}
\EE\left[ \mathcal{X}_{i,k_1}^{(j_1)} \mathcal{X}_{i,k_2}^{(j_2)} \right]
= \Cov\left[ \mathcal{X}_{i,k_1}^{(j_1)}, \mathcal{X}_{i,k_2}^{(j_2)} \right]
= \Cov\left[ \mathcal{G}_{i,k_1}^{(j_1)}, \mathcal{G}_{i,k_2}^{(j_2)} \right]
= \EE\left[ \mathcal{G}_{i,k_1}^{(j_1)} \mathcal{G}_{i,k_2}^{(j_2)} \right] \; ,
\end{align*}
and since all the other terms are independent of coordinate $i$, we have
\begin{IEEEeqnarray*}{l}
\EE \left[ \sum_{j_1, j_2 \in [\ell]}
\sum_{k_1,k_2>0} \mathcal{X}_{i,k_1}^{(j_1)} \mathcal{X}_{i,k_2}^{(j_2)} 
B_{k_1}^{(j_1)} B_{k_2}^{(j_2)} 
\frac{\partial^2}{\partial A^{(j_1)} \partial A^{(j_2)}} \Psi(\overline{A})
\right] \\
\qquad =
\EE \left[ \sum_{j_1, j_2 \in [\ell]}
\sum_{k_1,k_2>0} \mathcal{G}_{i,k_1}^{(j_1)} \mathcal{G}_{i,k_2}^{(j_2)} 
B_{k_1}^{(j_1)} B_{k_2}^{(j_2)} 
\frac{\partial^2}{\partial A^{(j_1)} \partial A^{(j_2)}} \Psi(\overline{A})
\right] \; .
\end{IEEEeqnarray*}
Therefore, by (\ref{eq:60a}), (\ref{eq:61a}) and (\ref{eq:65a}),
\begin{IEEEeqnarray*}{rCl}
\left| \EE\left[
\Psi(\overline{P}(\overline{\underline{\mathcal{T}}}))
- \Psi(\overline{P}(\overline{\underline{\mathcal{U}}}))
\right] \right| &\le&
\EE[|R_{\underline{\mathcal{T}}}|] + \EE[|R_{\underline{\mathcal{U}}}|]\\
&\le&
\frac{\ell^{3/2} B}{3} \left(\frac{8}{\sqrt{\alpha}}\right)^d
\left( \sum_{j=1}^\ell \Inf_i(P^{(j)}) \right)^{3/2} \; ,
\end{IEEEeqnarray*}
as claimed.
\end{proof}

\begin{proof}[Proof of Theorem \ref{thm:invariance-main}]
Recall that $\sum_{j=1}^\ell \Inf_i(P^{(j)}) \le \tau$
and that $\Var[P^{(j)}] \le 1$.
By Claim \ref{cl:invariance-triangle-inequality},
Lemma \ref{lem:invariance-single-diff}
and Claim \ref{lem:influence-vs-variance},
\begin{IEEEeqnarray*}{l}
\left| \EE \left[
	\Psi(\overline{P}(\overline{\underline{\mathcal{X}}}))
	- \Psi(\overline{P}(\overline{\underline{\mathcal{G}}}))
	\right] \right|
  \le
\sum_{i=1}^n \left| \EE \left[
	\Psi(\overline{P}(\overline{\underline{\mathcal{U}}}_{(i-1)}))
	- \Psi(\overline{P}(\overline{\underline{\mathcal{U}}}_{(i)}))
	\right] \right|
\\ \qquad \le \frac{\ell^{3/2}B}{3} \left(\frac{8}{\sqrt{\alpha}}\right)^d
\sum_{i=1}^n \left( \sum_{j=1}^\ell \Inf_i(P^{(j)})  \right)^{3/2}
\\ \qquad \le \frac{\ell^{3/2}B}{3} \left(\frac{8}{\sqrt{\alpha}}\right)^d
\sqrt{\tau} \sum_{i=1}^n \sum_{j=1}^\ell \Inf_i(P^{(j)})
\\ \qquad = \frac{\ell^{3/2}B}{3} \left(\frac{8}{\sqrt{\alpha}}\right)^d
\sqrt{\tau} \sum_{j=1}^\ell \Inf(P^{(j)})
\le \frac{\ell^{5/2} d B}{3} \left(\frac{8}{\sqrt{\alpha}}\right)^d \sqrt{\tau} \; .
\end{IEEEeqnarray*}
\end{proof}

\subsection{A tailored application of invariance principle}
\begin{definition}
\label{def:xi}
Define $\phi: \mathbb{R} \to \mathbb{R}$ as
\begin{align*}
  \phi(x) := \begin{cases}
    0 & \text{if $x \le 0$,}\\
    x & \text{if $x \in (0, 1)$,}\\
    1 & \text{if $x \ge 1$,}
  \end{cases}
\end{align*}
and $\chi: \mathbb{R}^\ell \to \mathbb{R}$
as $\chi(\overline{x}) := \prod_{j=1}^{\ell} \phi(x^{(j)})$.
\end{definition}

\begin{definition}
\label{def:gamma-smooth}
Let $P$ be a multilinear polynomial and $\gamma \in [0, 1]$.
We say that $P$ is \emph{$\gamma$-decaying} if
for each $d \in \mathbb{N}$ we have
\begin{align*}
	\EE \left[\left(P^{\ge d} \right)^2 \right] \le (1-\gamma)^d \; .
\end{align*}

We also say that a tuple of multilinear polynomials
$\overline{P} = (P^{(1)}, \ldots, P^{(\ell)})$ is $\gamma$-decaying
if $P^{(j)}$ is $\gamma$-decaying for each $j \in [\ell]$.
\end{definition}

Note that if a multilinear polynomial 
$P$ is $\gamma$-decaying, then, in particular,
$\Var[P] \le \EE[ P^2 ] \le 1$.

Our goal in this section is to prove a version of invariance principle
for $\gamma$-decaying multilinear polynomials and the function $\chi$:
\begin{theorem}
\label{thm:invariance-smoothed}
Let $(\overline{\underline{X}}, 
\overline{\underline{\mathcal{X}}},
\overline{\underline{\mathcal{G}}})$ be an ensemble collection
for a probability space 
$(\overline{\underline{\Omega}}, \underline{\mathcal{P}})$
with $\min_{x \in \Omega} \pi(x) \ge \alpha$, $\alpha \in (0, 1/2]$.

Let $\overline{P} = (P^{(1)}, \ldots, P^{(\ell)})$ be such that $P^{(j)}$
is a multilinear polynomial compatible with the ensemble sequence
$\underline{\mathcal{X}}^{(j)}$.

Let $\gamma \in [0,1]$, $\tau \in (0, 1]$ and assume that $\overline{P}$ is
$\gamma$-decaying and that $\sum_{j=1}^\ell \Inf_i(P^{(j)}) \le \tau$ for each
$i \in [n]$. There exists an absolute constant $C \ge 0$ such that
\begin{align*}
  \left| \EE \left[ 
  \chi(\overline{P}(\underline{\overline{\mathcal{X}}}))
  - \chi(\overline{P}(\underline{\overline{\mathcal{G}}}))
  \right] \right| \le
  C \ell^{5/2} \cdot \tau^{\frac{\gamma}{C \ln 1/ \alpha}} \; .
\end{align*}
\end{theorem}

Two obstacles to proving Theorem \ref{thm:invariance-smoothed}
by direct application of Theorem \ref{thm:invariance-main} are:
\begin{enumerate}
  \item The function $\chi$ is not $\mathcal{C}^3$.
  \item A $\gamma$-decaying multilinear polynomial does not have bounded degree.
\end{enumerate}
We will deal with those problems in turn.

\subsubsection{Approximating
\texorpdfstring{$\chi$}{chi}
with a 
\texorpdfstring{$\mathcal{C}^3$}{C\^{}3}
function}

To apply Theorem \ref{thm:invariance-main}, we are going
to approximate $\phi$ and $\chi$
with $\mathcal{C}^3$
(in fact, $\mathcal{C}^{\infty}$) functions.

For that we need to introduce the notion of convolution and a basic calculus 
theorem, whose proof we omit (see, e.g., Chapter 9 in \cite{Rud87}):
\begin{definition}
Let $f: \mathbb{R} \to \mathbb{R}$ and $S \subseteq \mathbb{R}$. We say
that $S$ is a $\emph{support}$ of $f$ if $x \notin S$ implies $f(x) = 0$.

We say that $f$ has $\emph{compact support}$ if there exists a bounded
interval $I$ that is a support of $f$.
\end{definition}

\begin{definition}
The convolution $f\ast g$ of two continuous functions 
$f,g: \mathbb{R}\to\mathbb{R}$, at least one of which has
compact support,
is $(f \ast g)(x) := \int_{-\infty}^\infty f(x-t)g(t)\,\mathrm{d}t$.
\end{definition}

\begin{theorem}
\label{thm:convolution-differentiation}
Let functions $f, g: \mathbb{R} \to \mathbb{R}$ be such that $f$ is continuous 
on $\mathbb{R}$, $g \in \mathcal{C}^{\infty}$ and $g$ has compact support. 
Then, $(f \ast g) \in \mathcal{C}^{\infty}$.
Furthermore, for every $k \in \mathbb{N}$ and $x \in \mathbb{R}$:
\begin{align*}
\frac{\partial^k}{\partial x^k}(f \ast g)(x) = 
\left(f \ast \frac{\partial^k g}{\partial x^k} \right)(x) \; .
\end{align*}
\end{theorem}

We also need a special density function with support $[-1, 1]$:

\begin{theorem}
\label{thm:smooth-distribution}
There exists a function $\psi: \mathbb{R} \to \mathbb{R}_{\ge 0}$
such that all of the following hold:
\begin{itemize}
\item $\psi \in \mathcal{C}^{\infty}$.
\item $\psi$ has support $[-1, 1]$.
\item $\forall x: \psi(x) = \psi(-x)$.
\item $\int_{-\infty}^\infty \psi(x) \, \mathrm{d}x = \int_{-1}^1 \psi(x) \, \mathrm{d}x = 1$.
\end{itemize}
\end{theorem}
\begin{proof}
Consider
\begin{align}
\Psi(x) := 
\begin{cases}
\exp(-\frac{1}{(x+1)^2})\cdot \exp(-\frac{1}{(x-1)^2})& 
\text{if $x \in (-1,1)$}\\
0 &\text{otherwise}
\end{cases}
\end{align}
and set $\psi(x) := \Psi(x)/c$ where $c := \int_{-1}^1 \Psi(x) \, \mathrm{d}x$.
\end{proof}

For any $\lambda > 0$ we can rescale
$\psi$ to an analogous distribution with support $[-\lambda, \lambda]$:
\begin{definition}
Let $\lambda > 0$ and define $\psi_\lambda: \mathbb{R} \to \mathbb{R}_{\ge 0}$
as $\psi_\lambda(x) := \frac{1}{\lambda} \psi\left( \frac{x}{\lambda} \right)$.
\end{definition}

It is easy to see that $\psi_\lambda$ has properties analogous to $\psi$:
\begin{claim}
Let $\lambda > 0$. $\psi_\lambda$ has the following properties:
\begin{itemize}
\item $\psi_\lambda \in \mathcal{C}^{\infty}$.
\item $\psi_\lambda$ has support $[-\lambda, \lambda]$.
\item $\forall x: \psi_\lambda(x) = \psi_\lambda(-x)$.
\item $\int_{-\infty}^\infty \psi_\lambda(x) \, \mathrm{d}x = 
\int_{-\lambda}^{\lambda} \psi_\lambda(x) \, \mathrm{d}x = 1$.
\end{itemize}
\end{claim}

We see that convoluting $\phi$ with $\psi_\lambda$ for a small $\lambda$ 
results in a smooth function that is still very close to $\phi$:
\begin{definition}
Let $\lambda \in (0, 1/2)$ and define $\phi_\lambda: \mathbb{R} \to \mathbb{R}$ 
as $\phi_\lambda  := \phi \ast \psi_\lambda$.
\end{definition}

To start with, we state some easy to verify properties of $\phi_\lambda$:
\begin{claim}
\label{cl:phi-lambda}
Let $\lambda \in (0, 1/2)$. The function $\phi_\lambda$ has the following 
properties:
\begin{itemize}
  \item $\phi_\lambda(x) = 
    \int_{-\lambda}^{\lambda} \psi_\lambda(y) \phi(x+y) \, \mathrm{d}y$.
  \item $x \le -\lambda \lor x \in [\lambda, 1-\lambda] \lor x \ge 1+\lambda
    \implies \phi_\lambda(x) = \phi(x)$.
  \item $x \in [-\lambda, \lambda] \implies \phi_\lambda(x) \in [0, \lambda]$.
  \item $x \in [1-\lambda, 1+\lambda] \implies \phi_\lambda(x) 
    \in [1-\lambda, 1]$.
  \item $x \le y \implies \phi_\lambda(x) \le \phi_\lambda(y)$.
\end{itemize}
\end{claim}

\begin{lemma}
\label{lem:phi-lambda}
Let $\lambda \in (0, 1/2)$:
\begin{enumerate}[1)]
\item $\forall x: | \phi_\lambda(x) - \phi(x) | \le \lambda$.
\item $\phi_\lambda \in \mathcal{C}_\infty$. Furthermore, for each
  $k \in \mathbb{N}$ there exists a constant $B_k \ge 0$ such that
  $\forall x: \left| \frac{\partial^k}{\partial x^k} \phi_\lambda(x) \right| \le
  \frac{B_k}{\lambda^{k}}$.
\end{enumerate}
\end{lemma}

\begin{proof}
\begin{enumerate}[1)]
\item From Claim \ref{cl:phi-lambda}.
\item Since $\phi_\lambda = \phi \ast \psi_\lambda$, due to Theorem
  \ref{thm:convolution-differentiation} we have
  $\phi_\lambda \in \mathcal{C}^{\infty}$.

  For $x \notin [-\lambda, 1+\lambda]$ the function $\phi_\lambda$
  is constant with
  $\left| \frac{\partial^k}{\partial x^k} \phi_\lambda(x) \right| \le 1$.

  For $x \in [-\lambda, 1+\lambda]$, first note that for every
  $k \in \mathbb{N}$, since $\psi$ has support $[-1, 1]$, also all of 
  its derivatives have support $[-1, 1]$ and therefore
  $\left| \frac{\partial^k}{\partial x^k} \psi(x) \right| \le B_k$.
  Together with Theorem \ref{thm:convolution-differentiation}
  this gives (substituting $z := y/\lambda$)
\begin{IEEEeqnarray*}{rCl}
  \left| \frac{\partial^k}{\partial x^k} \phi_\lambda(x) \right| & = & \left|
    \frac{\partial^k}{\partial x^k} \left( \phi \ast \psi_\lambda \right)(x)
  \right| = \left| \int_{-\infty}^{+\infty} \phi(x-y)
    \frac{\partial^k}{\partial y^k} \psi_\lambda(y) \, \mathrm{d}y \right| \\
  & = & 
  \left| \int_{-\lambda}^{\lambda} \phi(x-y)
    \frac{\partial^k}{\partial y^k} \psi_\lambda(y) \, \mathrm{d}y \right|
  \\ & = & \frac{1}{\lambda^{k+1}} \left| \int_{-\lambda}^{\lambda} \phi(x-y)
    \frac{\partial^k}{\partial z^k} \psi(z) \, \mathrm{d}y \right|
    \le \frac{2B_k}{\lambda^{k}} \; ,
\end{IEEEeqnarray*}
as claimed.
\end{enumerate}
\end{proof}

Now we are ready for the approximation of $\chi$:
\begin{definition}
  Let $\lambda \in (0, 1/2)$. Define function
  $\chi_\lambda: \mathbb{R}^\ell \to \mathbb{R}$ as
\begin{align*}
	\chi_\lambda(\overline{x}) := \prod_{j=1}^\ell \phi_\lambda(x^{(j)}) \; .
\end{align*}
\end{definition}

From Lemma \ref{lem:phi-lambda} we easily get:
\begin{corollary}
\label{cor:xi}
Let $\lambda \in (0, 1/2)$. The function $\chi_\lambda$ has the following
properties:
\begin{enumerate}[1)]
\item
  $\forall \overline{x} \in \mathbb{R}^\ell: \left| \chi(\overline{x}) -
    \chi_\lambda(\overline{x}) \right| \le \ell \lambda$.
\item There exists a universal constant $B \ge 0$ such that $\chi_\lambda$ is
$\frac{B}{\lambda^3}$-smooth.
\end{enumerate}
\end{corollary}

After developing the approximation we are ready to prove the invariance
principle for the function $\chi$:

\begin{theorem}
\label{thm:invariance-xi}
Let $(\overline{\underline{X}}, 
\overline{\underline{\mathcal{X}}},
\overline{\underline{\mathcal{G}}})$ be an ensemble collection
for a probability space 
$(\overline{\underline{\Omega}}, \underline{\mathcal{P}})$
with $\min_{x \in \Omega} \pi(x) \ge \alpha > 0$.

Let $\overline{P} = (P^{(1)}, \ldots, P^{(\ell)})$ be
such that $P^{(j)}$ is a multilinear polynomial compatible with 
the ensemble sequence $\underline{\mathcal{X}}^{(j)}$.

Let $d \in \mathbb{N}$ and $\tau \in [0, 1]$ and assume that
$\deg(P^{(j)}) \le d$ and $\Var[P^{(j)}] \le 1$ for each $j \in [\ell]$, 
and that $\sum_{j=1}^\ell \Inf_i(P^{(j)}) \le \tau$ for each $i \in [n]$.

There exists a universal constant $C \ge 0$ such that
\begin{align*}
  \left| \EE \left[ 
  \chi(\overline{P}(\overline{\underline{\mathcal{X}}}))
  -\chi(\overline{P}(\overline{\underline{\mathcal{G}}}))
  \right] \right|
  \le
  C \cdot \frac{\ell^{5/2} \tau^{1/8}}{\alpha^{4d}} \; .
\end{align*}
\end{theorem}

\begin{proof}
Let $\lambda := \tau^{1/8}/3$. By the triangle inequality we get
\begin{IEEEeqnarray*}{rCl}
\left| \EE \left[ 
	\chi(\overline{P}(\underline{\overline{\mathcal{X}}}))
	- \chi(\overline{P}(\underline{\overline{\mathcal{G}}}))
	\right] \right|
& \le & 
\left| \EE \left[ 
	\chi(\overline{P}(\underline{\overline{\mathcal{X}}}))
	- \chi_\lambda(\overline{P}(\underline{\overline{\mathcal{X}}}))
	\right] \right|
\\ & & +
\left| \EE \left[ 
	\chi_\lambda(\overline{P}(\underline{\overline{\mathcal{X}}}))
	- \chi_\lambda(\overline{P}(\underline{\overline{\mathcal{G}}}))
	\right] \right|
\\ & & +
\left| \EE \left[ 
	\chi_\lambda(\overline{P}(\underline{\overline{\mathcal{G}}}))
	- \chi(\overline{P}(\underline{\overline{\mathcal{G}}}))
	\right] \right| \; .
      \IEEEyesnumber \label{eq:66a}
\end{IEEEeqnarray*}
From Corollary \ref{cor:xi}.1 and the definition of $\lambda$ 
we get both
\begin{align}\label{eq:20a}
\left| \EE \left[ 
	\chi(\overline{P}(\underline{\overline{\mathcal{X}}}))
	- \chi_\lambda(\overline{P}(\underline{\overline{\mathcal{X}}}))
	\right] \right|
&\leq
\ell\lambda \leq O\left(\tfrac{\ell^{5/2}\tau^{1/8}}{\alpha^{4d}}\right)\\
\label{eq:21a}
\left| \EE \left[ 
	\chi_\lambda(\overline{P}(\underline{\overline{\mathcal{G}}}))
	- \chi(\overline{P}(\underline{\overline{\mathcal{G}}}))
	\right] \right|
&\leq
\ell\lambda \leq O\left(\tfrac{\ell^{5/2}\tau^{1/8}}{\alpha^{4d}}\right)\;.
\end{align}

By Theorem \ref{thm:invariance-main} and Corollary \ref{cor:xi}.2 we get
\begin{align}\label{eq:23a}
\left| \EE \left[ 
  \chi_\lambda(\overline{P}(\underline{\overline{\mathcal{X}}}))
  - \chi_\lambda(\overline{P}(\underline{\overline{\mathcal{G}}}))
	\right] \right|
 \le O\left( \frac{\ell^{5/2}d 8^d \tau^{1/2}}{\lambda^3 \alpha^{d/2}} \right) \; .
\end{align}

We can assume w.l.o.g.~that $\alpha \le 1/2$ (otherwise the theorem
is trivial).
Using the definition of $\lambda$, $d8^d \le 9^{d+1}$
and $9 \le \left( \frac{1}{\alpha} \right)^{3.5}$ we see that
\begin{align}\label{eq:22a}
\frac{\ell^{5/2} d 8^d \tau^{1/2}}{\lambda^3 \alpha^{d/2}} \le 
O \left(\frac{\ell^{5/2} d 8^d \tau^{1/8}}{\alpha^{d/2}} \right) \le 
O \left( \frac{\ell^{5/2} \tau^{1/8}}{\alpha^{4d}} \right) \; .
\end{align}
Inserting (\ref{eq:20a}), (\ref{eq:21a}), and the combination
of (\ref{eq:22a}) and~(\ref{eq:23a}) into (\ref{eq:66a}) gives the result.
\end{proof}

\subsubsection{Invariance principle for 
\texorpdfstring{$\gamma$}{gamma}-decaying polynomials}

Let $\overline{P} = (P^{(1)}, \ldots, P^{(\ell)})$ be a tuple of mutlilinear
polynomials and let $\overline{P}^{<d} := \Big(
\left( P^{(1)}  \right)^{<d}, \allowbreak \ldots, \allowbreak 
\left( P^{(\ell)} \right)^{<d}\Big)$.
We will deal with a $\gamma$-decaying $\overline{P}$
by estimating 
$| \EE[\chi(\overline{P}^{<d}(\overline{\underline{\mathcal{X}}})) 
-\chi(\overline{P}(\overline{\underline{\mathcal{X}}}))] |$ for appropriately
chosen $d$.

First, we need a bound on the change of $\chi$:

\begin{lemma}
\label{lem:xi-change}
For all $\overline{x} = (x^{(1)}, \ldots, x^{(\ell)}),
\overline{\epsilon} = (\epsilon^{(1)}, \ldots, \epsilon^{(\ell)})
\in \mathbb{R}^\ell$:
\begin{align*}
\left| \chi(x^{(1)}+\epsilon^{(1)}, \ldots, x^{(\ell)}+\epsilon^{(\ell)}) - 
\chi(x^{(1)}, \ldots, x^{(\ell)}) \right| \le 
\sum_{j=1}^\ell |\epsilon^{(j)}| \; .
\end{align*}
\end{lemma}

\begin{proof}
Letting $\overline{y}_{(j)} := (x^{(1)}, \ldots, x^{(j)}, x^{(j+1)}+\epsilon^{(j+1)},
\ldots x^{(\ell)}+\epsilon^{(\ell)})$,
\begin{IEEEeqnarray*}{l}
  \left| \chi(x^{(1)}+\epsilon^{(1)}, \ldots, x^{(\ell)}+\epsilon^{(\ell)}) -
    \chi(x^{(1)}, \ldots, x^{(\ell)}) \right| \\
  \qquad \le \sum_{j=1}^\ell \left| \chi(\overline{y}_{(j-1)}) -
    \chi(\overline{y}_{(j)}) \right| \le \sum_{j=1}^\ell \left| \epsilon^{(j)}
  \right| \; .
\end{IEEEeqnarray*}
\end{proof}

\begin{proof}[Proof of Theorem \ref{thm:invariance-smoothed}]
Let $d := \lfloor \frac{\ln 1 / \tau}{64 \ln 1/\alpha} \rfloor$. 
By the triangle inequality,
\begin{IEEEeqnarray*}{rCl}
  \left| \EE \left[ \chi(\overline{P}(\underline{\overline{\mathcal{X}}})) -
      \chi(\overline{P}(\underline{\overline{\mathcal{G}}})) \right] \right|
  & \le & \left| \EE \left[
      \chi(\overline{P}(\underline{\overline{\mathcal{X}}})) -
      \chi(\overline{P}^{<d}(\underline{\overline{\mathcal{X}}})) \right] \right|
  \\ & & + \left| \EE \left[
      \chi(\overline{P}^{<d}(\underline{\overline{\mathcal{X}}})) -
      \chi(\overline{P}^{<d}(\underline{\overline{\mathcal{G}}})) \right] \right|
  \\ & & + 
  \left| \EE \left[ \chi(\overline{P}^{<d}(\underline{\overline{\mathcal{G}}})) -
      \chi(\overline{P}(\underline{\overline{\mathcal{G}}})) \right] \right|
    \; .  \IEEEyesnumber \label{eq:67a}
\end{IEEEeqnarray*}

We proceed to demonstrate that all three terms on the right hand side of
(\ref{eq:67a}) are
$O\left( \ell^4 \tau^{\Omega \left( \frac{\gamma}{\ln 1/\alpha} \right)} \right)$,
which will finish the proof.

\begin{lemma}
\begin{align}
\label{eq:24a}
\left|  \EE\left[
  \chi(\overline{P}(\underline{\overline{\mathcal{X}}})) -
  \chi(\overline{P}^{<d}(\underline{\overline{\mathcal{X}}}))
\right] \right|
\le \ell (1-\gamma)^{d/2}
\le O\left(\ell \tau^{\Omega \left( \frac{\gamma}{\ln 1/\alpha} \right)} \right)
\end{align}
and, similarly,
\begin{align}
\label{eq:25a}
\left|  \EE\left[
  \chi(\overline{P}(\underline{\overline{\mathcal{G}}})) -
  \chi(\overline{P}^{<d}(\underline{\overline{\mathcal{G}}}))
\right] \right|
\le \ell (1-\gamma)^{d/2}
\le O\left(\ell \tau^{\Omega \left( \frac{\gamma}{\ln 1/\alpha} \right)} \right)
\end{align}
\end{lemma}

\begin{proof}

We prove only (\ref{eq:24a}), the argument for (\ref{eq:25a}) being the same.
Using Lemma \ref{lem:xi-change}, Cauchy-Schwarz, the fact that $\overline{P}$
is $\gamma$-decaying
and the definition of $d$,
\begin{IEEEeqnarray*}{l}
\left|  \EE\left[
  \chi(\overline{P}(\underline{\overline{\mathcal{X}}})) -
  \chi(\overline{P}^{<d}(\underline{\overline{\mathcal{X}}}))
\right] \right|
\le
\sum_{j=1}^\ell \EE \left[ \left|
	\left(P^{(j)}\right)^{\ge d}(\underline{\mathcal{X}}^{(j)}) \right| \right]
    \\ \qquad \le \sum_{j=1}^\ell
    \sqrt{\EE \left[ \left(\left( P^{(j)}\right)^{\ge d} \right)^2 \right]}
\le \ell (1-\gamma)^{d/2} \le 2 \ell \tau^{\frac{\gamma}{128 \ln 1/\alpha}} \; .
\end{IEEEeqnarray*}
\end{proof}

\begin{lemma}
\begin{align*}
\left| \EE \left[
      \chi(\overline{P}^{<d}(\underline{\overline{\mathcal{X}}})) -
      \chi(\overline{P}^{<d}(\underline{\overline{\mathcal{G}}}))
\right] \right|
\le
O\left( \ell^{5/2} \tau^{\Omega \left( \frac{\gamma}{\ln 1/\alpha} \right)} \right) \; .
\end{align*}
\end{lemma}

\begin{proof}
From Theorem \ref{thm:invariance-xi},
\begin{align*}
\left| \EE \left[
      \chi(\overline{P}^{<d}(\underline{\overline{\mathcal{X}}})) -
      \chi(\overline{P}^{<d}(\underline{\overline{\mathcal{G}}}))
\right] \right|
\le
O\left(\frac{\ell^{5/2} \tau^{1/8}}{\alpha^{4d}} \right) \; .
\end{align*}
From the definition of $d$ (recall that $\alpha \le 1/2$),
\begin{align*}
  \frac{\ell^{5/2} \tau^{1/8}}{\alpha^{4d}} \le \ell^{5/2} \tau^{1/16} \le 
  \ell^{5/2} \tau^{\Omega\left( \frac{\gamma}{\ln 1/\alpha} \right)} \; ,
\end{align*}
as claimed.
\end{proof}
This finishes the proof of Theorem~\ref{thm:invariance-smoothed}.
\renewcommand{\qedsymbol}{}
\end{proof}

\subsection{Reduction to the 
\texorpdfstring{$\gamma$}{gamma}-decaying case}

To apply Theorem \ref{thm:invariance-smoothed} we
need to show that ``smoothing out''  of multilinear polynomials 
$P^{(1)}, \ldots, P^{(\ell)}$
does not change the expectation of their product too much.

Recall Definitions \ref{def:t-rho} and \ref{def:t-rho-function} 
for the operator $T_\rho$. Our goal in this section is to prove:
\begin{theorem}
\label{thm:smoothing}
Let $\overline{\underline{X}}$ be a random vector distributed according
to $(\overline{\underline{\Omega}}, \underline{\mathcal{P}})$
with $\rho(\overline{\Omega}, \mathcal{P}) \le \rho \le 1$.
Let $\underline{\mathcal{Z}}$ be an ensemble sequence constructed from
$\overline{\underline{X}}$ and 
$\underline{\mathcal{X}}^{(1)}, \ldots, \underline{\mathcal{X}}^{(\ell)}$
be ensemble sequences constructed from 
$\underline{X}^{(1)}, \ldots, \underline{X}^{(\ell)}$, respectively.

Let $\epsilon \in (0, 1/2]$ and 
$\gamma \in \left[0, \frac{(1-\rho)\epsilon}{\ell \ln \ell/\epsilon}\right]$.

Then, for all multilinear polynomials $P^{(1)}, \ldots, P^{(\ell)}$ such that
$P^{(j)}(\underline{\mathcal{X}}^{(j)}) \in [0, 1]$:
\begin{align*}
  \left| \EE \left[ \prod_{j=1}^\ell P^{(j)}(\underline{\mathcal{X}}^{(j)})
  - \prod_{j=1}^\ell T_{1-\gamma}P^{(j)}(\underline{\mathcal{X}}^{(j)})
  \right] \right| \le \epsilon \; .
\end{align*}
\end{theorem}

Let us start with an intuition:
Due to Lemma \ref{lem:orthogonal-decomposition},
it is enough to bound
\begin{align*}
\EE\left[ \prod_{j=1}^\ell P_S^{(j)} - \prod_{j=1}^\ell T_{1-\gamma}P^{(j)}_S \right]
\end{align*}
for every $S \subseteq [n]$. If $|S|$ is small, we use the fact that
$P^{(j)}_S-T_{1-\gamma}P_S^{(j)}$ shrinks by a factor of $1-(1-\gamma)^{|S|}$
for every $j$.
If $|S|$ is large, we exploit that both
\begin{align*}
\EE\left[ \prod_{j=1}^\ell P^{(j)}_S\right], 
\EE\left[\prod_{j=1}^\ell T_{1-\gamma}P_S^{(j)}\right]
\end{align*}
are small (roughly $\rho^{|S|}$ times smaller compared to their variances).

To give a formal argument, we use yet another ensemble sequence:
let $j \in [\ell]$. We define  $\underline{\mathcal{Y}}^{(j)}$ to be an
ensemble sequence constructed from $\underline{X}^{[\ell]\setminus\{j\}}$.
Furthermore, let
\begin{align*}
A^{(j)} := \prod_{j' < j} T_{1-\gamma}P(\underline{\mathcal{X}}^{(j')})
\prod_{j' > j} P(\underline{\mathcal{X}}^{(j')}) \; .
\end{align*}
Note that since $A^{(j)} \in L^2(\underline{X}^{[\ell]\setminus\{j\}})$,
there exists a multilinear polynomial $Q^{(j)}$ compatible with
$\underline{\mathcal{Y}}^{(j)}$ such that
\begin{align*}
  A^{(j)} = Q^{(j)}(\underline{\mathcal{Y}}^{(j)}) \; .
\end{align*}

\begin{lemma}
\label{lem:smooth-step-decomposition}
\begin{align*}
  \prod_{j=1}^\ell P^{(j)}(\underline{\mathcal{X}}^{(j)})
  - \prod_{j=1}^\ell T_{1-\gamma}P^{(j)}(\underline{\mathcal{X}}^{(j)})
  = \sum_{j=1}^{\ell}
  (\Id-T_{1-\gamma})P^{(j)}(\underline{\mathcal{X}}^{(j)}) \cdot
  Q^{(j)}(\underline{\mathcal{Y}}^{(j)})
  \; .
\end{align*}
\end{lemma}

\begin{proof}
By definition of $Q^{(j)}$.
\end{proof}

\begin{lemma}
\label{lem:decomposition-rho-bound}
For every $j \in [\ell]$ and $S \subseteq [n]$, $S \ne \emptyset$:
\begin{align*}
\left|\EE\left[ P_S^{(j)}(\underline{\mathcal{X}}^{(j)}) \cdot
Q_S^{(j)}(\underline{\mathcal{Y}}^{(j)}) \right] \right| \le 
\rho^{|S|} \sqrt{\Var[P_S^{(j)}] \Var[Q_S^{(j)}]}  \; .
\end{align*}
\end{lemma}

\begin{proof}
For ease of notation let us write $P := P^{(j)}$,
$Q := Q^{(j)}$, $\underline{\mathcal{X}} := \underline{\mathcal{X}}^{(j)}$
and $\underline{\mathcal{Y}} := \underline{\mathcal{Y}}^{(j)}$.

Let $P(\underline{\mathcal{X}}) = 
\sum_{\sigma} \alpha(\sigma) \mathcal{X}_\sigma$
and $Q(\underline{\mathcal{Y}}) = 
\sum_{\sigma} \beta(\sigma) \mathcal{Y}_\sigma$.

We know that
$\mathcal{X}_{i,k} \in L^2(X^{(j)}_i)$
and $\mathcal{Y}_{i,k} \in L^2(X^{([\ell]\setminus\{j\})}_i)$
for every $i \in [n]$, $k, k' \ge 0$. Furthermore, if $k, k' > 0$,
then $\EE[\mathcal{X}_{i,k}] = \EE[\mathcal{Y}_{i,k'}] = 0$
and $\Var[\mathcal{X}_{i,k}] = \Var[\mathcal{Y}_{i,k'}] = 1$.
By definition of $\rho$, this implies
\begin{align}
\label{eq:52a}
\left| \EE \left[ \mathcal{X}_{i,k} \cdot \mathcal{Y}_{i,k'} 
\right] \right| = 
\left| \Cov \left[ \mathcal{X}_{i,k}, \mathcal{Y}_{i,k'} 
\right] \right| \le \rho.
\end{align}

Expanding the expectation and using (\ref{eq:52a}) and Cauchy-Schwarz,
\begin{IEEEeqnarray*}{rCl}
	\left| \EE \left[ P_S(\underline{\mathcal{X}}) 
            Q_S(\underline{\mathcal{Y}}) \right] \right|
	& = &
	\left| \EE \left[ \left( \sum_{\sigma: \supp(\sigma)=S}  
              \alpha(\sigma) \mathcal{X}_\sigma \right)
		\left( \sum_{\sigma': \supp(\sigma')=S} 
                  \beta(\sigma') \mathcal{Y}_{\sigma'} \right)
		\right] \right|
	\\ & \le &
	\sum_{\substack{\sigma, \sigma': \\ \supp(\sigma)=\supp(\sigma')=S}}
		\left| \alpha(\sigma) \beta(\sigma') \prod_{i \in S}
		\EE \left[ \mathcal{X}_{i, \sigma_i}
                  \mathcal{Y}_{i, \sigma'_i}  \right] \right|
	\\ & \le &
	\rho^{|S|} \sum_{\substack{\sigma, \sigma': 
            \\ \supp(\sigma)=\supp(\sigma')=S}}
	| \alpha(\sigma) \beta(\sigma') |
	\\ & \le &
	\rho^{|S|} \sqrt{\Var[P_S] \Var[Q_S]} \; ,
\end{IEEEeqnarray*}
\end{proof}

\begin{lemma}
\label{lem:k-epsilon-bound}
Let $k \in \mathbb{N}$. Then, $\min(1-(1-\gamma)^k, \rho^k) \le \epsilon/\ell$. 
\end{lemma}

\begin{proof}
If $\rho \in \{0,1\}$ we are done, therefore assume that $\rho \in (0, 1)$.
If $k \ge \log_{\rho} \epsilon/\ell$, then $\rho^k \le \epsilon / \ell$.

If $0 \le k < \log_{\rho} \epsilon/\ell$, then by Bernoulli's inequality,
\begin{align*}
	1-(1-\gamma)^k \le \gamma k \le \frac{1-\rho}{\ln(1/\rho)} \cdot
  \frac{\epsilon}{\ell} \le \frac{\epsilon}{\ell} \; .
\end{align*}
\end{proof}

\begin{lemma}
\label{lem:smooth-subset}
For every $j \in [\ell]$ and $S \subseteq [n]$, $S \ne \emptyset$:
\begin{align*}
\left| \EE\left[
(\Id-T_{1-\gamma})P_S^{(j)}(\underline{\mathcal{X}}^{(j)}) \cdot
Q_S^{(j)}(\underline{\mathcal{Y}}^{(j)}) 
\right] \right|   \le \frac{\epsilon}{\ell} \cdot 
\sqrt{\Var[P_S^{(j)}] \Var[Q_S^{(j)}]} \; . 
\end{align*}
\end{lemma}

\begin{proof}
As in the proof of Lemma \ref{lem:decomposition-rho-bound},
we will write $P := P^{(j)}$,
$Q := Q^{(j)}$, $\underline{\mathcal{X}} := \underline{\mathcal{X}}^{(j)}$
and $\underline{\mathcal{Y}} := \underline{\mathcal{Y}}^{(j)}$.

By definition of $T_{1-\gamma}$,
\begin{align}
\label{eq:51a}
(\Id-T_{1-\gamma})P_S(\underline{\mathcal{X}}) =
(1-(1-\gamma)^{|S|})P_S(\underline{\mathcal{X}}) \; .
\end{align}

From (\ref{eq:51a}), Lemma \ref{lem:decomposition-rho-bound}
and Lemma \ref{lem:k-epsilon-bound},
\begin{IEEEeqnarray*}{rCl}
\left| \EE\left[
  (\Id-T_{1-\gamma})P_S(\underline{\mathcal{X}}) \cdot
  Q_S(\underline{\mathcal{Y}})
\right]\right| & \le & 
\min\left( 1-(1-\gamma)^{|S|}, \rho^{|S|} \right) 
\sqrt{\Var[P_S] \Var[Q_S]} \\
& \le & \frac{\epsilon}{\ell} \sqrt{\Var[P_S] \Var[Q_S]} \; .
\end{IEEEeqnarray*}
\end{proof}

\begin{lemma}
\label{lem:smooth-coordinate}
Fix $j \in [\ell]$. Then,
\begin{align*}
\left| \EE\left[
(\Id-T_{1-\gamma})P^{(j)}(\underline{\mathcal{X}}^{(j)}) \cdot
Q^{(j)}(\underline{\mathcal{Y}}^{(j)}) 
\right] \right|   \le \epsilon / \ell \; . 
\end{align*}
\end{lemma}

\begin{proof}
For ease of notation write $P := P^{(j)}$,
$Q := Q^{(j)}$, $\underline{\mathcal{X}} := \underline{\mathcal{X}}^{(j)}$
and $\underline{\mathcal{Y}} := \underline{\mathcal{Y}}^{(j)}$.

Observe that since 
$P(\underline{\mathcal{X}}), Q(\underline{\mathcal{Y}}) \in [0, 1]$,
also $\Var[P], \Var[Q] \le 1$.

From Lemma \ref{lem:orthogonal-decomposition},
Lemma \ref{lem:smooth-subset} and Cauchy-Schwarz,
\begin{IEEEeqnarray*}{rCl}
\label{eq:50a}
\left| \EE \left[ (\Id-T_{1-\gamma})P(\underline{\mathcal{X}}) \cdot
Q(\underline{\mathcal{Y}}) \right| \right]
& \le & \sum_{S \subseteq [n]} \left| \EE \left[
(\Id-T_{1-\gamma})P_S(\underline{\mathcal{X}}) \cdot
Q_S(\underline{\mathcal{Y}})
\right] \right| \\
& \le &
\frac{\epsilon}{\ell} \sum_{S \ne \emptyset} \sqrt{\Var[P_S] \Var[Q_S]} \\
& \le &
\frac{\epsilon}{\ell} \sqrt{\Var[P] \Var[Q]} \le \epsilon/\ell \; . 
\end{IEEEeqnarray*}
\end{proof}

\begin{proof}[Proof of Theorem \ref{thm:smoothing}]
By Lemma \ref{lem:smooth-step-decomposition} and Lemma
\ref{lem:smooth-coordinate},
\begin{IEEEeqnarray*}{rCl}
\left| \EE \left[
\prod_{j=1}^\ell P^{(j)}(\underline{\mathcal{X}}^{(j)})
- \prod_{j=1}^\ell T_{1-\gamma}P^{(j)}(\underline{\mathcal{X}}^{(j)})
\right] \right|
& \le &
\sum_{j=1}^\ell \left| \EE \left[
(\Id-T_{1-\gamma})P^{(j)}(\underline{\mathcal{X}}^{(j)}) \cdot
Q^{(j)}(\underline{\mathcal{Y}}^{(j)})
\right] \right| \\
& \le & \epsilon \; .
\end{IEEEeqnarray*}
\end{proof}

\subsection{Gaussian reverse hypercontractivity}
\label{sec:gaussian-hyper}

\begin{definition}
Let $L^2(\mathbb{R}^n, \gamma^n)$ be the inner product space
of functions with standard $\mathcal{N}(0, 1)$ Gaussian measure.
\end{definition}

Our goal in this section is to prove the following bound:
\begin{theorem}
\label{thm:gaussian-hypercontractivity-main}
Let $(\overline{\underline{X}}, \underline{\overline{\mathcal{X}}},
\underline{\overline{\mathcal{G}}})$ 
be an ensemble collection
for a probability space 
$(\underline{\overline{\Omega}}, \underline{\mathcal{P}})$
with $\rho(\mathcal{P}) \le \rho < 1$
and such that each orthonormal ensemble in $\overline{\underline{\mathcal{G}}}$
has size $p$.

Then, for all $f^{(1)}, \ldots, f^{(\ell)} \in L^2(\mathbb{R}^{pn}, \gamma^{pn})$
such that $f^{(1)}, \ldots, f^{(\ell)}: \mathbb{R}^{pn} \to [0, 1]$ and
$\EE\left[ f^{(j)}(\underline{\mathcal{G}}^{(j)}) \right] = \mu^{(j)}$:
\begin{align*}
  \EE \left[ \prod_{j=1}^\ell f^{(j)}(\underline{\mathcal{G}}^{(j)}) \right]
  \ge \left( \prod_{j=1}^\ell \mu^{(j)}  \right)^{\ell/(1-\rho^2)} \; .
\end{align*}
\end{theorem}

\begin{remark}
Since the random variables $\mathcal{G}^{(j)}_{i,0}$ are constant,
it suffices to consider consider $f^{(j)}$ as functions of $pn$ rather 
than $(p+1)n$ inputs. 
\end{remark}

In order to prove Theorem \ref{thm:gaussian-hypercontractivity-main},
we will use a multidimensional version of Gaussian reverse hypercontractivity
stated as Theorem 1 in \cite{CDP13} (cf.~also Corollary 4 in \cite{Led14}).

\begin{theorem}[\cite{CDP13}]\label{thm:gaussian-cdp} Let $p > 0$ and let
$\overline{\underline{G}} = (\underline{G}^{(1)}, \ldots, \underline{G}^{(\ell)})$
be a jointly Gaussian collection of $\ell$ random vectors such that:
\begin{itemize}
  \item
    For each $j \in [\ell]$, 
    $\underline{G}^{(j)} = (G_1^{(j)}, \ldots, G_n^{(j)})$
    is a random vector distributed as $n$ independent $\mathcal{N}(0, 1)$
    Gaussians.
  \item
    For every collection of real numbers $\{\alpha_{i}^{(j)}\} \in \mathbb{R}$:
    \begin{align}
     \label{eq:46a}
     \Var\left[ \sum_{i,j} \alpha_i^{(j)} \cdot G_i^{(j)} \right]
      \ge p \cdot \sum_{i,j} \left( \alpha_i^{(j)} \right)^2 \; .
    \end{align}
\end{itemize}
Then, for all functions $f^{(1)}, \ldots, f^{(\ell)} \in L^2(\mathbb{R}^n, \gamma^n)$
such that $f^{(1)}, \ldots, f^{(\ell)}: \mathbb{R}^n \to [0, 1]$
and $\EE\left[f^{(j)}(\underline{G}^{(j)})\right] = \mu^{(j)}$:
\begin{align*}
  \EE \left[ \prod_{j=1}^\ell f^{(j)}(\underline{G}^{(j)}) \right]
  \ge \left( \prod_{j=1}^\ell \mu^{(j)} \right)^{1/p} \; .
\end{align*}
\end{theorem}

\begin{remark}
An equivalent formulation of the condition in (\ref{eq:46a}) 
is that the matrix $(T-p \Id)$ is positive semidefinite,
where $T$ is the covariance matrix of $\overline{\underline{G}}$.  
\end{remark}

To reduce Theorem \ref{thm:gaussian-hypercontractivity-main}
to Theorem \ref{thm:gaussian-cdp} we 
first look at a single-coordinate variance bound for 
ensembles from $\overline{\underline{\mathcal{X}}}$. 
Next, we will extend this bound to multiple coordinates
and ensembles from  $\overline{\underline{\mathcal{G}}}$.

\begin{lemma}
\label{lem:gaussian-single-coordinate}
Let $(\overline{\underline{X}}, \overline{\underline{\mathcal{X}}}, 
\overline{\underline{\mathcal{G}}})$
be an ensemble collection
for a probability space 
$(\overline{\underline{\Omega}}, \underline{\mathcal{P}})$
with $\rho(\mathcal{P}) \le \rho < 1$
and such that each orthonormal ensemble in $\overline{\underline{\mathcal{X}}}$
has size $p$.

Fix $i \in [n]$ and for ease of notation
let us write 
$\mathcal{X}^{(j)} = (\mathcal{X}^{(j)}_{0}, \ldots, \mathcal{X}^{(j)}_{p})$
for the random ensemble 
$\mathcal{X}_i^{(j)} = (\mathcal{X}^{(j)}_{i,0}, \ldots, \mathcal{X}^{(j)}_{i,p})$.

Then,
for every collection of real numbers $\{\alpha_{k}^{(j)}\} \in \mathbb{R}$:
\begin{align*}
 \Var\left[ \sum_{j\ge 1, k>0} \alpha_k^{(j)} \cdot \mathcal{X}_k^{(j)} \right]
      \ge \frac{1-\rho^2}{\ell} \cdot \sum_{j \ge 1,k > 0} 
      \left( \alpha_k^{(j)} \right)^2 \; .
\end{align*}
\end{lemma}

\begin{proof}
For any $j \in [\ell]$ we define
$A_j := \sum_{k > 0} \alpha_k^{(j)} \cdot \mathcal{X}_{k}^{(j)}$
and $B_j := \allowbreak \sum_{j'\in [\ell]\setminus\{j\}}\sum_{k > 0} 
\alpha_{k}^{(j')} \cdot \mathcal{X}_{k}^{(j')}$.

We compute
\begin{align*}
\Var[B_j]\cdot \Var[A_j+B_j] 
&=
\Var[A_j]\cdot\Var[B_j] + (\Var[B_j])^2 + 2\Var[B_j]\Cov[A_j,B_j]\\
&=
\Var[A_j]\cdot\Var[B_j] + (\Var[B_j]+\Cov[A_j,B_j])^2 - \Cov[A_j,B_j]^2\\
&\geq
\Var[A_j]\cdot\Var[B_j]  - \Cov[A_j,B_j]^2\\
&\geq \Var[A_j]\Var[B_j](1-\rho^2)\;,
\end{align*}
where in the last inequality we used that the definition of $\rho$ implies
\begin{align*}
\bigl|\Cov[A_j, B_j]\bigr| \le \rho \sqrt{\Var[A_j]\Var[B_j]}
\end{align*}
since $A_j \in L^2(X_i^{(j)})$ and 
$B_i \in L^2(X_i^{([\ell]\setminus \{j\})})$.

Therefore, 
\begin{align*}
 \Var\left[ \sum_{j\ge 1, k>0} \alpha_k^{(j)} \cdot \mathcal{X}_k^{(j)} \right]
&= \frac{1}{\ell} \sum_{j=1}^{\ell} \Var[A_j+B_j] 
\geq \frac{1-\rho^2}{\ell}\sum_{j=1}^\ell \Var[A_j] \\
&= 
\frac{1-\rho^2}{\ell}\sum_{j=1}^\ell \sum_{k > 0}
\bigl(\alpha_j^{(k)}\bigr)^2 \; .  \qedhere
\end{align*}
\end{proof}

\begin{lemma}
\label{lem:gaussian-x}
Let $(\overline{\underline{X}}, 
\overline{\underline{\mathcal{X}}},
\overline{\underline{\mathcal{G}}})$ be an ensemble collection
for a probability space 
$(\overline{\underline{\Omega}}, \underline{\mathcal{P}})$
with $\rho(\mathcal{P}) \le \rho < 1$.

Then,
for every collection of real numbers $\{\alpha_{i,k}^{(j)}\} \in \mathbb{R}$:
\begin{align*}
 \Var\left[ \sum_{i,j\ge1, k>0} \alpha_{i,k}^{(j)} \cdot \mathcal{X}_{i,k}^{(j)} 
  \right]
      \ge \frac{1-\rho^2}{\ell} \cdot \sum_{i,j \ge 1,k > 0} 
      \left( \alpha_{i,j}^{(k)} \right)^2 \; .
\end{align*}
\end{lemma}

\begin{proof}
Since ensembles $\overline{\mathcal{X}}_i$ are independent, by
Lemma \ref{lem:gaussian-single-coordinate},
\begin{IEEEeqnarray*}{rCl}
 \Var\left[ \sum_{i,j\ge1, k>0} \alpha_{i,k}^{(j)} \cdot \mathcal{X}_{i,k}^{(j)} 
  \right] & = & \sum_{i=1}^n \Var\left[
  \sum_{j\ge 1,k>0} \alpha_{i,k}^{(j)} \cdot \mathcal{X}_{i,k}^{(j)} \right]
  \\ & \ge & \frac{1-\rho^2}{\ell} \cdot \sum_{i,j\ge1,k>0} 
  \left(\alpha_{i,k}^{(j)}\right)^2 \; .
\end{IEEEeqnarray*}
\end{proof}

\begin{lemma}
\label{lem:gaussian-g}
Let $(\overline{\underline{X}}, 
\overline{\underline{\mathcal{X}}},
\overline{\underline{\mathcal{G}}})$ be an ensemble collection
for a probability space 
$(\overline{\underline{\Omega}}, \underline{\mathcal{P}})$
with $\rho(\mathcal{P}) \le \rho < 1$.

Then,
for every collection of real numbers $\{\alpha_{i,k}^{(j)}\} \in \mathbb{R}$:
\begin{align*}
 \Var\left[ \sum_{i,j\ge1, k>0} \alpha_{i,k}^{(j)} \cdot \mathcal{G}_{i,k}^{(j)} 
  \right]
      \ge \frac{1-\rho^2}{\ell} \cdot \sum_{i,j \ge 1,k > 0} 
      \left( \alpha_{i,j}^{(k)} \right)^2 \; .
\end{align*}
\end{lemma}

\begin{proof}
By Corollary \ref{lem:cov-xg}  and Lemma \ref{lem:gaussian-x}.
\end{proof}

\begin{proof}[Proof of Theorem~\ref{thm:gaussian-hypercontractivity-main}]
By application of Theorem \ref{thm:gaussian-cdp} to
$\overline{\underline{G}} = 
(\underline{G}^{(1)}, \ldots, \allowbreak \underline{G}^{(\ell)})$,
where $\underline{G}^{(j)} = 
(\mathcal{G}_{i,1}^{(j)}, \allowbreak \ldots, \allowbreak \mathcal{G}_{i,p}^{(j)},
\ldots, \mathcal{G}_{n,1}^{(j)}, \ldots, \mathcal{G}_{n,p}^{(j)})$.

Since $\underline{\mathcal{G}}^{(j)}$ is a Gaussian ensemble sequence,
$\underline{G}^{(j)}$ is distributed as $pn$ independent
$\mathcal{N}(0, 1)$ Gaussians. Condition (\ref{eq:46a}) 
for $p: = \frac{1-\rho^2}{\ell}$ is fulfilled
due to Lemma \ref{lem:gaussian-g}.
\end{proof}

\subsection{The main theorem}

We recall the low-influence theorem that we want to prove:
\lowinfluence*

We need to define some new objects in order to proceed with the proof.
Let $(\overline{\underline{X}}, 
\overline{\underline{\mathcal{X}}},
\overline{\underline{\mathcal{G}}})$ be an ensemble collection for
$(\overline{\underline{\Omega}}, \underline{\mathcal{P}})$.

For $j \in [\ell]$, let $P^{(j)}$ be a multilinear polynomial compatible
with $\underline{\mathcal{X}}^{(j)}$ and equivalent to 
$f^{(j)}(\underline{X}^{(j)})$. For
some small $\gamma > 0$ to be fixed later let 
$Q^{(j)} := T_{1-\gamma}P^{(j)}$. Finally, letting $p$ be the size
of each of the ensembles $\mathcal{X}_i^{(j)}$ and $\mathcal{G}_i^{(j)}$,
define a function $R^{(j)}: \mathbb{R}^{pn} \to \mathbb{R}$ as
\begin{align*}
R^{(j)}(\underline{x}) := \begin{cases}
  0 & \text{if $Q^{(j)}(\underline{x}) < 0$,}\\
  Q^{(j)}(\underline{x}) & \text{if $Q^{(j)}(\underline{x}) \in [0, 1]$,}\\
  1 & \text{if $Q^{(j)}(\underline{x}) > 1$.}
\end{cases}
\end{align*}
Note that it might be impossible to write $R^{(j)}$ as a multilinear 
polynomial, but it will not cause problems in the proof.
Finally, let 
$\mu'^{(j)} := \EE\left[R^{(j)}(\underline{\mathcal{G}}^{(j)})\right]$.

The proof proceeds by decomposing the expression we are bounding into
several parts:
\begin{IEEEeqnarray}{l}
  \EE \left[ \prod_{j=1}^\ell f^{(j)}(\underline{X}^{(j)}) \right]
  = \EE \left[ \prod_{j=1}^\ell P^{(j)}(\underline{\mathcal{X}}^{(j)}) \right] 
  = \nonumber \\
  \quad = \EE \left[ \prod_{j=1}^\ell P^{(j)}(\underline{\mathcal{X}}^{(j)})   
  - \prod_{j=1}^\ell Q^{(j)}(\underline{\mathcal{X}}^{(j)}) \right] +  
  \label{eq:53a} \\
  \qquad + \EE \left[ \prod_{j=1}^\ell Q^{(j)}(\underline{\mathcal{X}}^{(j)})
  - \prod_{j=1}^\ell R^{(j)}(\underline{\mathcal{G}}^{(j)}) \right] + 
  \label{eq:54a} \\
  \qquad + \EE \left[ \prod_{j=1}^\ell R^{(j)}(\mathcal{G}^{(j)}) \right] \; .
  \label{eq:55a}
\end{IEEEeqnarray}
We use the theorems proved so far to bound each of the terms 
(\ref{eq:53a}), (\ref{eq:54a}) and (\ref{eq:55a})
in turn. First, we apply Theorem \ref{thm:smoothing} to show that (\ref{eq:53a})
has small absolute value. Then, we use the invariance principle
(Theorem \ref{thm:invariance-smoothed}) to argue that (\ref{eq:54a}) has
small absolute value. Finally, using Gaussian reverse hypercontractivity
(Theorem \ref{thm:gaussian-hypercontractivity-main}) we show that
(\ref{eq:55a}) is bounded from below by (roughly) 
$\left(\prod_{j=1}^\ell \mu^{(j)}\right)^{\ell/(1-\rho^2)}$.

We proceed with a detailed argument in the following lemmas.
In the following assume w.l.o.g~that $\epsilon \le 1/2$
and $\alpha \le 1/2$.

\begin{lemma}
\label{lem:main-smoothing}
Set $\gamma := \frac{(1-\rho)\epsilon}{2\ell \ln 2\ell/\epsilon}$. Then,
\begin{align*}
  \left| \EE \left[ \prod_{j=1}^\ell P^{(j)}(\underline{\mathcal{X}}^{(j)}) 
  - \prod_{j=1}^{\ell} Q^{(j)}(\underline{\mathcal{X}}^{(j)})
  \right] \right| \le \epsilon / 2 \; .
\end{align*}
\end{lemma}

\begin{proof}
By Theorem \ref{thm:smoothing}.
\end{proof}

\begin{lemma}
\label{lem:main-invariance}
There exists an absolute constant $C > 0$ such that
\begin{align*}
\left| \EE\left[ \prod_{j=1}^\ell Q^{(j)}(\underline{\mathcal{X}}^{(j)})
- \prod_{j=1}^\ell R^{(j)}(\underline{\mathcal{G}}^{(j)}) \right] \right|
\le C \ell^{5/2} \cdot \tau^{\frac{\gamma}{C \ln 1/\alpha}} \; .
\end{align*}
\end{lemma}

\begin{proof}
Note that for every $j \in [\ell]$ the polynomial $Q^{(j)}$ is 
$\gamma$-decaying and that it has bounded influence for every $i \in [n]$:
\begin{align*}
\Inf_i(Q^{(j)}) \le \Inf_i(P^{(j)}) = \Inf_i(f^{(j)}(\underline{X}^{(j)})
\le \tau \; .
\end{align*}

By definition of $\chi$ (Definition \ref{def:xi}) and Theorem
\ref{thm:invariance-smoothed},
\begin{IEEEeqnarray*}{rCl}
\left| \EE \left[ \prod_{j=1}^\ell Q^{(j)}(\underline{\mathcal{X}}^{(j)}) - 
\prod_{j=1}^\ell R^{(j)}(\underline{\mathcal{G}}^{(j)}) \right] \right|
& = & \left| \EE \left[ 
  \chi\left(\overline{Q}(\overline{\underline{\mathcal{X}}})\right) 
  - 
  \chi\left(\overline{Q}(\overline{\underline{\mathcal{G}}})\right)
  \right] \right|
  \\
  & \le & C \ell^{5/2} \cdot \tau^{\frac{\gamma}{C \ln 1/\alpha}} \; .
\end{IEEEeqnarray*}
\end{proof}

\begin{lemma}
\label{lem:main-gaussian}
\begin{align*}
\EE \left[ \prod_{j=1}^\ell R^{(j)}(\underline{\mathcal{G}}^{(j)}) \right]
  \ge \left( \prod_{j=1}^\ell \mu'^{(j)} \right)^{\ell / (1-\rho^2)} \; .
\end{align*}
\end{lemma}

\begin{proof}
By Theorem \ref{thm:gaussian-hypercontractivity-main}.
\end{proof}

Lastly, we need to show that the difference between
$\prod_{j=1}^\ell \mu'^{(j)}$ and
$\prod_{j=1}^\ell \mu^{(j)}$ is small.

\begin{claim}
\label{cl:bound-mu}
Let $a \ge 0, \epsilon \ge 0, a+\epsilon \le 1, \beta \ge 1$.
Then, $(a+\epsilon)^\beta - a^\beta \le \beta\epsilon$.
\end{claim}
\begin{proof}
The function $h_{\beta,\epsilon}(a) := (a+\epsilon)^\beta-a^{\beta}$
is non-decreasing (since 
$\frac{\mathrm{d}}{\mathrm da} h_{\beta,\epsilon} = 
\beta((a+\epsilon)^{\beta-1}-a^{\beta-1}) \geq 0$).
Hence, 
\begin{align*}
(a+\epsilon)^\beta-a^\beta \le 1-(1-\epsilon)^\beta \le \beta \epsilon \; ,
\end{align*}
where in the last step we applied Bernoulli's inequality.
\end{proof}

\begin{lemma}
\label{lem:main-mu}
There exists an absolute constant $C > 0$ such that
\begin{align*}
\left| \left( \prod_{j=1}^\ell \mu^{(j)} \right)^{\ell / (1-\rho^2)}  - 
\left( \prod_{j=1}^\ell \mu'^{(j)} \right)^{\ell / (1-\rho^2)}  \right|
\le \frac{C\ell^2}{1-\rho^2} \cdot \tau^{\frac{\gamma}{C \ln 1/\alpha}} \; .
\end{align*}
\end{lemma}

\begin{proof}
By Claim \ref{cl:bound-mu},
\begin{align}
\label{eq:58a}
\left| \left( \prod_{j=1}^\ell \mu^{(j)} \right)^{\ell / (1-\rho^2)}  - 
\left( \prod_{j=1}^\ell \mu'^{(j)} \right)^{\ell / (1-\rho^2)}  \right|
\le \frac{\ell}{1-\rho^2} \cdot
\left| \prod_{j=1}^\ell \mu^{(j)}  - 
 \prod_{j=1}^\ell \mu'^{(j)}  \right| \; .
\end{align}

Since $\mu^{(j)}, \mu'^{(j)} \in [0, 1]$,
\begin{align}
\label{eq:56a}
\left| \prod_{j=1}^\ell \mu^{(j)} - \prod_{j=1}^\ell \mu'^{(j)} \right|
\le \sum_{j=1}^\ell \left| \mu^{(j)} - \mu'^{(j)} \right| \; .
\end{align}
For a fixed $j \in [\ell]$, from the definition of $\chi$ and Theorem
\ref{thm:invariance-smoothed} applied with $\ell=1$,
\begin{align}
\label{eq:57a}
  \left| \mu^{(j)} - \mu'^{(j)} \right|
  = \left| \EE \left[ \chi\left(Q^{(j)}(\underline{\mathcal{X}}^{(j)})\right)
  - \chi\left(Q^{(j)}(\underline{\mathcal{G}}^{(j)})\right) \right] \right|
  \le C \cdot \tau^{\frac{\gamma}{C \ln 1/\alpha}} \; .
\end{align}
Inequalities (\ref{eq:58a}), (\ref{eq:56a}) and (\ref{eq:57a}) together give the claim.
\end{proof}

\begin{proof}[Proof of Theorem \ref{thm:low-influence}]
Following the decomposition of $\prod_{j=1}^\ell f^{(j)}(\underline{X}^{(j)})$
into subexpressions (\ref{eq:53a}), (\ref{eq:54a}) and (\ref{eq:55a}), from
Lemma \ref{lem:main-smoothing}, Lemma \ref{lem:main-invariance},
Lemma \ref{lem:main-gaussian} and Lemma \ref{lem:main-mu},
\begin{IEEEeqnarray*}{rCl}
\EE \left[ \prod_{j=1}^\ell f^{(j)}(\underline{X}^{(j)}) \right]
& \ge & 
\left(\prod_{j=1}^\ell \mu^{(j)} \right)^{\ell / (1-\rho^2)}  
\!\!\!\!\!\!\! -\epsilon/2
- C\ell^{5/2} \cdot \tau^{\frac{\gamma}{C \ln 1/\alpha}}
- \frac{C\ell^2}{1-\rho^2} \cdot \tau^{\frac{\gamma}{C\ln 1/\alpha}} \\
& \ge & 
\left(\prod_{j=1}^\ell \mu^{(j)} \right)^{\ell / (1-\rho^2)} 
\!\!\!\!\!\!\! - \epsilon/2
- \frac{2C\ell^{5/2}}{1-\rho^2} \cdot \tau^{\frac{\gamma}{C\ln 1/\alpha}} \; .
\end{IEEEeqnarray*}

By choosing $\tau(\epsilon,\rho,\alpha,\ell,\gamma)$
small enough we get 
\begin{align}
\label{eq:59a}
\frac{2C\ell^{5/2}}{1-\rho^2} \cdot \tau^{\frac{\gamma}{C\ln 1/\alpha}} 
\le \epsilon / 2 \; ,
\end{align}
which is the main part of the theorem (recall that 
$\gamma = \frac{(1-\rho)\epsilon}{2\ell \ln(2\ell/\epsilon)}$).

To see that we can choose $\tau$ as in (\ref{eq:38a}), note that
for $D > 0$ big enough we have
\begin{IEEEeqnarray*}{rCl}
\tau &:=& \left( \frac{(1-\rho^2)\epsilon}{\ell^{5/2}}
  \right)^{D \frac{ \ell \ln(\ell/\epsilon)\ln(1/\alpha)}{(1-\rho)\epsilon}}
  \le
  \left( \frac{(1-\rho^2)\epsilon}{\ell^{5/2}}  
  \right)^{D' \frac{2 C \ell \ln(2\ell/\epsilon)\ln(1/\alpha)}
  {(1-\rho)\epsilon}}
\\
&=& 
\left( \frac{(1-\rho^2)\epsilon}{\ell^{5/2}} 
\right)^{D' \frac{C \ln(1/\alpha)}{\gamma}}
\end{IEEEeqnarray*}
for $D' > 0$ as needed.
Hence, we obtain
\begin{align*}
\frac{2C\ell^{5/2}}{1-\rho^2} \cdot \tau^{\frac{\gamma}{C\ln 1/\alpha}}
= 2C \cdot \frac{\ell^{5/2}}{1-\rho^2} \cdot \left( 
\frac{(1-\rho^2)\epsilon}{\ell^{5/2}} \right)^{D'}
\le 2C \epsilon^{D'} \le \epsilon / 2 \; ,
\end{align*}
which establishes (\ref{eq:59a}) for this choice of $\tau$.
\end{proof}


\bibliographystyle{alpha}
\newcommand{\etalchar}[1]{$^{#1}$}


\begin{dajauthors}
\begin{authorinfo}[jh]
  Jan Hązła\\
  Massachusetts Institute of Technology\\
  Cambridge, Massachusetts, USA\\
  jhazla\imageat{}mit\imagedot{}edu \\
  \url{https://idss.mit.edu/staff/jan-hazla/}
\end{authorinfo}
\begin{authorinfo}[th]
  Thomas Holenstein\\
  Google\\
  Zurich, Switzerland\\
  thomas\imagedot{}holenstein\imageat{}google\imagedot{}com \\
\end{authorinfo}
\begin{authorinfo}[em]
  Elchanan Mossel\\
  Massachusetts Institute of Technology\\
  Cambridge, Massachusetts, USA\\
  elmos\imageat{}mit\imagedot{}edu\\
  \url{https://math.mit.edu/~elmos/}
\end{authorinfo}
\end{dajauthors}

\end{document}